\documentclass[11pt]{llncs}

\usepackage{geometry}
\usepackage{amssymb}

\usepackage{amsmath}
\allowdisplaybreaks 
\numberwithin{equation}{section} 

\usepackage{breakcites}
\usepackage{caption}
\usepackage{subcaption}
\usepackage{nicefrac}
\usepackage[para,online,flushleft]{threeparttable}

\usepackage{soul}

\usepackage{ifdraft}

\usepackage{anyfontsize}

\usepackage{color}
\usepackage[dvipsnames]{xcolor}

\usepackage{graphicx} \setkeys{Gin}{draft=false}
\usepackage{mathtools}

\usepackage{booktabs}
\usepackage{array}
\newcolumntype{L}[1]{>{\raggedright\let\newline\\\arraybackslash\hspace{0pt}}m{#1}}
\newcolumntype{C}[1]{>{\centering\let\newline\\\arraybackslash\hspace{0pt}}m{#1}}
\newcolumntype{R}[1]{>{\raggedleft\let\newline\\\arraybackslash\hspace{0pt}}m{#1}}

\usepackage{multirow}

\usepackage{makecell}

\usepackage{booktabs,colortbl}

\usepackage{longtable}

\usepackage{physics}

\usepackage{xspace}
\usepackage[utf8]{inputenc}

\usepackage{amsthm}

\usepackage{cases} 

\usepackage{enumitem}
\setitemize{
topsep=0.5em,
itemsep=0.5em
}
\setenumerate{
topsep=0.5em,
itemsep=0.5em
}

\usepackage{scalerel}

\usepackage{algorithm,algpseudocode}

\usepackage[pagebackref=true,draft=false]{hyperref}
\hypersetup{
  hypertexnames=true, 
	linktoc=all, 
	hidelinks, 
	colorlinks=true,
	allcolors=blue
}

\AddToHook{cmd/appendix/before}{%
}

\ifdraft{
  \usepackage[multiuser,inline,nomargin,marginclue,draft]{fixme}
}{
  \usepackage[multiuser,inline,nomargin,marginclue]{fixme}
}
\fxusetheme{color}
\colorlet{fxtarget}{Red}

\FXRegisterAuthor{x}{ex}{\color{Plum} Xiao}

\FXRegisterAuthor{Ziqing}{EnvZiqing}{\color{Red} Ziqing}

\FXRegisterAuthor{n}{en}{\color{ForestGreen} Fuyuki}

\usepackage{nameref}
\makeatletter
\let\orgdescriptionlabel\descriptionlabel
\renewcommand*{\descriptionlabel}[1]{%
  \let\orglabel\label
  \let\label\@gobble
  \phantomsection
  \edef\@currentlabel{#1}%
  \let\label\orglabel
  \orgdescriptionlabel{#1}%
}
\makeatother

\usepackage[nameinlink]{cleveref}
\Crefname{Algorithm}{Algorithm}{Algorithms}
\Crefname{appendix}{Appx.}{Appx.}
\Crefname{Assumption}{Assumption}{Assumption}
\Crefname{Case}{Case}{Cases}
\Crefname{Condition}{Condition}{Conditions}
\Crefname{definition}{Def.}{Def.}
\Crefname{equation}{Eq.}{Eq.}
\Crefname{Equation}{Eq.}{Eq.}
\Crefname{Execution}{Execution}{Executions}
\Crefname{Expression}{Expression}{Expressions}
\creflabelformat{Expression}{#2(#1)#3}
\Crefname{figure}{Fig.}{Fig.}
\Crefname{Fact}{Fact}{Facts}
\Crefname{Inequality}{Inequality}{Inequalities}
\creflabelformat{Inequality}{#2(#1)#3}
\Crefname{Issue}{Issue}{Issues}
\Crefname{lemma}{Lem.}{Lem.}
\Crefname{Language}{Language}{Languages}
\creflabelformat{Language}{#2(#1)#3}
\Crefname{obs}{Observation}{Observations}
\Crefname{Path}{Path}{Paths}
\Crefname{Problem}{Problem}{Problems}
\Crefname{Property}{Property}{Properties}
\Crefname{remark}{Rmk.}{Rmk.}
\Crefname{Round}{Round}{Round}
\Crefname{Reason}{Reason}{Reasons}
\Crefname{section}{Sec.}{Sec.}
\Crefname{Stage}{Stage}{Stages}
\Crefname{Step}{Step}{Steps}
\Crefname{theorem}{Thm.}{Thm.}
\Crefname{Trick}{Trick}{Tricks}
\Crefname{Type}{Type}{Types}

\usepackage[many]{tcolorbox}
\newtcolorbox
[auto counter,
Crefname={Constr.}{Constr.},
]{ConstructionBox}[2][]
{
  title    = Construction~\thetcbcounter: #2,#1,
  fonttitle=\small\bfseries,
  enforce breakable, 
  colframe = black,
  coltitle = black,
  colbacktitle = white, 
  coltext  = black, 
  colback  = white,
  arc=0pt,
  outer arc=0pt,
  boxrule=0.5pt,
  left=0pt, 
  right=0pt, 
   top=0pt,
  bottom=0pt,
}

\newtcolorbox
[auto counter,
Crefname={Prot.}{Prot.}
]{ProtocolBox}[2][]
{
  title    = Protocol~\thetcbcounter: #2,#1,
  fonttitle=\small\bfseries,
  enforce breakable, 
  colframe = black,
  coltitle = black,
  colbacktitle = white, 
  coltext  = black, 
  colback  = white,
  arc=0pt,
  outer arc=0pt,
  boxrule=0.5pt,
  toprule at break = 0pt,
  bottomrule at break = 0pt,
  left=0pt, 
  right=0pt, 
   top=0pt,
  bottom=0pt,
}

\newtcolorbox
[auto counter,
Crefname={Fig.}{Fig.},
]{FigureBox}[2][]
{
  title    = Figure~\thetcbcounter: #2,#1,
  fonttitle=\small\bfseries,
  enforce breakable, 
  colframe = black,
  coltitle = black,
  colbacktitle = white, 
  coltext  = black, 
  colback  = white,
  arc=0pt,
  outer arc=0pt,
  boxrule=0.5pt,
  toprule at break = 0pt,
  bottomrule at break = 0pt,
  left=0pt, 
  right=0pt, 
   top=0pt,
  bottom=0pt,
}

\newtcolorbox[auto counter]{CommentBox}[2][]
{
  breakable,
  colframe = gray!10,
  colback  = gray!2,
  coltitle = red, 
  coltext  = red, 
  title    = #2,
  arc=0pt,
  outer arc=0pt,
  boxrule=1.5pt,
  left=0mm,
  right=0pt,
  top=0pt,
  bottom=0pt,
}

\newtcolorbox[
auto counter,
number within=section, 
Crefname={Algo.}{Algo.}
]{AlgorithmBox}[2][]
{
  title    = Algorithm~\thetcbcounter: #2,#1,
  fonttitle=\small\bfseries,
  enforce breakable, 
  colframe = black,
  coltitle = black,
  colbacktitle = white, 
  coltext  = black, 
  colback  = white,
  arc = 0pt,
  outer arc = 0pt,
  boxrule = 0.5pt,
  toprule at break = 0pt,
  bottomrule at break = 0pt,
  left = 0pt, 
  right = 0pt, 
  top=0pt,
  bottom=0pt
}

\newtcolorbox
[auto counter,
number within=section,
Crefname={Expr.}{Expr.}
]{ExperimentBox}[2][]
{
  title    = Experiment~\thetcbcounter: #2,#1,
  fonttitle=\small\bfseries,
  enforce breakable, 
  colframe = black,
  coltitle = black,
  colbacktitle = white, 
  coltext  = black, 
  colback  = white,
  arc=0pt,
  outer arc=0pt,
  boxrule=0.5pt,
  toprule at break = 0pt,
  bottomrule at break = 0pt,
  left=0pt, 
  right=0pt, 
   top=0pt,
  bottom=0pt,
}

\newtcolorbox
[auto counter,
number within=section,
Crefname={Game}{Game}
]{GameBox}[2][]
{
  title    = Game~\thetcbcounter: #2,#1,
  fonttitle=\small\bfseries,
  enforce breakable, 
  colframe = black,
  coltitle = black,
  colbacktitle = white, 
  coltext  = black, 
  colback  = white,
  arc=0pt,
  outer arc=0pt,
  boxrule=0.5pt,
  toprule at break = 0pt,
  bottomrule at break = 0pt,
  left=0pt, 
  right=0pt, 
   top=0pt,
  bottom=0pt,
}

\spnewtheorem{construction}{Construction}{\bfseries}{\itshape}
\crefname{construction}{construction}{constructions}
\Crefname{construction}{Construction}{Constructions}

\spnewtheorem*{MyComment}{Comment}{\bfseries}{\itshape}

\spnewtheorem{MyClaim}[theorem]{Claim}{\bfseries}{\itshape}
\Crefname{MyClaim}{Claim}{Claims}
\creflabelformat{Claim}{#2(#1)#3}

\makeatletter
\renewcommand\subsubsection{\@startsection{subsubsection}{3}{\z@}%
                       {-18\p@ \@plus -4\p@ \@minus -4\p@}%
                       {0.5em \@plus 0.22em \@minus 0.1em}%
                       {\normalfont\normalsize\bfseries\boldmath}}

\newcommand\subsubsubsection{\@startsection{paragraph}{4}{\z@}
{-2.5ex\@plus -1ex \@minus -.25ex}{1.25ex \@plus .25ex}{\normalfont\normalsize\bfseries}}
\makeatother

\Crefname{subsection}{Sec.}{Sec.}
\Crefname{APPXsubsubsection}{Appx.}{Appx.}

\newcommand{\para}[1]{\vspace{1em}\noindent {\bf #1}}
\newcommand{\subpara}[1]{\vspace{1em}\noindent\underline{#1}}

\usepackage[normalem]{ulem}
\newcommand{\dlt}{\bgroup\markoverwith{\textcolor{red}{\rule[0.5ex]{1pt}{0.9pt}}}\ULon}

\usepackage{tikz}
\usepackage{etoolbox}
\newcommand{\circledtwodigital}[2][]{%
  \tikz[baseline=(char.base)]{%
    \node[shape = circle, draw, inner sep = 0pt]
    (char) {\phantom{\ifblank{#1}{#2}{#1}}};%
    \node at (char.center) {\makebox[0pt][c]{#2}};}}
\robustify{\circledtwodigital}

\newcommand{\msf}{\mathsf}
\newcommand{\mcal}{\mathcal}
\newcommand{\mrm}{\mathrm}
\newcommand{\mbf}{\mathbf}

\newcommand{\Adv}{\ensuremath{\mathcal{A}}\xspace}
\newcommand{\A}{\ensuremath{\mathcal{A}}\xspace}

\renewcommand{\bar}{\overline}

\newcommand{\bits}{\ensuremath{\{0,1\}}\xspace}

\newcommand{\calX}{\ensuremath{\mathcal{X}}}

\newcommand{\cind}{\ensuremath{\stackrel{\text{c}}{\approx}}\xspace}

\renewcommand{\epsilon}{\varepsilon}

\newcommand{\Naturals}{\ensuremath{\mathbb{N}}\xspace}

\newcommand{\negl}{\ensuremath{\mathsf{negl}}\xspace}

\newcommand{\NP}{\ensuremath{\mathbf{NP}}\xspace}

\renewcommand{\paragraph}{\para}

\newcommand{\poly}{\ensuremath{\mathsf{poly}}\xspace}

\newcommand{\ra}{\ensuremath{\rightarrow}\xspace}

\newcommand{\reg}{\mathtt}

\newcommand{\sample}{\gets}

\newcommand{\secpar}{\ensuremath{\lambda}\xspace}

\newcommand{\Set}[1]{\ensuremath{\{#1\}}\xspace}

\newcommand{\Sim}{\ensuremath{\mathcal{S}}\xspace}

\renewcommand{\tilde}{\widetilde}

\newcommand{\xor}{\ensuremath{\oplus}\xspace}

\usepackage{bm}
\usepackage{mdframed}

\makeatletter
\def\@fnsymbol#1{\ensuremath{\ifcase#1\or *\or \dagger\or \ddagger\or
   \mathsection\or \mathparagraph\or \|\or **\or \dagger\dagger
   \or \ddagger\ddagger \else\@ctrerr\fi}}
\makeatother

\begin{document}

\ifdraft{
}{}

\ifdraft{
    \pagenumbering{gobble} 
    \listoffixmes
    \addcontentsline{toc}{section}{List of Corrections}
    \clearpage
}{}

\title{I Prove, Therefore I Am: \\ 
Spatiotemporal Multi-Party Computation}

\author{
Ziqing Guo\inst{1}
	 \and
Fuyuki Kitagawa\inst{2}
 \and
 Xiao Liang\inst{3}\thanks{The work described in this paper was partially supported by a 
grant from the Research Grants Council of the Hong Kong Special Administrative 
Region, China [Project No.: CUHK 
24208025].}
 }
 \institute{
 The Chinese University of Hong Kong, Hong Kong\\ \email{ziqing-guo@link.cuhk.edu.hk}
  \and
  NTT Social Informatics Laboratories, Japan\\ \email{fuyuki.kitagawa@ntt.com}
 \and
 The Chinese University of Hong Kong, Hong Kong\\ \email{xiaoliang@cuhk.edu.hk}
 }

\let\oldaddcontentsline\addcontentsline
\def\addcontentsline#1#2#3{}
\maketitle
\def\addcontentsline#1#2#3{\oldaddcontentsline{#1}{#2}{#3}}

\begin{abstract}
  {\phantomsection
  \addcontentsline{toc}{section}{Abstract}}

  Secure multiparty computation (MPC) enables mutually distrustful parties to compute on private digital inputs. We initiate the study of \emph{spatiotemporal MPC}, extending this paradigm to functionalities whose inputs additionally depend on physical facts such as the parties' locations, times, or trajectories. Such protocols must simultaneously hide spatiotemporal information and ensure its physical consistency: a malicious party should not be able to make the functionality operate on a spatiotemporal input inconsistent with its actual physical state.\\[-0.5em]
  
  The main conceptual challenge is to formulate extraction of spatiotemporal information within the simulation-based security framework. We introduce \emph{arguments of spatiotemporal knowledge}, following the principle ``I prove, therefore I am:'' rather than defining physical presence directly through a mathematical relation, we define it operationally through the ability to complete a sound spatiotemporal verification protocol. Accordingly, an extractor recovers a spatiotemporal point from a successful prover and certifies its physical validity by using the extracted prover to succeed in an auxiliary spatiotemporal verification protocol. Building on this notion, we define universally composable (UC) security for spatiotemporal MPC, capturing privacy, physical consistency, and composability.\\[-0.5em]

  We provide constructions achieving this new MPC notion. We first construct UC-secure commit-and-prove protocols for spatiotemporal knowledge: in the CRS model under LWE against quantum provers without pre-shared entanglement, and in the QROM against quantum provers with unbounded pre-shared entanglement. Using these protocols, we obtain general-purpose UC-secure spatiotemporal MPC from semi-honest post-quantum MPC. We also extend our framework to UC-secure spatiotemporal multiparty quantum computation over private quantum inputs and classical spatiotemporal information.

  
  \end{abstract}

	\pagenumbering{roman}
\setcounter{tocdepth}{4} 
\tableofcontents
\addcontentsline{toc}{section}{Table of Contents}
\clearpage

\pagenumbering{arabic}


\section{Introduction}

A central theme in modern cryptography is how to compute while preserving privacy. This theme is epitomized by \emph{Secure Multiparty Computation} (MPC)
\cite{FOCS:Yao82b,FOCS:Yao86,STOC:GolMicWig87}, which allows a set of mutually distrustful parties to jointly evaluate a function on their private inputs while revealing no information beyond what is implied by the output. In this sense, MPC offers an idealized form of private computation for essentially any efficiently computable functionality. In recent decades, advances in both the theoretical and practical
foundations of MPC have brought us closer to realizing this vision for
a broad range of \emph{conventional} computational tasks.

The scope of what we regard as ``computation,'' however, has expanded
together with advances in science and technology. Computation is no
longer confined to evaluating mathematical functions on explicitly
represented digital inputs. For example, \emph{position-verification} protocols~\cite{C:CGMO09} seek to
certify that a prover is located at a specified geographic position; \emph{physical zero-knowledge}
protocols~\cite{C:FisFreNao14} allow a prover to establish physical properties (such as the equality of neutron radiographs or properties of
a DNA profile) without revealing the physical information underlying
those properties; 
\emph{proof of quantumness}~\cite{FOCS:BCMVV18} allows an efficient classical
verifier to certify that a device
possesses genuinely quantum computational capabilities; \emph{verification of
quantum depth}~\cite{chia2022classical} allows a classical verifier to certify that a device
performed a computation having sufficiently large quantum circuit
depth. In each of these examples, the task of
interest depends not only on abstract data, but also on a physical
object or a physical property of a participating device. This broader perspective naturally raises the following
question:
\begin{quote}
    \centering
    \emph{Can the notion of MPC be extended to support the secure computation of physical functionalities?}
\end{quote}

\para{Spatiotemporal MPC.}
In this work, we focus on a particular class of physical functionalities, which we call \emph{spatiotemporal functionalities}. Informally, consider parties $P_1, \ldots, P_n$, where each $P_i$ holds an ordinary private input $x_i$ and is associated with spatiotemporal information $\ell_i$, such as its physical location at a specified time or, more generally, a trajectory over an interval of time. The parties wish to jointly evaluate a functionality
\(
    F\bigl((x_1,\ell_1), \ldots, (x_n,\ell_n)\bigr)
\)
without revealing their inputs or spatiotemporal information beyond what follows from the prescribed output.

A protocol for \emph{spatiotemporal MPC} should provide guarantees analogous to those of traditional MPC, together with an additional form of \emph{physical consistency}. In particular, the private inputs $x_i$ and spatiotemporal data $\ell_i$ should remain hidden, while a corrupted party should be unable to make the functionality operate on spatiotemporal information that is inconsistent with its physical reality. Also, we emphasize the \emph{general-purpose} nature of this objective. Our goal is not merely to design a specialized location-verification mechanism for each individual application. Rather, as in general-purpose MPC, we seek a framework capable of securely realizing arbitrary functionalities whose inputs may include physical locations, times, or trajectories.

We next give three motivating examples:

\begin{enumerate}
    
\item {\bf Privacy-preserving regional analytics.}
Consider a service provider (e.g., an LLM provider) that offers different subscription plans across geographic regions. The provider, payment processors, and regional partners may wish to compute statistics such as regional conversion rates, demand estimates, fraud rates, or candidate price tiers. In the notation above, the ordinary inputs $x_i$ include subscription attributes and payment records, while the spatiotemporal inputs $\ell_i$ describe the participating devices' locations at the relevant transaction times. Traditional MPC can protect the privacy of the ordinary inputs. However, merely supplying a geographic region as another ordinary input does not establish that it corresponds to the participant's actual location. Spatiotemporal MPC adds precisely this physical-consistency guarantee while preserving the input privacy of MPC: the computation uses locations consistent with the participants' physical behavior, yet hides both the ordinary records and the underlying spatiotemporal information beyond what follows from the aggregate output.

\item {\bf Anonymous voting with geography-dependent representation.}
    Geographic information can also determine how votes are aggregated or translated into representation. In the U.S.\ Electoral College, each state is allocated a number of electors equal to the size of its congressional delegation, and, with the exceptions of Maine and Nebraska, states use a winner-take-all rule for appointing electors.\footnote{\url{https://www.archives.gov/electoral-college/allocation}}
    Similarly, seats in the European Parliament are allocated among Member States according to \emph{degressive proportionality}: more populous states elect more representatives, while less populous states receive more seats per capita.\footnote{\url{https://www.europarl.europa.eu/news/en/faq/4/how-are-meps-elected}} These examples motivate voting functionalities in which the rule applied to a ballot depends on the voter's jurisdiction. Spatiotemporal MPC could allow voters to submit secret ballots while privately binding each ballot to the correct district or region. The protocol would then apply the appropriate aggregation or weighting rule without revealing voters' identities, precise addresses, or districts beyond what is implied by the final result. Such a mechanism could also be useful in shareholder voting, union governance, digital cooperatives, or cross-regional online communities that seek geographically structured representation while retaining ballot secrecy.

\item {\bf Collaborative augmented reality.}
Co-located augmented-reality (AR) systems integrate virtual objects into a shared physical environment in real time~\cite{azuma1997survey}. Collaborative AR additionally requires multiple users to observe and interact with a consistent collection of virtual objects~\cite{billinghurst2002collaborative}. Participating devices may wish to jointly compute object placement, scene alignment, occlusion, and interaction updates from ordinary inputs $x_i$, such as camera data, spatial maps, and user interactions, together with spatiotemporal inputs $\ell_i$ describing their physical locations and trajectories. Although traditional MPC can keep the supplied data private, it cannot by itself prevent a malicious device from manipulating the shared scene through a fictitious location or trajectory. Our spatiotemporal extension addresses this gap by binding the locations and trajectories used in these updates to the devices' actual physical behavior. Crucially, this binding does not come at the expense of privacy: both the ordinary data and the spatiotemporal information remain hidden except for what is revealed by the prescribed scene updates.

\end{enumerate}
These examples are only illustrative. Other potential applications include privacy-preserving proximity services, location-dependent access control, collaborative robotics, traffic and mobility analytics, disaster response, contact tracing, and the enforcement of geographic or temporal constraints on distributed computations. Rather than treating each such application in isolation, spatiotemporal MPC seeks a common cryptographic abstraction for computations whose correctness depends on spatiotemporal facts.

\para{Wait! Why Not Position Verification?}
A vigilant reader may observe that there is already a closely related notion known as \emph{position verification} (PV)~\cite{C:CGMO09,C:BCFGGO11,buhrman2014position}.
As mentioned earlier, a PV protocol allows a prover to convince a (collection of) verifier(s) that the prover is located at a designated position. However, PV alone does not suffice for our intended notion of spatiotemporal MPC. A standard PV protocol merely certifies  the prover's asserted position; it does not enable multiple parties to evaluate an arbitrary functionality of their respective inputs and physical states. Also, it does not provide the input-privacy guarantee required by general-purpose MPC: the asserted position is ordinarily part of the public statement, and the prover and verifiers do not have the additional private  inputs that appear in an MPC execution. 

 
Recent work of Girish et al.~\cite{girish2026private} introduces a zero-knowledge variant of position verification, allowing a prover to establish its location while hiding all information not implied by this fact. This represents an important step toward privacy-preserving proofs of spatiotemporal facts. Nevertheless, zero-knowledge PV remains a proof functionality rather than a general-purpose computation functionality. Moreover, the soundness guarantee in~\cite{girish2026private} is seemingly relaxed: a prover can convince the verifier that it is located within a specified (large) region without revealing its exact position within that region. While this suffices for certain applications, it does not support finer-grained spatiotemporal functionalities in which the prover must prove that it is at a specific point in space-time. 

In fact, this should not be treated as a drawback of \cite{girish2026private}; it is  due to an inherent distinction between hiding a position and computing on it. If the public statement identifies an exact point in space-time, then that point is necessarily revealed by the statement itself and cannot simultaneously be hidden by zero knowledge. Conversely, a statement asserting membership in a larger region may conceal the exact point, but does not \emph{by itself} enable a functionality whose output depends on that exact point. Spatiotemporal MPC must reconcile these requirements by allowing the functionality to use fine-grained spatiotemporal information internally while hiding it from the other parties, except insofar as it is revealed by the output. Thus, a substantive gap remains between existing notions of position verification and our desired notion of spatiotemporal MPC.

\para{Spatiotemporal Knowledge Extraction.}
The aforementioned gap directs attention to a critical difference between proof/argument systems and general-purpose MPC: the role of \emph{input extraction} (and \emph{arguments of knowledge}). The security of traditional MPC is typically formulated through the real/ideal simulation paradigm. In the ideal world, each party submits an input to a trusted functionality, which evaluates the desired function and returns the prescribed outputs. Security requires that every real-world adversary can be simulated by an ideal-world adversary, so that the real execution reveals no more and permits no greater influence than the ideal execution. For a corrupted party, the input supplied to the ideal functionality need not coincide with the input that the party was nominally instructed to use. A malicious party may deviate from the protocol and thereby use some other \emph{effective input}. The simulator must therefore determine (or \emph{extract}) the effective input on which the corrupted party's behavior is based, and submit that input to the ideal functionality. Without such an extracted input, there is no well-defined ideal-world computation against which the real execution can be compared.

For ordinary MPC, the extracted object is mathematical: a bit string, a field element, an $\NP$ witness, or some other explicitly represented value. Spatiotemporal MPC requires something conceptually different. A corrupted party's relevant input may include its spatiotemporal property. To formulate simulation security, we must explain what it means for a simulator to \emph{extract such a property} and how the extracted object is related to the party's physical reality.

To the best of our knowledge, no existing work provides the corresponding notion of \emph{spatiotemporal knowledge extraction}.\footnote{A concurrent work by Bartusek et al.~\cite{bartusek2026track} introduces a related definition. We provide more information in \Cref{sec:related-works}. }  This issue is not merely definitional. A party may move during the execution of a protocol, relay messages through agents at different locations, or behave in a manner that is not naturally associated with a single point in space-time. A satisfactory solution must determine when such behavior corresponds to a well-defined spatiotemporal input, what an extractor should output when it does not, and which physical resources or observations the extractor is permitted to use. These questions have no direct analogue when the target of extraction is an ordinary mathematical witness.

In summary, spatiotemporal MPC presents new conceptual and technical challenges, and calls for a formal treatment of conventional cryptographic guarantees in a physically-consistent manner.

\subsection{Our Results}

We initiate a systematic study of spatiotemporal MPC and develop a
collection of definitions and constructions that lay the foundation
for this new notion. We summarize our main results below.

\subsubsection{New Definitions: I Prove, Therefore I Am}
\label{sec:new-definitions}

\para{Arguments of Spatiotemporal Knowledge.}
As discussed above, the main conceptual gap between position
verification and spatiotemporal MPC is the absence of an appropriate
notion of \emph{extraction} for spatiotemporal information. As our first
conceptual contribution, we introduce a new paradigm for overcoming
this definitional obstacle.

Our approach is guided by the following philosophy: it is challenging
to define a physical property using mathematical logic. We first
examine successful examples of defining physical properties, such as
proofs of quantumness and position verification, and observe that they
all rely on the notion of a ``proof.'' This observation suggests a
simple reversal of perspective: why should we undertake the
challenging task of capturing a physical property using mathematical
logic, \emph{when the interactive protocol used to test that property can
itself serve as its definition}? In other words, we can directly use the prover's   \textbf{ability} to pass position verification      as the \textbf{definition} of its physical presence at the asserted position.\footnote{This is reminiscent of Ren\'e Descartes'
famous dictum ``\emph{cogito, ergo sum}'' (``I think, therefore I
am''), which arose from his reflections on how to establish the
existence of one's own mind. We therefore call our framework
``\emph{I prove, therefore I am}.''} Inspired by this perspective, we propose the following notion of an
\emph{argument of spatiotemporal knowledge}:

\begin{definition}[Informal]
\label{def:informal-zk-aostk}
A position-verification protocol
$\Pi=\langle P,V\rangle$ is an \emph{argument of spatiotemporal
knowledge} if there exist an efficient two-stage extractor
$\mathsf{Ext}=(\mathsf{Ext}_0,\mathsf{Ext}_1)$ and a sound
position-verification protocol $\Pi'=\langle P',V'\rangle$ with the
following property:
\begin{itemize}
    \item For every potentially malicious efficient prover $P^*$, if
    $P^*$ convinces $V$ with probability $\delta$, then
    $\mathsf{Ext}_0^{P^*}$ extracts a spatiotemporal point
    $\mathsf{st}^*$ together with auxiliary information
    $\mathsf{aux}$ such that
    $\mathsf{Ext}_1^{P^*}(\mathsf{st}^*,\mathsf{aux})$ can act as the
    prover in $\Pi'$ and convince $V'$ that it is located at
    $\mathsf{st}^*$ with probability at least
    $\poly(\delta)-\mathsf{negl}(\lambda)$.
\end{itemize}
\end{definition}

A few aspects of this definition merit further clarification.
First, observe that the second position-verification protocol $\Pi'$
is {existentially} quantified. Thus, satisfying
\Cref{def:informal-zk-aostk} does not require
$\mathsf{Ext}_1^{P^*}(\mathsf{st}^*,\mathsf{aux})$ to succeed in every
possible position-verification protocol. It suffices for it to succeed
in a single protocol $\Pi'$, provided that $\Pi'$ satisfies the
standard soundness requirement for position verification. This choice reflects the role of $\Pi'$ in the definition: it is used
only to certify the physical validity of the extracted spatiotemporal
information $\mathsf{st}^*$. If $\mathsf{st}^*$ were not $P^*$'s real location, then $\mathsf{Ext}_1^{P^*}(\mathsf{st}^*,\mathsf{aux})$
could not successfully complete \emph{any} sound position-verification protocol.
Taking the contrapositive, successful participation in \emph{some} sound protocol $\Pi'$
suffices to certify the validity of $\mathsf{st}^*$. This existential
formulation also provides greater flexibility in designing $\Pi$ and
$\mathsf{Ext}$, as their construction and security analysis may be
tailored to a particular sound protocol $\Pi'$.

Second, \Cref{def:informal-zk-aostk} subsumes standard
position-verification soundness in a straightforward manner. In a
standard position-verification protocol, the claimed position is
public. The extractor can therefore output this claimed position as
$\mathsf{st}^*$ and take $\Pi'$ to be the original protocol $\Pi$.
Indeed, this serves as a useful sanity check on the definition. The main strength of \Cref{def:informal-zk-aostk}, however, is that it
captures spatiotemporal soundness in settings where the relevant
position is hidden, such as zero-knowledge position verification
\cite{girish2026private} and the MPC setting developed in this work.
In these settings, there is no publicly specified position against
which the prover's behavior can be checked. Instead, the position is
implicitly determined by the behavior of the potentially malicious
prover. Our definition addresses this difficulty by requiring the
extractor to recover spatiotemporal information $\mathsf{st}^*$ from
the prover's behavior and by certifying the physical validity of
$\mathsf{st}^*$ through the auxiliary position-verification protocol
$\Pi'$.

Finally, we emphasize that \Cref{def:informal-zk-aostk} is deliberately
simplified and intended only to convey the central idea behind
``I prove, therefore I am.'' In particular, no formal definition in the main body
corresponds directly to \Cref{def:informal-zk-aostk}. Rather, its
operational interpretation of spatiotemporal extraction underlies
the formal definitions introduced in this work
(e.g., \Cref{def:ext-stc,def:uc-smpc}). These formal definitions must
address several additional subtleties, including the need to
strengthen the extraction guarantees to support simulation
against quantum provers. We defer discussion of these subtleties
to \Cref{sec:tech-overview}.

\para{Spatiotemporal MPC.}
The notion of extracting spatiotemporal knowledge provides the
foundation needed to define spatiotemporal MPC. We formulate our
definition in the universally composable (UC) framework
\cite{FOCS:Canetti01}, extending it with our notion of arguments of
spatiotemporal knowledge. This yields a definition of
\emph{UC-secure spatiotemporal MPC}. Informally, our definition captures what one would naturally expect: it allows
parties to jointly evaluate a spatiotemporal functionality over their
private inputs while hiding both their ordinary inputs and their
spatiotemporal information. At the same time, it prevents a malicious
party from supplying spatiotemporal information that is inconsistent
with its actual physical state. Finally, our framework inherits UC's general
composition guarantees: UC-secure spatiotemporal MPC protocols can be
used as modular components of larger UC (spatiotemporal) protocols while preserving their
security guarantees.

\subsubsection{New Constructions}

We obtain a collection of constructions realizing our notion of
UC-secure spatiotemporal MPC under standard assumptions.

\para{Commit-and-Prove of Spatiotemporal Knowledge.}
Our first step is to construct a special two-party functionality with
a commit-and-prove flavor. Roughly speaking, this functionality allows
a party to commit to both an ordinary private input and private
\emph{spatiotemporal} information, and subsequently prove statements about
the committed values without revealing them. In addition to the
standard binding and zero-knowledge guarantees of commit-and-prove,
the functionality guarantees that the committed spatiotemporal
information is consistent with the committing party's actual physical
state. In our UC spatiotemporal framework, this consistency is
captured by requiring the simulator to extract the committed
spatiotemporal information and certify its physical validity through
a sound spatiotemporal verification protocol. A protocol realizing this
functionality plays a central role in all of our subsequent
constructions. We realize it under standard assumptions.

\begin{theorem}[Informal]
    \label{thm:informal:commit-and-prove}
    UC-secure \emph{commit-and-prove of spatiotemporal
    knowledge} exists in each of the following settings:
    \begin{enumerate}
        \item In the common reference string (CRS) model, assuming the
        quantum hardness of Learning with Errors (QLWE), knowledge
        soundness holds against quantum-polynomial-time (QPT) provers
        without pre-shared entanglement.
        
        \item In the quantum random oracle model (QROM), knowledge
        soundness holds against QPT provers with unbounded pre-shared
        entanglement.
    \end{enumerate}
\end{theorem}
Here, pre-shared entanglement refers to entanglement distributed among
the spatially separated components of a malicious prover before the
protocol begins. Such entanglement constitutes a particularly powerful
resource in position-verification protocols: it allows a malicious prover's
components to coordinate their attacks across different locations and
poses a significant obstacle to standard extraction techniques. 


\para{Spatiotemporal MPC.}
We next turn to general-purpose functionalities. At a high level, we
use \Cref{thm:informal:commit-and-prove} in a similar way that
traditional commit-and-prove protocols are used in the GMW compiler
\cite{STOC:GolMicWig87,STOC:CLOS02}, upgrading a semi-honest MPC
protocol for traditional functionalities into a maliciously secure
MPC protocol for \emph{spatiotemporal functionalities}. The
spatiotemporal nature of our setting, however, introduces several
unique challenges. We discuss these challenges and our solutions in
\Cref{sec:tech-overview}. For now, we state our result as follows.
\begin{theorem}[Informal]
    \label{thm:informal:spatiotemporal-mpc}
    Assuming the existence of semi-honest post-quantum MPC and
    UC-secure commit-and-prove of spatiotemporal
    knowledge (as provided by
    \Cref{thm:informal:commit-and-prove}), UC-secure spatiotemporal MPC
    exists.
\end{theorem}

\para{Spatiotemporal MPQC.}
A promising approach to constructing position-verification
protocols (and consequently, spatiotemporal MPC protocols) is to
exploit the no-cloning property of quantum information
(e.g.,~\cite{C:BCFGGO11,beigi2011simplified,
kent2011quantum,tomamichel2013monogamy,C:Unruh14,
das2021practically,ITCS:LiuLiuQia22,girish2026private,
bartusek2026track}). Since the parties in this setting are already
granted quantum capabilities, it is natural to ask whether our notion
of spatiotemporal MPC can be extended to support \emph{quantum}
functionalities. More specifically, we consider a setting in which each party holds a
quantum input, and the parties wish to securely and jointly evaluate a
quantum circuit over their respective quantum inputs and their
respective classical\footnote{Although the parties possess quantum capabilities and may
hold quantum inputs, we model their spatiotemporal information
classically. Quantum-mechanical notions of position, relational
time, quantum clocks, and quantum reference frames have been studied e.g.,~\cite{PhysRevD.27.2885,
PhysRevX.7.031022,GiacominiCastroRuizBrukner2019,
CastroRuizEtAl2020}. However, to the best of our knowledge, these formalisms have not yet
yielded a canonical operational model of quantum spatiotemporal
information suitable for our cryptographic setting. We leave the
formulation of a cryptographically meaningful quantum analogue to
future work.} spatiotemporal information. We call this notion
\emph{spatiotemporal multi-party quantum computation}
(\emph{spatiotemporal MPQC}). We show that this is indeed possible
assuming UC-secure MPQC (for traditional quantum functionalities without spatiotemporal aspects).

\begin{theorem}[Informal]
    \label{thm:informal:spatiotemporal-mpqc}
    Assuming the existence of post-quantum one-way functions
    (PQ-OWFs), UC-secure MPQC, and UC-secure zero-knowledge
    commit-and-prove of spatiotemporal knowledge (as provided by
    \Cref{thm:informal:commit-and-prove}), UC-secure
    spatiotemporal MPQC exists.
\end{theorem}

It follows from~\cite{C:BCKM21b} that UC-secure MPQC can, in turn, be
constructed from UC-secure post-quantum MPC.\footnote{Technically, the
construction also requires quantum garbled circuits.
However,~\cite{STOC:BraYue22} shows how to construct quantum garbled
circuits assuming only PQ-OWFs, which are themselves implied by the
existence of post-quantum MPC.}
We therefore obtain the following corollary.

\begin{corollary}
    \label{cor:informal:spatiotemporal-mpqc}
    Assuming the existence of UC-secure post-quantum MPC and
    UC-secure commit-and-prove of spatiotemporal
    knowledge, UC-secure spatiotemporal MPQC exists.
\end{corollary}

\para{Stand-Alone Security without CRS.} Note that a CRS can be generated via a quantum-secure coin-flipping protocol. Several such protocols are known: constant-round protocols from QLWE plus QFHE~\cite{STOC:BitShm20,EC:ABGKM21,FOCS:LiaPanYam23}; polynomial-round protocols from semi-honest OT~\cite{C:CLPY25} or PQ-OWFs plus quantum communication~\cite{C:BCKM21b,EC:GLSV21}; and constant-round protocols from semi-honest OT with relaxed simulation quality\footnote{There is strong evidence that, without such relaxation, constant-round constructions are unlikely to exist~\cite{FOCS:CCLY21,EC:CCLL26}, unless one resorts to non-black-box simulation techniques, which, in the current landscape, require stronger assumptions than minimal ones such as semi-honest OT or PQ-OWF.}~\cite{C:CCLY22,FOCS:LomMaSpo22,C:CLPY25}. We can instantiate the CRS in the first construction in \Cref{thm:informal:commit-and-prove} using any of these protocols to obtain a ``plain-model'' (i.e., without any setup such as a CRS) ZK commit-and-prove protocol for spatiotemporal knowledge, at the cost of downgrading UC security to stand-alone security.

Similarly, we can obtain plain-model versions (with stand-alone security) of \Cref{thm:informal:spatiotemporal-mpc}, \Cref{thm:informal:spatiotemporal-mpqc}, and \Cref{cor:informal:spatiotemporal-mpqc}  by first implementing the UC-secure protocols in the CRS model and then replacing the CRS with one of these coin-flipping protocols. For the first step, we recall that UC-secure PQ-MPC follows from QLWE in constant rounds~\cite{C:PeiVaiWat08,EC:GarSri18a}; the same holds for PQ-MPQC~\cite{C:BCKM21b}.

\subsection{Extensions and Future Directions}
\label{sec:extensions-and-future-directions}

Our work opens several avenues for future research, some of which we
highlight below.

\para{Optimized Constructions.}
The constructions presented herein should be viewed as initial
feasibility results, leaving substantial room for optimization, as was
also the case in the development of traditional MPC. One natural goal
is to obtain \emph{black-box} constructions---that is, constructions
that access their underlying building blocks only through their input/output
interfaces---thereby offering greater modularity and more transparent
efficiency guarantees. We believe that such constructions may be
obtained by combining our ideas with the \emph{MPC-in-the-head}
techniques~\cite{STOC:IKOS07,C:CCLY22}.

Another important direction concerns round complexity. Our current
constructions are not optimized in this respect, partly because they
invoke multiple position-verification instances spanning the MPC execution to ensure that
each party remains at its authenticated (hidden) position throughout the
computation. For two- or three-round protocols, however, the relevant
spatiotemporal interval could potentially be narrowed to the moment at
which a party sends its response in the 2PC/MPC protocol. This direction
is closely related to the traditional model of \emph{non-interactive secure
computation} (NISC)~\cite{EC:IKOPS11}. A formal treatment would require
a careful analysis of the concurrent/parallel security of our
constructions when composed with quantum NISC protocols, such
as~\cite{TCC:Bartusek21,AnanthBhardwajGupte26}.

\para{Deployment Potential on NISQ Devices.} 
Recent theoretical and experimental progress suggests that quantum position verification (QPV) is
approaching implementation on near-term quantum hardware
\cite{Bluhm2022,kanneworff2025towards,kavuri2026quantum}. Also, MPC has
undergone extensive optimization, resulting in increasingly practical
 implementations 
\cite{USENIX:MNPS04,ESORICS:BogLauWil08,AC:PSSW09,C:DPSZ12,CCS:Keller20,USENIX:PKYDPH21,USENIX:DalEscKel21,USENIX:WatWagPop22,SP:CZOWBZ23}.
Our constructions require  little quantum capability beyond
that needed to execute the underlying QPV protocol. Consequently, after
suitable cryptographic and systems-level optimization, they may become
candidates for implementation using NISQ-era components. Such an implementation would initiate a promising line of work at the
intersection of QPV and pragmatic MPC. In
particular, it could lead to the first experimental demonstration of
spatiotemporal MPC and illustrate a concrete application of near-term
quantum technology.

\para{Integration with Different PV Protocols.} The constructions
underlying \Cref{thm:informal:commit-and-prove} use a PV protocol based
on the no-cloning property of quantum information. However, our framework is more general and can be instantiated with any PV protocol
satisfying the following ``nice'' property, adapted from
\cite[Definition~2.8]{girish2026private}. Roughly, we call a PV
protocol \emph{nice} if the verifiers' joint verdict
can be computed by a classical predicate applied to the classical
portion of the interaction transcript, together with the relevant
timing information and verifier randomness. In particular, evaluating
this predicate does not require access to any residual quantum state
held by the verifiers.

This observation may allow our framework to be instantiated with PV
protocols based on other assumptions, including hardware
models~\cite{gabber1998prove}, bounded-retrieval
models~\cite{C:CGMO09}, and bounded-quantum-storage
models~\cite{Bluhm2022}. This would yield analogues of
\Cref{thm:informal:commit-and-prove,thm:informal:spatiotemporal-mpc}
in a variety of settings, with security inherited from the adversarial
model of the underlying PV protocol. For example, instantiating our
framework with a PV protocol secure against
adversaries with bounded quantum memory would yield a spatiotemporal
MPC protocol secure against the corresponding class of adversaries.
We leave a formal treatment of these instantiations to future work.

Worth particular mention is a new notion of PV proposed by 
Bartusek et al.~\cite{bartusek2026track}, referred to therein as
\emph{quantum localization}. Their protocols belong to the class of
quantum-based PV protocols, but achieve a stronger notion of soundness.
Whereas conventional soundness guarantees only that at least one member
of any successful adversarial coalition is present at the claimed
location, {quantum localization} certifies that a designated
unclonable quantum state exists at the claimed spacetime point and
cannot be found elsewhere. This rules out distributed strategies in
which different adversaries collectively emulate a \emph{single} prover across
multiple locations or times.

It would be appealing to integrate quantum localization into our
framework. Unfortunately, the verification procedure
of~\cite{bartusek2026track} relies on access to a structured
classical oracle in the ideal-obfuscation model. Such oracle
access does not directly provide an explicit polynomial-size
verification circuit: the resulting oracle-dependent acceptance
predicate is not, as given, a standard $\NP$ relation.
This presents an obstacle because our construction uses
zero-knowledge proofs for $\NP$ relations to establish that
the underlying verification procedure accepts, without
revealing the private information used in that verification.
These proofs therefore cannot be applied directly to the
oracle-dependent acceptance predicate.
Integrating the construction of~\cite{bartusek2026track}
requires addressing this oracle dependence while preserving
its stronger localization soundness and the privacy guarantees
required by our framework.
This is closely related to the open direction
in~\cite[Section~1.6]{bartusek2026track} of constructing
trajectory commitments and zero-knowledge localization.
If the oracle-dependent verification can be made compatible
with our zero-knowledge machinery, while retaining the
remaining properties required by our construction, this
would provide a route to spatiotemporal MPC with the \cite{bartusek2026track}-flavor localization
guarantee.

\para{Extension to General Physical Functionalities.}
Although this work focuses on spatiotemporal functionalities, it is
natural to ask whether our framework can accommodate more general
physical properties. To some extent, the framework extends to
properties admitting suitable verification protocols. Consider, for
example, proof of quantumness. Following the approach of
\Cref{def:informal-zk-aostk}, one could formulate an ``argument of
quantumness knowledge'' and, subsequently, MPC functionalities whose
inputs include a verified claim of quantumness.\footnote{This extension
may be less straightforward than it first appears. In particular, one
must carefully restrict the extractor's computational power and its
 access interface to $P^*$. We believe that potential issues can be addressed
by starting from a classical-verifier proof-of-quantumness protocol,
but a detailed treatment would take us beyond the scope of the present
work.}

Nevertheless, some physical properties  appear to lie beyond the reach
of our current framework. The difficulty stems from a subtle aspect of
the extractor in \Cref{def:informal-zk-aostk}: intuitively, the
extractor must obtain some form of access to the adversarial prover
$P^*$; however, combining the extractor with $P^*$, may sometimes blur the
boundary of the physical resource being certified and thereby yield
vacuous attacks on the underlying verification protocol.

To illustrate, suppose that a proof-of-quantum-depth protocol
guarantees that any successful prover must possess a quantum circuit
of depth at least $d$. A naive definition of an ``argument of
quantum-depth knowledge'' might give the extractor $\mathsf{Ext}_1$
black-box access to $P^*$. This immediately creates a vacuous
strategy: if $P^*$ implements a depth-$d$ quantum computation, then
$\mathsf{Ext}_1$ may invoke $P^*$ repeatedly and sequentially,
effectively assembling a quantum computation of arbitrarily large
polynomial depth. The resulting extracted resource would therefore
reflect the combined power of the extractor and the prover, rather
than the physical capability possessed by the prover itself. This
problem becomes even more pronounced for MPC functionalities involving
several physical properties simultaneously, because the mechanisms
used to extract one property may inadvertently supply resources
relevant to another.

It remains unclear whether this difficulty is inherent in any attempt
to define MPC for general physical functionalities or is merely a
limitation of the framework developed in this work. Resolving this
question is a very appealing, albeit ambitious, direction for future
research.

\subsection{Related Works}
\label{sec:related-works}

Chandran et al.~\cite{C:CGMO09} introduced the notion of
\emph{position-based} MPC, which supports functionalities that depend on the
parties' \emph{public} locations. Their notion differs fundamentally from our
spatiotemporal MPC: it does not
hide the parties' locations. In contrast, our notion guarantees the privacy of
the parties' spatiotemporal information while simultaneously ensuring
its physical authenticity. Achieving authenticity without revealing
this spatiotemporal information is, in fact, the central challenge of our
work.

The concurrent and independent work of Bartusek et al.~\cite{bartusek2026track} also
employs extraction in defining their new notion of \emph{quantum
localization}. This similarity, however, is only superficial: the two
works pursue fundamentally different goals and operate in distinct
settings. In their framework, extraction is used to \emph{localize a
quantum object}. Specifically, the verifier retains one share of an
entangled state, while the prover holds the other; possession of the
prover's share serves as an unclonable certificate of identity.
Security requires the existence of an extractor that acts only on the
prover's state at the claimed spacetime point and recovers a register
entangled with the verifier's share. Extraction thereby establishes
that a particular quantum object---and, consequently, the entity
identified by it---was present at the claimed spacetime point.
Repeatedly applying this guarantee makes it possible to track the
object's movement and thereby recover its trajectory.

In contrast, our extractor does not recover or track an identity-bearing
object located at a designated spacetime point. Rather, it plays the
standard simulation-based role of identifying the effective inputs
with which a malicious party participates in a computation, including
its claimed spatiotemporal information. Our security definition then
requires this extracted spatiotemporal information to be consistent
with the party's actual physical state. Thus, extraction in their work
is itself the mechanism by which the location of a persistent quantum
object is defined, whereas extraction in our work connects an
adversarial protocol execution to a well-defined ideal computation
whose spatiotemporal inputs must be physically authentic. Moreover,
their focus is on verifying the location or trajectory of an
unclonable quantum object, while ours is on securely computing general
functionalities over private ordinary and spatiotemporal inputs.

\section{Technical Overview}
\label{sec:tech-overview}

Our construction separates three requirements that are intertwined in
spatiotemporal computation: extracting a party's effective input,
certifying the physical meaning of that input, and using it consistently
throughout a private computation. Ordinary MPC addresses the first and
third requirements, but treats every input as an abstract string.
Position verification addresses the second, but ordinarily makes the
position public and does not provide the extraction interface needed
for simulation. Our approach connects these two forms of security
through a simulation-compatible notion of ``argument of spatiotemporal knowledge.''

We first explain this notion and its integration into the UC framework (\Cref{sec:overview-framework}).
We then construct a commit-and-prove of spatiotemporal knowledge protocol (\Cref{sec:overview-cnp}), use it
to authenticate private inputs during classical computation (\Cref{sec:overview-mpc}), and
finally give a generic extension to quantum functionalities (\Cref{sec:overview-mpqc}).
For exposition, we primarily consider two parties; the same
ideas apply to the multiparty setting.

\subsection{Simulation with Physically Meaningful Inputs}
\label{sec:overview-framework}

Suppose that party $P_b$ has private input
$(x_b,\mathsf{st}_b)$, where initially
$\mathsf{st}_b=(L_b,T_b)$ denotes a single spacetime point.
Once positions and times are represented by finite strings, evaluating    $F\bigl((x_0,\mathsf{st}_0),(x_1,\mathsf{st}_1)\bigr)$
is an ordinary computational task. The difficulty is that an ordinary
ideal functionality cannot distinguish an authentic position from a
fabricated one. Consequently, securely computing this function on
self-reported strings does not realize spatiotemporal computation.

The same distinction appears in simulation. An ordinary simulator may
extract a corrupted party's effective input
$(x_b^*,\mathsf{st}_b^*)$, but recovering the string
$\mathsf{st}_b^*$ does not establish that the adversary occupied the
corresponding point. We need extraction to certify not only
\emph{which value the adversary used}, but also
\emph{why that value has a valid physical interpretation}.

\para{Operational Certification of Extracted Information.}
Our starting point is the ``I prove, therefore I am'' principle from
\Cref{def:informal-zk-aostk}. Rather than require the extractor to
produce a mathematical witness of physical presence, we require it to
produce a strategy that succeeds in a sound position-verification
protocol for the extracted claim.

At a schematic level, an extractor $\msf{Ext}_0$ produces
$    (\mathsf{st}^*,\rho)
    \leftarrow
    \mathsf{Ext}_0^{P^*}(1^\lambda)$, 
and a second algorithm $\msf{Ext}_1$ uses $\rho$ and the prescribed access to $P^*$
to participate in a sound position verification protocol $\Pi'$ on public input
$\mathsf{st}^*$. Successful verification supplies the operational
meaning of the extracted claim. The auxiliary protocol $\Pi'$ is
existentially quantified: one sound test suffices to certify the claim;
success in every conceivable position-verification protocol is
unnecessary.

An essential restriction concerns the extractor's physical resources.
The proving strategy must respect the causal and spatial restrictions
of the adversarial model. In particular, extraction must not introduce
a new agent at the claimed position or provide communication or
entanglement unavailable to the adversary. Otherwise, successful
verification would certify the extractor's resources rather than the
adversary's physical presence.

\para{Simulatable Extraction.}
The operational certification in \Cref{def:informal-zk-aostk}
explains why an extracted claim has physical meaning, but does not
specify how the adversary can continue participating in the
surrounding computation after extraction. This is an additional
requirement: extracting information may disturb the adversary's
quantum state and thereby change its subsequent
behavior~\cite{EC:GLSV21,C:BCKM21b,C:CCLY22}. A strategy that succeeds
in a separate physical verification experiment does not, by itself,
resolve this problem.

We therefore strengthen operational certification with
\emph{simulation while extracting}. The first stage produces an
extracted claim and a residual state, and the second stage has two
alternative modes:
\[
\begin{aligned}
    \mathsf{Ext}_0(1^\lambda)
        &\longrightarrow (\mathsf{st}^*,\rho),\\
    \mathsf{Ext}_1(\mathtt{sim}, \mathsf{st}^*, \rho)
        &\longrightarrow
        \text{information for a simulated continuation,} \\
    \mathsf{Ext}_1(\mathtt{prv}, \mathsf{st}^*, \rho)
        &\longrightarrow
        \text{a prover for }\Pi'(\mathsf{st}^*).
\end{aligned}
\]
The $\mathtt{sim}$-mode output of $\msf{Ext}_1$ enables a continuation
that is computationally indistinguishable from the real execution,
while maintaining the appropriate correlations between the
extracted information, the adversary's residual state, and the surrounding
environment; the $\mathtt{prv}$ mode supplies the
physical certification. These roles are captured by two distinct
requirements. 

More concretely, consider a malicious party $P^*$ that first
commits to spatiotemporal information and then jointly evaluates a
predicate on it with an honest party $V$, who ultimately outputs
a decision bit. Let $v_{\mathsf{real}}$ denote the committed value
in the real commit-stage execution, or $\bot$ if there is no
unique valid committed value, and let $\mathsf{ST}_{P^*}$ denote
the adversary's residual state at the end of that stage. Our definition requires the following two properties hold simultaneously:
\begin{enumerate}
\item \textbf{Simulation.}
Running $\mathsf{Ext}_0$ followed by $\mathsf{Ext}_1$ in
$\mathtt{sim}$ mode, which outputs a simulated ``post-extraction'' state $\tilde{\msf{ST}}_{P^*}$ for the malicious $P^*$, as well as a simulated commit-stage
acceptance bit $\tilde{b}$ of the receiver. These outputs must satisfy:
\[
    \bigl(
        \Gamma_{\tilde{b}}(\mathsf{st}^*),
        \widetilde{\mathsf{ST}}_{P^*}
    \bigr)
    \cind
    \bigl(v_{\mathsf{real}},\mathsf{ST}_{P^*}\bigr),
    \qquad
    \text{where}~~
    \Gamma_{\tilde{b}}(\mathsf{st}^*) \coloneqq
    \begin{cases}
        \mathsf{st}^* & \text{if } \tilde{b}=1,\\
        \bot & \text{if } \tilde{b}=0.
    \end{cases}
\]
Crucially, the comparison concerns the
\emph{joint distribution of the value and residual state}, not
merely their separate distributions. Thus, the extractor must
both recover the committed information and produce a simulated
state with the appropriate correlations with that information.
This is what allows extraction to be used within a larger
computation without detectably changing the adversary's
continuation.

\item \textbf{Physical certification.}
If $P^*$ convinces $V$ to accept with probability
$\delta>\kappa(\lambda)$, where $\kappa$ is the knowledge error,
then the same first-stage extractor must also support a
successful physical proving strategy:
\begin{equation}\label{eq:overview-extraction-certification}
    \Pr\!\left[
        \left\langle
            \mathsf{Ext}_1(
                \mathtt{prv}, \mathsf{st}^*, \rho
            ),
            V'
        \right\rangle
        (1^\lambda,\mathsf{st}^*)=\msf{accept}
        \;:\;
        (\mathsf{st}^*,\rho)
        \leftarrow \mathsf{Ext}_0(1^\lambda)
    \right]
    \geq \poly(\delta).
\end{equation}
The proving strategy must obey the physical resource restrictions
described above.
\end{enumerate}
The two requirements are evaluated in separate experiments,
each starting with a fresh execution of $\mathsf{Ext}_0$.
They do not require copying $\rho$ or using it for both purposes
in one execution. Only $\mathtt{sim}$ mode is used to continue
the simulated computation; $\mathtt{prv}$ mode is an alternative
experiment certifying the physical meaning of the extracted
claim. Thus, simulation mode supplies the additional guarantee
needed beyond the operational intuition of
\Cref{def:informal-zk-aostk}, rather than another physical test.

\para{Integration into UC.}
The UC framework compares a real execution with an ideal execution
in the presence of an interactive environment
$\mathcal{Z}$~\cite{FOCS:Canetti01,STOC:CLOS02}.
Crucially, $\mathcal{Z}$ can exchange information with the adversary
during the execution, rather than only inspect a final transcript.
The simulator must therefore support an ongoing interaction with
an external machine whose state it cannot reset. This essentially rules out
rewinding-based simulation 
and motivates an online, straight-line interface to the environment.

Our spatiotemporal extension retains this interactive simulation
requirement, with quantum machines and quantum auxiliary states.
The simulator is organized as
$\mathcal{S}=(\mathcal{S}_0,\mathcal{S}_1)$.
During the execution, $\mathcal{S}_0$ extracts the corrupted
parties' effective ordinary and spatiotemporal inputs (encoded as a classical string) and
submits them through the ideal functionality's prescribed
interface. As with $\mathsf{Ext}_1$ above, $\mathcal{S}_1$
has two alternative modes. Its $\mathtt{sim}$ mode continues
the ideal execution so that the ordinary UC simulation requirement is satisfied:
$
    \mathsf{REAL}_{\Pi,\mathcal{A},\mathcal{Z}}
    \approx_c
    \mathsf{IDEAL}_{\mathcal{F},\mathcal{S},\mathcal{Z}}
$;
 its $\mathtt{prv}$ mode certifies the extracted
spatiotemporal claims in the spirit of
\Cref{eq:overview-extraction-certification}.

Unlike the simple predicate-evaluation setting above,
a general-purpose secure computation need not have a designated verifier
or an acceptance event defining $\delta$. We therefore
use the simulator's submission of an extracted input
as the relevant event: for each spatiotemporal string
$\mathsf{st}^*$, let $\delta_{\mathsf{st}^*}$ denote the
probability that $\mathcal{S}_0$ extracts $\mathsf{st}^*$
and the $\mathtt{sim}$-mode $\mcal{S}_1$ submits it to the ideal functionality. We will also utilize a knowledge error $\kappa(\secpar)$; whenever $\delta_{\mathsf{st}^*}>\kappa(\secpar)$, we require
$\mathcal{S}_1$ in proving mode to convince the verifier
of $\Pi'(\mathsf{st}^*)$ with probability at least
$\poly(\delta_{\mathsf{st}^*})$ in the corresponding
certification experiment. As above, the two modes are
evaluated in separate experiments, each starting with
a fresh execution of $\mathcal{S}_0$.
Only the simulation mode participates in the real/ideal
comparison; the proving mode supplies the additional
physical guarantee. Thus, physical certification supplements
ordinary simulation security rather than replacing it.

\para{On Composition.} In the UC framework, 
the environment's ability to interact with the adversary is 
what makes modular composition possible: a subprotocol's simulator
must remain valid when the surrounding protocol acts as its
environment. Our extension preserves this interface. Consequently,
the indistinguishability part of a composition proof follows the
usual UC replacement argument.

The physical part requires an additional step. The composed simulator
must retain the extracted claims and their associated proving
strategies, and the resulting strategies must remain admissible in
the physical adversarial model. In particular, replacing a
cryptographic subroutine must not introduce a forbidden
communication path or additional pre-shared entanglement into the
position-verification reduction. Under these compatibility
conditions, the same simulator composition carries the physical
certificate through the replacement. This is why merely observing that our definition
is stronger than UC would not, by itself, prove composition.
For functionalities with no spatiotemporal inputs, the additional
requirement is vacuous, recovering ordinary quantum-UC security.

\para{From points to trajectories.}
The ideal functionality may equally receive a finite ordered sequence
\[
    \mathsf{st}_b
    =
    \bigl((L_{b,1},T_{b,1}),\ldots,(L_{b,m},T_{b,m})\bigr).
\]
The extraction stage now recovers this sequence, while the proving
mode participates in a \emph{spatiotemporal trajectory verification}
protocol, as formalized in \Cref{def:STV}. 
Using a trajectory-level experiment avoids silently identifying
separate point-wise success probabilities with successful verification
of the entire trajectory. Our trajectories specify authenticated observations, not continuous
monitoring between observations. Moreover, for a distributed
adversary, authenticity means that the claimed points were occupied
by components of the coalition. In contrast to \cite{bartusek2026track}, it does not assert that a single
persistent device traversed them all (see
\Cref{rem:malicious-committer-trajectory}).

\subsection{Commit-and-Prove of Spatiotemporal Knowledge}
\label{sec:overview-cnp}

We next show how to construct a commit-and-prove protocol for
spatiotemporal predicates. Unlike a commitment to a single spacetime
point, our primitive records a private \emph{trajectory} and supports
subsequent proofs about its properties. The central challenge is to
make this physical record compatible with ideal/real simulation:
the simulator must extract a corrupted prover's trajectory, certify
its physical meaning, and preserve the prover's subsequent interaction
with the environment.

\para{On the Definition.}
This primitive is a special case of the UC spatiotemporal 2PC
definition introduced in \Cref{sec:overview-framework}. Its ideal
functionality $\mathcal{F}_{\mathsf{CnP}}^{\mathsf{ST}}$ has two
phases. In the \emph{Commit Phase}, it records a private trajectory
$\mathsf{st}$ under a commitment identifier and notifies the verifier
that commitment is complete. In the \emph{Prove Phase}, it receives
a public predicate $f$ and informs the verifier whether the recorded
trajectory satisfies $f(\mathsf{st})=1$, without revealing the
trajectory itself. Multiple commitments may accumulate over time,
and a later predicate may refer to several previously recorded
trajectories.

We require security of the commit phase itself, not merely of the
combined commit-and-prove execution. Essentially, this phase is an \emph{extractable-and-equivocal commitment}. The prove phase
 is a zero-knowledge argument of knowledge, linking the
requested predicate to the same committed trajectory.
Additionally, the extractability provided by the commit phase ensures that the
extracted
trajectory must satisfy our physical-certification requirement:
the extractor's proving mode must yield a successful strategy for
a sound trajectory-verification protocol.

A basic causality constraint determines the organization of the
construction. One may commit in advance to a \emph{description}
of an intended trajectory, but this does not certify that the
trajectory will actually be traversed. Physical authentication
must concern observations that have already occurred by the time
the protocol completes.

To see the distinction, consider a purported position-verification
protocol that terminates at time $t_{\mathsf{end}}$ and certifies
that the prover will occupy position $Y$ at a later time
$t_{\mathsf{future}}>t_{\mathsf{end}}$.
A malicious prover can always follow the honest strategy until
$t_{\mathsf{end}}$ and subsequently arrange to occupy
$Y'\neq Y$ at $t_{\mathsf{future}}$.
The verifier's completed view is unchanged, contradicting soundness.
Thus, a protocol cannot authenticate a freely chosen future
position merely by authenticating a present intention.
 Consequently, committing to an authenticated trajectory requires
a physical interaction spanning its observation window.

\para{Meshing the Observation Window.}
Let $\mathsf{T}$ be a publicly specified family of admissible
trajectories, represented over a polynomial-size spacetime mesh.
Partition the observation window into small intervals with
designated observation times
\[
    T_1<\cdots<T_m,
\]
and let $\mathsf{L}_i$ be the finite set of candidate locations at
time $T_i$. A trajectory has the form
\[
    \mathsf{st}=\bigl((L_i,T_i)\bigr)_{i=1}^{m},
    \qquad L_i\in\mathsf{L}_i.
\]
We schedule a position-verification instance for each candidate
point in
\[
    S=\bigcup_{i=1}^{m}
        \bigl(\mathsf{L}_i\times\{T_i\}\bigr).
\]
The commitment interval includes the transmission and processing
margins needed for these instances. The mesh specifies the
observations being authenticated.

\para{Position Verification under Commitments.}
For simplicity, consider a two-message position-verification protocol
$\Pi_{\mathsf{PV}}$ with a classical challenge $q$ and a classical
response $a$ (e.g., \cite{ITCS:LiuLiuQia22}). Let $\mathsf{Com}$ be a post-quantum commitment that is both extractable and equivocal (such a commitment is known, e.g., in the CRS model assuming QLWE \cite{C:DFLSS09}).
At each mesh point $\alpha\in S$,
the verifier sends the prescribed challenge $q_\alpha$.
Instead of returning $a_\alpha$ in the clear, the prover returns
a commitment $c_\alpha=       \mathsf{Com}(a_\alpha;r_\alpha)$. Notice that hiding the response alone is not enough: the pattern of responses
must not disclose the trajectory. Similar to \cite{girish2026private}, we use a public
traffic schedule covering every candidate point. At points on
its trajectory, the prover commits to genuine responses; elsewhere,
it commits to a dummy symbol $\bot$, with the same
message lengths and prescribed reception pattern.
Dummy commitments are independent of the challenges and can be
prepared in advance. Their transmissions are scheduled to meet
the candidate instance's reception times, whereas genuine
commitments depend on challenges received during the physical
interaction.

To fix which trajectory subsequent proofs refer to, the prover
also sends an extractable commitment
    $C=\mathsf{Com}(\mathsf{st};r)$ 
at initialization. This commitment fixes the claimed trajectory;
the commitments $\{c_\alpha\}_{\alpha\in S}$ subsequently supply
its physical authentication. Separating these roles prevents
the prover from selecting different subtrajectories in different
prove phases.\footnote{Technically, we also need to run a zero-knowledge argument to prove the consistency of the trajectory commitment $C$ with the response commitments $\{c_\alpha\}_{\alpha\in S}$; this is crucial to make the commit phase itself a well-formed extractable commitment. But we omit this discussion from the overview for simplicity.}

\para{Proof of Consistency.}
Let $\tau$ denote the (classical) transcript of the commit-phase
interaction. For a public polynomial-time predicate $f$, the prover gives a
zero-knowledge argument of knowledge for a relation $R_f$ with
public input $(C,\tau,\mathsf{T},f)$. The witness consists of a
trajectory $\mathsf{st}$, randomness opening $C$ to $\mathsf{st}$,
and, for every point $\alpha$ on that trajectory, a response
$a_\alpha$ and randomness opening $c_\alpha$ to $a_\alpha$.
The relation checks that:
\begin{enumerate}
    \item $\mathsf{st}$ belongs to the publicly specified family
    $\mathsf{T}$ of admissible trajectories, and the supplied
    opening of $C$ to $\mathsf{st}$ is valid;
    \item for every point $\alpha$ on $\mathsf{st}$, the supplied
    opening of $c_\alpha$ to $a_\alpha$ is valid, and $a_\alpha$
    satisfies the underlying position-verification predicate
    for that instance, using the verifier data and timing
    information recorded in $\tau$; and
    \item the committed trajectory satisfies $f(\mathsf{st})=1$.
\end{enumerate}
Thus, the same witness links the trajectory fixed by $C$ to
the responses fixed during the physical interaction and to the
requested predicate. Provided that the underlying verification
checks are classically computable in polynomial time, $R_f$ is
a classical $\NP$ relation, even if generating valid
position-verification responses requires quantum computation.
Zero knowledge hides the trajectory, the responses, and their
openings.

\para{Extraction with Physical Meaning.}
The commitment $C$ allows the simulator to extract
the a trajectory $\mathsf{st}^*$.
Extraction from the response commitments fixes the corresponding
responses. Binding and the prove-phase consistency argument ensure that any
accepting proof refers to these same values.
The remaining task is to connect these extracted
strings to physical reality. 

Toward that, we use the underlying $\Pi_{\mathsf{PV}}$ as the auxiliary
verification protocol. In the security reduction, the simulator
interacts with an external position verifier and forwards its
fresh challenge to the malicious prover as the challenge of the
corresponding commitment instance. When the prover returns
$c_\alpha$, the simulator extracts $a_\alpha$ and forwards it
to the external verifier.  That is, the reduction  places the external challenges
inside the adversary's actual interaction and converts committed
responses into ordinary responses online. The consistency relation transfers acceptance to the underlying
verification predicate, while position-verification soundness
supplies the physical interpretation.

\para{Obtaining the UC Realization.}
We instantiate the commitment and ZKAoK components of this construction with their
UC counterparts. This yields a realization of
$\mathcal{F}_{\mathsf{CnP}}^{\mathsf{ST}}$ satisfying the UC
spatiotemporal 2PC definition from the preceding subsection.
In particular, the construction supports concurrent commitment
sessions and later proofs about previously authenticated
trajectory segments.

\para{Extractable Spatiotemporal Trajectory Commitments.} 
In \Cref{sec:spatiotemporal-trajectory-commitment}, we also abstract the commitment to trajectories, together with its
physical-meaningfulness check, as an \emph{extractable spatiotemporal
trajectory commitment}. Its interface exposes a classical
commitment record and a classical opening predicate, together
with extraction that both preserves a simulated continuation and
supports physical certification of the extracted trajectory.
We believe that this abstraction provides a simple, modular
building block for future applications requiring private,
physically authenticated trajectory data.

\subsection{From Commit-and-Prove to Spatiotemporal MPC}
\label{sec:overview-mpc}

Our construction starts from a semi-honest MPC protocol and
follows the GMW paradigm: each party commits to its inputs and
randomness at the outset, and proves that each outgoing message
is consistent with these commitments and the prescribed
computation~\cite{STOC:GolMicWig87,STOC:CLOS02}. To evaluate a
spatiotemporal functionality, each party additionally commits
to a classical string describing its trajectory, which the MPC
uses together with the parties' other inputs. The key additional
requirement is that this string must describe the party's
\emph{physical trajectory}, rather than merely a trajectory that
it claims to follow.

Ordinary consistency proofs do not establish this physical
requirement: a party could execute the prescribed computation
correctly using a fabricated trajectory. Our spatiotemporal
commit-and-prove protocol supplies the missing authentication.
Its commit phase records trajectory information during the
execution, and its prove phase establishes consistency between
that physically authenticated information and the committed
trajectory string. The challenge is to integrate this physical
authentication with the computational consistency
checks.

\para{Round-by-Round Authentication and Consistency.} A naive attempt  is to commit to trajectories in parallel
with the GMW-style MPC execution; at the end, all parties perform the proof stage of the spatiotemporal commit-and-prove protocol to establish the proofs of physical
consistency through the computation. This checks the
right condition, but \emph{too late}! Consider a functionality
that releases a secret only to a party satisfying a
location-dependent condition. Even if every MPC message is
consistent with the prescribed computation, a party could
supply a fabricated trajectory and obtain the secret before
its physical consistency is checked. Rejecting its final proof
would not undo that disclosure. 

We therefore divide the execution into rounds, with a designated
physical observation interval $I_r$ for each round $r$.
During this interval, each party $P_b$ runs the commit phase
of our spatiotemporal commit-and-prove protocol for the
trajectory segment
$\mathsf{st}_{b,r}=\mathsf{st}_b|_{I_r}$, concurrently with the
corresponding MPC round. The prove phase then enforces two
conditions jointly:
\begin{itemize}
    \item \emph{Computational consistency:} the party's
    round-$r$ messages follow the semi-honest MPC protocol,
    using its committed inputs, randomness, and trajectory
    string, together with the preceding accepted transcript.
    \item \emph{Physical consistency:} the trajectory segment
    authenticated during $I_r$ agrees with the corresponding
    segment of the trajectory string committed to at the outset.
\end{itemize}
In this way, the proof stage both enforces honest behavior in
the underlying semi-honest MPC execution and certifies the
physical information on which that behavior depends.

The protocol advances to the next round only after all required
proofs for the current round have been accepted. Until then,
honest parties treat the round's incoming messages as pending
and take no dependent computational step. Crucially, merely
delaying the processing of incoming messages is not sufficient:
any outgoing message or output whose disclosure relies on the
current trajectory segment must itself be withheld until the
required physical consistency proofs have been accepted.
The final round is subject to the same rule. This ordering
prevents an unauthenticated trajectory segment from authorizing
a disclosure that a later rejection cannot undo. This  yields
\Cref{thm:informal:spatiotemporal-mpc}.

\subsection{Lifting the Construction to Quantum Functionalities}
\label{sec:overview-mpqc} 

We finally allow the parties' computational inputs and outputs to be
quantum, while retaining classical descriptions of their trajectories. First, notice that the ideas developed in \Cref{sec:overview-mpc} does not directly extend, because they require running a consistency (spatiotemporal) proof to check polynomial-time relations
between the committed trajectory segments and the party's round-$r$ computation for each round $r$; this works well for classical computations. However, one round of a \emph{quantum} MPC protocol may involve exchanging registers whose
correctness may require asserting that a particular transmitted
state was obtained by applying a prescribed quantum operation to
another state. Such an assertion is not an ordinary $\mathsf{NP}$ statement about
a classical transcript.
Nor does simply replacing $\mathsf{NP}$ proofs with proofs for
$\mathsf{QMA}$ solve the interface problem: standard $\mathsf{QMA}$
statements are classical, even though their witnesses may be
quantum. 
We therefore avoid using spatiotemporal proofs to certify
quantum message generation.

\para{A First Route: Authenticate the Classical Encoding of a
Quantum Computation.} The garbling-based approach of~\cite{C:BCKM21a} suggests a way to
separate the classical and quantum tasks. It uses the quantum
garbled circuit scheme of~\cite{STOC:BraYue22}, in which garbling
a quantum circuit consists of a \emph{classical} procedure that produces an
entirely classical garbled program: the quantum information
resides in the encoded inputs, rather than in the garbled
program itself. Classical secure computation can therefore
generate the garbled program and its classical control
information, while quantum input encoding and evaluation
take place outside that computation. Authentication and consistency mechanisms are developed
(thanks to \cite{C:BCKM21a}) to link these components.

For our purposes, the parties could supply their trajectories as
additional private inputs to this classical computation, which
programs the target quantum functionality using those values.
To avoid revealing a trajectory through the circuit description,
one can express the target using a fixed public universal circuit
with the trajectories supplied as private classical data.

We then apply a round-by-round compilation similar to \Cref{sec:overview-mpc}: recall that the protocol of~\cite{C:BCKM21a}
first runs a post-quantum classical MPC to generate the
classical garbled circuit and then performs additional
quantum rounds for input encoding and evaluation. As
explained above, the trajectory values are committed in
the classical stage and fixed as private data programming
the garbled quantum computation. Using the same spacetime
mesh and round-by-round authentication strategy, we run
our commit-and-prove protocol in parallel with the entire
execution, covering both the classical and quantum stages.
At each round boundary, the consistency proofs establish
that the trajectory segment authenticated during that
round agrees with the corresponding portion of the
trajectory already fixed in the garbled computation;
no dependent step proceeds until these proofs are accepted.
Thus, the checks during the later quantum rounds remain
tied to the same trajectory values ``hard-wired'' into the (garbled) target quantum circuit, yielding a protocol-specific route
to spatiotemporal MPQC.

Although this approach works, it lacks modularity: the consistency argument is tied to the ``classical-garbling'' computation of the \cite{C:BCKM21a} construction. Ideally, we would like to have a generic construction that can be applied to any underlying quantum protocol, allowing us to use the underlying quantum protocol as a black box.

\para{A Generic Approach.}
Our eventual construction requires only a 
UC-secure MPQC protocol and the spatiotemporal commit-and-prove
we build in \Cref{sec:overview-cnp}.
The  central observation is that the interface between these
components is entirely classical: both must use the same trajectory.

Each party first commits to its trajectory using an ordinary
post-quantum  binding
commitment, such as Naor's commitment instantiated with a
post-quantum pseudorandom generator.
Write $C_b = \mathsf{Com}(\mathsf{st}_b;r_b)$,
where $C_b$ denotes the complete public commitment record,
including any receiver message required by the scheme. Let $\mathcal{Q}$ be the desired spatiotemporal quantum
functionality, and let $\mathbf{C}=\Set{C_b}_b$ be the (public)
collection of the parties' trajectory commitments.
Instead of evaluating $\mathcal{Q}$ directly, the parties
invoke the MPQC protocol for a modified functionality
$\mathcal{Q}^{\mathbf{C}}_{\mathsf{check}}$.
It receives each party's quantum input register together
with $(\mathsf{st}_b,r_b)$, verifies
$
    \mathsf{Verify}_{\mathsf{Com}}
        (C_b,\mathsf{st}_b,r_b)=1
$ for each $b$, 
and aborts if any check fails. Otherwise, it applies
$\mathcal{Q}$ to the joint quantum input and the supplied
trajectories. The joint input may be entangled across
parties and with an external reference system; no
product-state assumption is needed.

In parallel with the MPQC execution, the parties run
the commit-and-prove protocol using a similar round-by-round
construction as in  \Cref{sec:overview-mpc}.
For each round, the commit phase authenticates the
corresponding trajectory segment, and the prove phase
establishes that this segment agrees with the corresponding
portion of the trajectory committed in $C_b$.
These are the same commitments whose openings are checked
by $\mathcal{Q}^{\mathbf{C}}_{\mathsf{check}}$, linking the
authenticated segments to the trajectories supplied to
the quantum computation. As before, no dependent step
proceeds until the required consistency proofs are accepted;
in particular, all trajectory segments on which an output
depends must be authenticated before that output is released.

The division of responsibility is now explicit.
The MPQC protocol guarantees correct quantum computation but does
not interpret physical presence.
The commit-and-prove protocol certifies physical presence but does
not verify quantum computation.
The ordinary commitment forces both components to refer to the
same classical trajectory.

\para{Returning to Classical Functionalities.}
The generic construction above treats the underlying MPQC
protocol as a black box. In particular, replacing the MPQC
protocol with a maliciously secure MPC protocol for classical
functionalities yields secure spatiotemporal MPC for classical
functionalities. This provides an alternative route to the
feasibility results described in \Cref{sec:overview-mpc}.
In the main body, we present only this generic construction,
which handles both classical and quantum functionalities in a unified manner.
We nevertheless retain the description in \Cref{sec:overview-mpc}
to provide a complementary approach, which  may admit/ease further
optimizations in future work, such as black-box constructions
based on black-box commit-and-prove techniques.

\para{Why an Ordinary Commitment Can Suffice.}
The linking commitment $C_b$ is not used as a separately UC-realized
commitment functionality.
In particular, the simulator need not extract a trajectory from
$C_b$ alone or open an honest commitment to a newly selected value. Instead, extraction occurs at the stronger interfaces already
provided by the construction.
In the MPQC hybrid, the simulator obtains the corrupted party's
effective classical trajectory and opening.
In the ``physical-relevant'' hybrid, it obtains the trajectory authenticated by
the commit-and-prove instance.
Both must be consistent with the same public record $C_b$.
Statistical binding therefore identifies the two values, while
computational hiding protects honest trajectories.
This is the reason that the linking commitment need not itself
be UC secure.

It is also worth emphasizing that an ordinary commitment may be \emph{malleable} \cite{STOC:DolDwoNao91}, and UC security of the surrounding protocols does not by itself provide non-malleability. This poses no problem for our construction: even if an adversary copies or transforms an honest party's commitment $C_b$, using the resulting record requires it to provide a valid opening within the quantum computation and to physically certify the same trajectory through its own commit-and-prove interface. The UC simulations for these components capture such behavior even across interleaved sessions with related public inputs, while binding ensures that the computational and physical components cannot accept different trajectories for the same record. A related trajectory that the adversary can both open and physically authenticate is simply an admissible effective input, whereas a copied record without a usable opening is insufficient. Thus, our security argument does not rule out commitment malleation per se; rather, it shows that malleation cannot bypass the interfaces through which the committed value must be supplied and physically certified. These interfaces incorporate session and party identifiers to prevent acceptance in one context from being reused as authorization in another.

The resulting transformation treats the underlying MPQC protocol
as a computational component, rather than inspecting its quantum
messages or reproducing its internal authentication mechanisms.
Under the stated composition and corruption assumptions, it yields
\Cref{thm:informal:spatiotemporal-mpqc}.


\section{Preliminaries}

\para{Basic Notations.} Let $\secpar \in \Naturals$ denote security parameter. 
For a positive integer $n$, let $[n]$ denote the set $\{1,2,...,n\}$.
For a finite set $\calX$, $x\sample \calX$ means that $x$ is uniformly chosen from $\calX$.

A function $f:\mathbb{N}\ra [0,1]$ is said to be \emph{negligible} if for all polynomial $p$ and sufficiently large $\secpar \in \mathbb{N}$, we have $f(\secpar)< 1/p(\secpar)$; it is said to be \emph{overwhelming} if $1-f$ is negligible, and said to be \emph{noticeable} if there is a polynomial $p$ such that $f(\secpar)\geq  1/p(\secpar)$ for sufficiently large $\secpar\in \mathbb{N}$.
We denote by $\poly$ an unspecified polynomial and by $\negl$ an unspecified negligible function.

Unless stated otherwise, all adversaries considered herein are modeled as a non-uniform QPT  algorithm, specified by a sequence of polynomial-size quantum circuits with quantum advice $\{\A_\secpar, \rho_\secpar\}_{\secpar\in\mathbb{N}}$. 
In an execution with the security parameter $\secpar$, $\A$ runs $\A_{\secpar}$ taking $\rho_\secpar$ as the advice.For simplicity, we often omit the index  $\secpar$ and just write $\A(\rho)$ to mean a non-uniform QPT algorithm specified by  $\{\A_\secpar, \rho_\secpar\}_{\secpar\in \mathbb{N}}$. When it is clear from the context, we also omit the advice $\rho$ and just write $\A$ to mean a non-uniform QPT algorithm.

\subsection{Post-Quantum Extractable and Commitment scheme}

We define post-quantum extractable commitment schemes with simulation. Our definition is taken almost verbatim from~\cite{FOCS:CCLY21}, with only cosmetic changes to fit our notation. 

\begin{definition}[Post-Quantum Commitment Scheme~\cite{FOCS:CCLY21}] A post-quantum commitment scheme $\Pi$ is a protocol between two probabilistic polynomial-time (PPT) machines: a committer $\mathcal{C}$ and a receiver $\mathcal{R}$. Let $m \in \bits^{\ell(\lambda)}$ (where $\lambda$ is the security parameter and $\ell(\cdot)$ is some polynomial) be a message that $\mathcal{C}$ wants to commit to. The protocol consists of the following stages: 

\begin{itemize}
   \item \textbf{Commit Stage:} $\mathcal{C}(m)$ and $\mathcal{R}$ interact with each other to generate a commitment $\mathsf{com}$, $\mathcal{C}$'s state $\mathsf{ST}_C$ and $\mathcal{R}$'s output state $\mathsf{OUT}_{R}$. We denote this execution by $(\mathsf{com},\textup{ST}_c,\mathsf{OUT}_{R}) \leftarrow \langle \mathcal{C}(m),\mathcal{R}\rangle(1^\lambda)$. When we consider a malicious quantum commiter $\mcal{C}^*$, we allow it to generate any quantum state $\mathsf{ST}_C^*$. Also, a malicious quantum receiver $\mathcal{R}^*$ can output any quantum state, which we denote by $\mathsf{OUT}_{R^*}$.

   \item \textbf{Decommit Stage:} $\mathcal{C}$ generates a decommitment $\mathsf{decom} \leftarrow \mcal{S}(\mathsf{ST}_\mcal{C})$. Then it sends a message $m$ and decommitment $\mathsf{decom}$ to $\mathcal{R}$, and $\mathcal{R}$ outputs a bit $b\in \{0,1\}$ indicating acceptance $(i.e.,b = 1)$ or rejection $(i.e.,b = 0)$. We assume that $\mathcal{R}$'s verification procedure is deterministic and denote it by $\mathsf{Verify}(\mathsf{com},m,\mathsf{decom})$. 
\end{itemize}
The scheme satisfies the following requirement:

\begin{itemize} 
  \item \textbf{Correctness.} For any $m \in \bits^{\ell(\lambda)}$, it holds that:
  \[\Pr \left[
  \begin{array}{cc}
       b = 1
  \end{array}
  :
  \begin{array}{l}
 
        (\mathsf{com}, \textup{ST}_C, \mathsf{OUT}_{R}) \leftarrow \langle \mathcal{C}(m),\mathcal{R}\rangle(1^\lambda) \\
        \mathsf{decom} \leftarrow \mathsf{ST}_C \\
        b\leftarrow \mathsf{Verify}(\mathsf{com},m,\mathsf{decom})
  \end{array}
  \right]
  =1\]

  \item \textbf{Computational Hiding.} For any $m_0,m_1 \in \{0, 1\}^{\ell(\lambda)}$ and any non-uniform QPT receiver $\mathcal{R}^*$, the following holds:
  \begin{align*}
    &\{\mathsf{OUT}_{R^*}: (\mathsf{com}, \textup{ST}_C, \mathsf{OUT}_{R^*}) \leftarrow \langle \mathcal{C}(m_0),\mathcal{R}^*\rangle(1^\lambda)\}_{\lambda}\\ 
    \overset{c}{\approx} ~
    &\{\mathsf{OUT}_{R^*}:(\mathsf{com}, \textup{ST}_C, \mathsf{OUT}_{R^*}) \leftarrow \langle \mathcal{C}(m_1),\mathcal{R}^*\rangle(1^\lambda)\}_{\lambda}.
  \end{align*}

  \item \textbf{Statistical binding.} For all (potentially unbounded) committer $\mathcal{C}^*$, the following holds:
  \[
  \Pr \left[ 
  \begin{array}{c} 
  \exists m_0, m_1, \mathsf{decom_0},\mathsf{decom_1},
  \\s.t.\ m_0 \neq m_1 \wedge  \\ 
  \mathsf{Verify}(\mathsf{com},m_0,\mathsf{decom_0})=1 \wedge \\
  \mathsf{Verify}(\mathsf{com},m_1,\mathsf{decom_1})=1 
  \end{array} 
  : 
  \begin{array}{l} 
  (\mathsf{com}, \textup{ST}_{C^*}, \mathsf{OUT}_{R}) \leftarrow \langle \mathcal{C}^*,\mathcal{R}\rangle(1^\lambda)
  \end{array}
  \right]
  =\negl(\lambda)
  \]

\end{itemize}
\end{definition}

We define the value function as follows:
\begin{definition}[Value Function]\label{def:value-function}
  For a post-quantum commitment scheme $\Pi$ satisfying \Cref{def:extcom}, we define the value function as follows:
  \[
  \mathsf{val}_\mathrm{\Pi}(\mathsf{com}) := 
  \begin{cases} 
    m & \textit{if } \exists \textit{ unique } m \textit{ s.t. } \exists \ \mathsf{decom}, \mathsf{Verify}(\mathsf{com}, m, \mathsf{decom}) = 1 \\ 
    \perp & \textit{otherwise} 
  \end{cases} \]
  We say that $\mathsf{com}$ is valid if $\mathsf{val}_\mathrm{\Pi}(\mathsf{com})\neq \perp$ and invalid if $\mathsf{val}_\mathrm{\Pi}(\mathsf{com})= \perp$.

\end{definition}

Then we define the post-quantum extractable commitment  with simulation as follows:
\begin{definition}[Post-Quantum Extractable Commitment with Simulation]\label{def:extcom}
  A post-quantum commitment scheme $\Pi$ is extractable with simulation if for any non-uniform QPT $\mathcal{C}^*$, there exists a QPT algorithm $\mcal{SE}$ such that 
  $$ \{ \mcal{SE}(1^\lambda) \}_{\lambda} \overset{c}{\approx} 
  \{(\mathsf{val}_\Pi(\mathsf{com}), \msf{ST}_{C^*})~:~(\mathsf{com}, \textup{ST}_{C^*}) \leftarrow \langle C^*, R \rangle(1^\lambda) \}_{\lambda}$$ 
\end{definition}

\para{Instantiation.} Constructions satisfying \Cref{def:extcom} are known in the following settings: in the CRS model based on QLWE \cite{C:DFLSS09}, and in the QROM without any additional assumptions \cite{EC:ABKK23}.

\subsection{Post-Quantum Zero-Knowledge Arguments} 

We define post-quantum zero-knowledge arguments.

\begin{definition}[Post-Quantum Zero-Knowledge Arguments]\label{def:zero-knowledge}$\Pi$ is an interactive protocol between two probabilistic polynomial-time (PPT) machines: a prover $P$ and a verifier $V$. After the interaction, $V$ outputs a bit $b \in \{0,1\}$, where $b=1$ denotes acceptance and $b=0$ denotes rejection. We use $b \leftarrow \langle P,V\rangle(1^\lambda,x)$ to denote the verifier's output after the interaction between $P$ and $V$ on common input $x$ and security parameter $\lambda$.

We say that $\Pi$ is a \emph{post-quantum  zero-knowledge argument} for a  language $L \in \NP$ if the following properties hold:
\begin{enumerate}
    \item \textbf{Completeness.} For all $x\in L$ and all $w\in \mcal{R}_L(x)$, 
    \[ \Pr[b = 1 ~:~ b \gets \langle P(w),V\rangle(1^\lambda,x)]\geq 1-\negl(\lambda).\]

    \item \textbf{Soundness.} For all $x\notin L$ and all non-uniform QPT prover $P^*$,
    \[ \Pr[b = 1 ~:~ b \gets \langle P^*,V\rangle(1^\lambda,x)]\leq \negl(\lambda).\]

    \item \textbf{Zero Knowledge.} For a malicious non-uniform QPT verifier $V^*$, we denote its interaction with the honest prover as $\msf{OUT}_{V^*} \gets \langle P(w), V^*\rangle(1^\lambda,x)$ where $\msf{OUT}_{V^*}$ is the output of the verifier $V^*$. There exists a QPT simulator $\mathsf{Sim}$ such that for all $x\in L$ and all $w\in \mcal{R}_L(x)$
    \[\{\mathsf{Sim}(1^\lambda,x)\}_{\lambda\in\mathbb{N}} \overset{c}{\approx}\{\msf{OUT}_{V^*}:\msf{OUT}_{V^*}\gets\langle P(w), V^*\rangle(1^\lambda,x)\}_{\lambda\in\mathbb{N}}.\]
\end{enumerate}
    
\end{definition}

\para{Instantiation.}
In the plain model, constructions satisfying \Cref{def:zero-knowledge} are known in polynomial rounds from QLWE~\cite{STOC:Watrous06}, and in constant rounds from QLWE together with the existence of quantum fully homomorphic encryption~\cite{STOC:BitShm20}.

Constant-round constructions are also known in the CRS model under QLWE~\cite{C:DFLSS09}, as well as in the QROM without additional assumptions~\cite{EC:ABKK23}. These constructions follow from the observation that post-quantum zero-knowledge arguments can be obtained from post-quantum extractable commitment schemes with only $O(1)$ additional rounds and no additional assumptions, relying on (a straightforward post-quantum extension of) \cite{TCC:Rosen04}. Moreover, both constructions can be made UC secure if the underlying post-quantum extractable commitment scheme is a UC secure commitment.

\subsection{UC Framework}

We assume familiarity with the UC framework \cite{FOCS:Canetti01}. We provide a brief overview here, following the wording of~\cite{EC:GarKiyPan17}. A formal treatment can be found in the literature (e.g., \cite{FOCS:Canetti01,STOC:CLOS02}). Recall that in the UC framework, the model for protocol execution consists of the environment \(\mcal{Z}\), the adversary \(\mcal{A}\), and the parties running protocol \(\pi\). In this paper, we consider static adversaries and assume the existence of authenticated communication channels. Let \(\mathsf{EXEC}_{\pi,\mcal{A},\mcal{Z}}(\kappa,z)\) denote a random variable representing the output of \(\mcal{Z}\) on security parameter \(\kappa\in\mathbb{N}\) and input \(z\in\{0,1\}^*\), with a uniformly chosen random tape. Let \(\mathsf{EXEC}_{\pi,\mcal{A},\mcal{Z}}\) denote the ensemble $\left\{
\mathsf{EXEC}_{\pi,\mcal{A},\mcal{Z}}(\kappa,z)
\right\}_{\kappa\in\mathbb{N},\,z\in\{0,1\}^*}.$

The security of a protocol \(\pi\) is defined using the \emph{ideal protocol}. In an execution of the ideal protocol, all parties simply provide their inputs to the \emph{ideal functionality} \(\mcal{F}\). The ideal functionality \(\mcal{F}\) securely performs the desired task and provides outputs to the parties, which forward these outputs to \(\mcal{Z}\). The adversary \(\mcal{S}\) in an execution of the ideal protocol is often called the \emph{simulator}. Let \(\pi(\mcal{F})\) denote the ideal protocol for functionality \(\mcal{F}\).

We say that a protocol \(\pi\) \emph{emulates} a protocol \(\phi\) if, for every adversary \(\mcal{A}\), there exists an adversary \(\mcal{S}\) such that no environment \(\mcal{Z}\), on any input, can distinguish with non-negligible probability whether it is interacting with \(\mcal{A}\) and parties running \(\pi\), or with \(\mcal{S}\) and parties running \(\phi\). We say that \(\pi\) \emph{securely realizes} an ideal functionality \(\mcal{F}\) if it emulates the ideal protocol \(\pi(\mcal{F})\).

Based on the UC framework, we define the UC commitment and UC zero-knowledge proof as follows:

\para{UC Commitment.} We define the ideal functionality of $\mathcal{F}_{\msf{Com}}$ as follows:
\begin{FigureBox}[label={F-Com}]{Ideal Functionality $\mathcal{F}_{\msf{Com}}$}
  $\mathcal{F}_{\msf{Com}}$ proceeds as follows, running with parties $P_1,\cdots,P_n$ and adversary $\mcal{S}$.
    \begin{itemize}
      \item \textbf{Commit phase.} Upon receiving a value $(\msf{Commit}, sid, P_i, P_j,x)$ from $P_i$, where $x \in \bits^{m}$, record the value $x$ and send the message $(\msf{Receipt},sid, P_i, P_j,x)$ to $P_j$.
      \item \textbf{Open phase.} Upon receiving a message $(\msf{open}, sid, P_i, P_j,x)$ from $P_i$, proceed as follows: If some value $x$ was previously recorded, then send the message $(\msf{Open},sid, P_i, P_j,x)$ to $P_j$ and $\mcal{S}$ and halt. Otherwise, halt.
    \end{itemize}
\end{FigureBox}

  \begin{definition}[UC Commitment]\label{def:uc-commit}
  A protocol is a \emph{universal composable (UC) commitment} if it securely realizes functionality $\mathcal{F}_{\msf{Com}}$.
  \end{definition}

\para{UC Zero-Knowledge.} We define the ideal functionality of $\mathcal{F}_{\msf{ZK}}$ as follows:
\begin{FigureBox}{Ideal Functionality $\mathcal{F}_{\msf{ZK}}$}

  $\mathcal{F}_{\msf{ZK}}$ proceeds as follows, running with parties $P_1,\cdots,P_n$ and adversary $\mcal{S}$. The functionality is
  parameterized by a binary relation $R$.
  
  \begin{enumerate} 
      \item Wait to receive a value $(\msf{verifier}, id,P_i,P_j,x)$ from some party $P_i$. Once such a value is received, send $(\msf{verifier}, id,P_i,P_j,x)$ to $S$ and ignore all subsequent $(\msf{verifier}, \cdot)$ values.
      \item Upon receipt a value $(\msf{prover}, id, P_i, P_j,x',w)$ from $P_j$, let $v=1$ if $x=x'$ and $R(x,w)=1$ holds, and $v=0$ otherwise. Send $(id,v)$ to $P_i$ and $\msf{S}$ and halt.
  \end{enumerate}
  \end{FigureBox}
  
  \begin{definition}[UC Zero-Knowledge]\label{def:uc-zk}
      A protocol is a \emph{universal composable (UC) zero-knowledge proof} if it securely realizes functionality $\mathcal{F}_{\msf{zk}}$.
      \end{definition}


\section{Spatiotemporal Verification and Commitment}
In this section, we formally introduce the model of spacetime and spacetime trajectory. Starting with a ``nice'' position verification, we give the definition and protocol of spatiotemporal trajectory commitment.
\subsection{The Model}\label{sec:model}
Our model is based on that specified in \cite{ITCS:LiuLiuQia22}. The model consists of three types of parties: provers, verifiers, and adversaries.
\begin{itemize}
    \item Space and time are continuous, and the clocks of all parties are synchronized. A spacetime point is represented by \(\alpha=(L,t)\in\mathbb{R}^d\times\mathbb{R}\), where \(L\in\mathbb{R}^d\) denotes the spatial location, \(d\) is the spatial dimension, and \(t\in\mathbb{R}\) denotes time. In this work, whenever we consider position-verification protocols with privacy guarantees against \emph{malicious} parties, we restrict our attention to the one-dimensional setting, i.e., \(d=1\). See \Cref{sec:spatial-dimension-barrier} for further discussion.

    \item Each party's location at a particular time is described by a spacetime point as \(\alpha=(L,t)\in\mathbb{R}^d\times\mathbb{R}\).

    \item Adversaries may also communicate over private quantum channels, thereby preventing the verifiers from detecting malicious activity through these communications.

    \item 
    The adversary considered herein may consist of multiple parties. For a coalition of quantum adversaries, we distinguish between the following two settings:
    \begin{itemize}
        \item \textbf{Unentangled adversaries:} The adversarial parties share no entanglement before the protocol begins.
        \item \textbf{Entangled adversaries:} The adversarial parties may share an arbitrary polynomial amount of entanglement before the protocol begins.
    \end{itemize}
    However, explicitly accounting for both types of adversaries throughout all our definitions would be cumbersome. We therefore use the phrase ``a coalition of non-uniform QPT adversaries'' in our definitions (e.g., in \Cref{def:stc}) as shorthand for two separate security notions: one against unentangled adversaries and the other against entangled adversaries. When presenting our constructions, security theorems, and proofs, we will then treat the two settings separately and explicitly state the hardness and model assumptions required for each.

    \item All computations are performed instantaneously, while messages transmitted over any communication channel propagate at unit speed (the speed of light).
\end{itemize}

We next introduce the notion of a spatiotemporal trajectory. In what follows, most definitions and constructions are formulated in terms of a party's trajectory rather than a single spacetime point.

\begin{definition}[Spatiotemporal Trajectory]
    \label{def:spatiotemporal-trajectory}
    A \emph{spatiotemporal trajectory} over
    \(\mathbb{R}^d \times \mathbb{R}\) is an ordered set
    \[
        \bigl(\mathbf{st}, \prec_{\msf{time}}\bigr),
    \]
    where, for some positive integer \(m\),
\[
    \mathbf{st}
    =
    \bigl\{(L_i,T_i) : i \in [m]\bigr\}
    \subseteq \mathbb{R}^d \times \mathbb{R}
\]
is a set of \(m\) observations with pairwise distinct timestamps,     i.e.,
    \[
        i \neq j \quad\Longrightarrow\quad T_i \neq T_j,
    \]
    and \(\prec_{\msf{time}}\) is the strict total order induced by the
    timestamps:
    \[
        (L_i,T_i) \prec_{\msf{time}} (L_j,T_j)
        \quad\Longleftrightarrow\quad
        T_i < T_j.
    \]

    Throughout this paper, a spatiotemporal trajectory is always ordered
    by its timestamps. We therefore simply write \(\mbf{st}\) for the ordered
    set \(\bigl(\mbf{st},\prec_{\msf{time}}\bigr)\). By a slight abuse
    of notation, when \(S\) is a possibly unordered set of spatiotemporal
    points, we write
    \[
        \mbf{st} \subseteq S
    \]
    to mean that \emph{the underlying set} $\mbf{st}$ satisfies
    \(\mathbf{st}\subseteq S\), disregarding the order. The intended
    interpretation will be clear from context.
\end{definition}

\subsection{Verification of Position and Spatiotemporal Trajectory}\label{sec:position-verification}

We first present the definition of position verification, following the formalism in \cite{girish2026private}.

\begin{definition}[Position Verification]\label{def:position-verification}
  Let $\Pi=\Set{\mcal{P},\mcal{V}}$ be an interactive protocol consisting of two QPT machines $\mcal{P}$ and $\mcal{V}$. The scheme should satisfy the following properties:
  
  \begin{itemize}
    \item \textbf{Completeness}: $\mcal{P}$ and $\mcal{V}$ hold the security parameter $1^{\lambda}$ and a spacetime point $\alpha$ as their public input. After the interaction, $\mcal{V}$ outputs a bit $b \in \{0,1\}$, where $b=1$ denotes acceptance and $b=0$ denotes rejection. We denote the interaction by $ b \gets \langle \mcal{P}, \mcal{V} \rangle(1^\lambda, \alpha).$
    It holds that
    \[ \Pr [b=1 ~:~ b\gets\langle \mcal{P}, \mcal{V} \rangle(1^\lambda, \alpha)] = 1-\negl(\lambda).\]
  
    \item \textbf{Soundness}: For (a collection of) malicious QPT provers $\mcal{P}^*$, let $\msf{R}^* = (\alpha_1^*, \cdots, \alpha_k^*)$ denote the set of spacetime points that $\mcal{P}^*$ occupies. Let $b \gets \langle \mcal{P}^*, \mcal{V} \rangle(1^\lambda, \alpha)$ denote the interaction between $\mcal{P}^*$ and $\mcal{V}$, where $b$ is the verifier's output. For all $\alpha\notin \msf{R}^*$, it holds that
    $$\Pr [b=1:~b\gets \langle \mcal{P}^*, \mcal{V} \rangle(1^\lambda, \alpha)] = \negl(\lambda).$$
  \end{itemize}
  \end{definition}
 One may further require that the position-verification protocol satisfy the following ``nice'' property introduced in~\cite{girish2026private}.
  \begin{definition}[Nice Position Verification Protocol]\label{def:nice}
  We say a position verification protocol $\Pi$ (as per \Cref{def:position-verification}) is \emph{nice} if it satisfies the following properties.
  \begin{itemize}
    \item $\Pi$ is a two-message, one-round challenge--response protocol in which the verifiers send challenges and the prover returns a response. For each position \(\alpha\), the verifiers apply the predicate \(W_{\alpha}(\cdot)\) to the prover's response to verify the claimed position.
    \item The predicate $W_{\alpha}(\cdot)$ used to verify the position $\alpha$ by $\mcal{V}$ is deterministic and classical.
    \item For all $\alpha = (L,t) \in \mathbb{R}^d\times \mathbb{R}$, there exist static verifier locations $X_1^S,\ldots,X_k^S\in\mathbb{R}^d$ whose convex hull contains the spatial point $L$.
\end{itemize}
\end{definition} 
\begin{remark}
    A position-verification protocol directly verifies only a spatial location, rather than a spatiotemporal point. However, the location can be associated with the time at which the protocol is executed. Thus, without loss of generality, we treat $\alpha$ as a predicate that verifies a spatiotemporal point.
\end{remark}

Then, based on the notion of spatiotemporal trajectory \Cref{def:spatiotemporal-trajectory}, we provide the definition of \emph{spatiotemporal trajectory verification}.  
\begin{definition}[Spatiotemporal Trajectory Verification (STV)]
    \label{def:STV} 
    A spatiotemporal trajectory verification scheme $\Pi=(\mcal{P}, \mcal{V}, \msf{T})$ consists of two QPT machines $\mcal{P}$ (dubbed the prover) and $\mcal{V}$ (dubbed the verifier), and a so-called \emph{admissible trajectory set} $\msf{T}$, which is a finite set where each element $\msf{st}\in\msf{T}$ is a spatiotemporal trajectory over $\mathbb{R}^d\times \mathbb{R}$ (as per \Cref{def:spatiotemporal-trajectory}).    The scheme should satisfy the following properties:
\begin{enumerate}
\item \textbf{Completeness}: $\mcal{P}$ and $\mcal{V}$, holding the security parameter $1^\secpar$ and a trajectory $\msf{st}\in\msf{T}$ as their pubic input, interact within a time window
during which $\mcal{P}$ moves along the trajectory as specified by $\msf{st}$. After the interaction, $\mcal{V}$ outputs a bit $b\in\{0,1\}$ indicating acceptance(i.e., $b=1$) or rejection(i.e., $b=0$). We denote this execution by $b \gets \langle \mcal{P}, \mcal{V} \rangle(1^\secpar,\msf{st})$. 

Using this notation, completeness requires that for all $\msf{st}\in\msf{T}$, it holds that  
$$\Pr [ b = 1 ~:~ b \gets \langle \mcal{P}, \mcal{V} \rangle(1^\secpar,\msf{st})] = 1 - \negl(\secpar).$$

\item \textbf{Soundness}: For a (coalition of) malicious $\mcal{P}^*$ interacting with the honest $\mcal{V}$, let $\msf{S}^* = \Set{(L_i, T_i)}_i$ denote the set of all spatiotemporal points that (some part of) $\mcal{P}^*$ \emph{has occupied} (see the discussion in \Cref{rem:malicious-committer-trajectory})
during the interaction. We denote this execution as $b \gets \langle \mcal{P}^*[\msf{S}^*], \mcal{V} \rangle(1^\secpar,\msf{st})$, where $b$ is the output of $\mcal{V}$ at the end of the interaction. 

Using this notation, soundness requires that for all  $\msf{st}\in\msf{T}$ and all a (coalition of) non-uniform QPT $\mcal{P}^*$ and its associated $\msf{S}^*$, it holds that
$$\Pr [ (b = 1) ~\land~ (\msf{st} \nsubseteq \msf{S}^*) ~:~ b \gets \langle \mcal{P}^*[\msf{S}^*], \mcal{V} \rangle(1^\secpar, \msf{st}) = 1 ] = \negl(\secpar).$$

\end{enumerate}
\end{definition}

This construction builds on the position-verification protocol. For a trajectory \(\msf{st}\), the verifiers execute the position-verification protocol to verify each spacetime point \(\alpha\in\msf{st}\).

\subsection{Spatiotemporal Trajectory Commitment}
\label{sec:spatiotemporal-trajectory-commitment}

In this part, we formally define the spatiotemporal trajectory commitment and present a construction.

\subsubsection{Definition}

In this section, we formally define spatiotemporal trajectory commitment schemes. We consider two variants: a standard scheme \Cref{def:stc} and an extractable scheme \Cref{def:ext-stc}. The latter strengthens the former by additionally guaranteeing that the committed trajectory can be efficiently extracted without affecting the adversary's view.

\begin{definition}[Spatiotemporal Trajectory Commitment]\label{def:stc} 
A \emph{spatiotemporal trajectory commitment} scheme $\Pi = (\mcal{C},\mcal{R},\msf{T})$ consists of two QPT machines  $\mcal{C}$ (dubbed the committer) and $\mcal{R}$ (dubbed the receiver), and a so-called \emph{committable-trajectory} set $\msf{T}$, which is a finite set where each element $\msf{st}\in\msf{T}$ is a spatiotemporal trajectory over $\mathbb{R}^d\times\mathbb{R}$ (as per \Cref{def:spatiotemporal-trajectory}).
The scheme satisfies the following requirements: 

The interaction between $\mcal{C}$ and $\mcal{R}$ is split into two stages, \textbf{Commit} and \textbf{Decommit}, which have the following syntax:
\begin{itemize}
    \item \textbf{Commit}: $\mcal{C}$ and $\mcal{R}$, holding the security parameter $1^\secpar$ as their common input, interact within a time window  
    during which $\mcal{C}$ moves along a trajectory $\msf{st}\in\msf{T}$. After the interaction, the $\mcal{C}$ output a state $\rho_{\reg{C}}$ and $\mcal{R}$ a state $\rho_{\reg{R}}$, which will be used in the later \textbf{Decommit} stage; $\mcal{R}$ also outputs a bit $b_{\mrm{com}}$ indicating acceptance (i.e., $b_{\mrm{com}} = 1$) or rejection (i.e., $b_{\mrm{com}} = 0$) of the commit-stage interaction; and let $\tau$ denote the collection of the \emph{classical parts} across all the messages exchanged between $\mcal{C}$ and $\mcal{R}$ during the interaction. We denote this execution as\footnote{Technically, the classical part $\tau$ of the message exchanged during \textbf{Commit}, as well as the bit $b_{\mrm{com}}$, could be included in $\rho_{\reg{R}}$. However, we treat them as a separate
output, as this convention will be useful for subsequent definitions, e.g., \Cref{def:ext-stc}.}  
    $$(\rho_{\reg{C}}, \tau, \rho_{\reg{R}}, b_{\mrm{com}}) \leftarrow \langle \mcal{C}_{\msf{st}}, \mcal{R}\rangle(1^\secpar).$$ 

    \item \textbf{Decommit:} $\mcal{C}$ on input $\rho_{\reg{C}}$ generates the \emph{decommitment information} $\msf{decom}$ and sends $(\msf{st}, \msf{decom})$ to $\mcal{R}$. Then, $\mcal{R}$ computes a (possibly quantum) predicate $\msf{Verify}$, which on input $\msf{st}$, $\msf{decom}$, $\rho_{\reg{R}}$, and $\tau$ outputs a bit $b_{\mrm{dec}}$ indicating acceptance (i.e., $b_{\mrm{dec}} = 1$) or rejection (i.e., $b_{\mrm{dec}} = 0$) of the decommitment.\footnote{W.l.o.g., we take the convention that $\msf{Verify}(\cdot, \cdot, \cdot, \cdot)$ always outputs 0 whenever $b_{\mrm{com}} = 0$.} We denote this execution as
    $$(\msf{st}, \msf{decom}) \gets \mcal{C}(\rho_{\reg{C}});~ b_{\mrm{dec}} \gets \msf{Verify}(\msf{st}, \msf{decom}, \rho_{\reg{R}}, \tau).$$
\end{itemize}  
The scheme should satisfy completeness, computational hiding, and \textbf{either} statistical binding \textbf{or} computational binding, as defined below:
\begin{enumerate}
\item
    \textbf{Correctness:} For all $\msf{st}\in\msf{T}$, it holds that 
    $$\Pr[b_{\mrm{com}} = b_{\mrm{dec}} = 1 ~:~ 
    \begin{array}{l}
        (\rho_{\reg{C}}, \tau, \rho_{\reg{R}}, b_{\mrm{com}}) \leftarrow \langle \mcal{C}_{\msf{st}}, \mcal{R}\rangle(1^\secpar); \\
        (\msf{st}, \msf{decom}) \gets \mcal{C}(\rho_{\reg{C}}); \\
        b_{\mrm{dec}} \gets \msf{Verify}(\msf{st}, \msf{decom}, \rho_{\reg{R}}, \tau)
    \end{array} ] = 1 -\negl(\secpar).$$ 

\item \textbf{Computational Hiding:} For a malicious receiver $\mcal{R}^*$ and an honest $\mcal{C}$ moving along a trajectory $\msf{st}$, we denote the commit-stage interaction as $\rho_{\reg{R}^*} \gets \langle \mcal{C}_{\msf{st}}, \mcal{R}^* \rangle (1^\secpar)$, where $\rho_{\reg{R}^*}$ denotes the commit-stage output of $\mcal{R}^*$. 

Using this notation, \emph{computational hiding} requires that for all QPT $\mcal{R}^*$, there exists a QPT machine $\mathsf{Sim}$ (dubbed the simulator) such that for all committable trajectories $\msf{st} \in \msf{T}$, it holds that
$$\Set{\msf{Sim}(1^\secpar)}_{\secpar \in \mathbb{N}} \cind \Set{\rho_{\reg{R}^*}~:~ \rho_{\reg{R}^*} \gets \langle \mcal{C}_{\msf{st}}, \mcal{R}^* \rangle (1^\secpar)}_{\secpar \in \mathbb{N}}.$$

\item \textbf{Computational Binding:}  For a (coalition of) quantum  $\mcal{C}^*$ interacting with the honest $\mcal{R}$, let $\msf{S}^* = \Set{(L_i, T_i)}_i$ denote the set of spatiotemporal points that (some part of) $\mcal{C}^*$ \emph{has occupied} (see the discussion in \Cref{rem:malicious-committer-trajectory})
 during the commit-stage interaction. We denote the commit-stage execution as
$(\tau,  \rho_{\reg{R}}, b_{\mrm{com}}) \gets \langle \mcal{C}^*[\msf{S}^*], \mcal{R} \rangle (1^\secpar),$
where $\tau$ denotes the collection of the \emph{classical parts} across all the messages exchanged between $\mcal{C}^*[\msf{S}^*]$ and $\mcal{R}$ during the commit-stage,  $\rho_{\reg{R}}$ denotes $\mcal{R}$'s commit-stage output, and $b_{\mrm{com}}$ is $\mcal{R}$'s bit indicating acceptance (i.e., $b_{\mrm{com}} = 1$) or rejection (i.e., $b_{\mrm{com}} = 0$) of the commit-stage interaction.

Using this notation, \emph{computational binding} requires that for any (coalition of) non-uniform QPT  $\mcal{C}^*$ and her associated $\msf{S}^*$, with probability at least $1 - \negl(\secpar)$ over the commit-stage execution $(\tau,  \rho_{\reg{R}}, b_{\mrm{com}}) \gets \langle \mcal{C}^*[\msf{S}^*], \mcal{R} \rangle (1^\secpar)$, there exists a unique $\msf{st}^* \in \msf{T}$ such that the following two properties hold:
\begin{enumerate}
    \item  $\msf{st}^* \subseteq \msf{S}^*$, and
    \item\label[Condition]{cond:binding}  for all other trajectories $\msf{st} \ne \msf{st}^*$, it holds that
    $$\Pr[ 
        \exists \msf{decom} ~\text{s.t.}~ \msf{Verify}(\msf{st}, \msf{decom}, \rho_{\reg{R}}, \tau) = 1]  = \negl(\secpar).$$
\end{enumerate}

\end{enumerate}

\end{definition}

\begin{remark}[On Malicious Committer's Trajectory]
    \label{rem:malicious-committer-trajectory} Unlike for an honest committer, the notion of a trajectory for a malicious
committer $\mcal{C}^*$ is subtle to define. This is because $\mcal{C}^*$ may
consist of a collection of QPT machines occupying multiple spatial locations
simultaneously. For example, $\mcal{C}^*$ may comprise three parties
$\mcal{C}^*_1,\mcal{C}^*_2,\mcal{C}^*_3$ that remain at three distinct locations
$L_1,L_2,L_3$, respectively, throughout the commit stage. These parties could
then easily emulate a single moving party following the trajectory
\[
    \msf{st}=\Set{(L_1,T_1),(L_2,T_2),(L_3,T_3)}
\]
for any time points $T_1<T_2<T_3$.

We emphasize that \Cref{def:stc} is \textbf{not} intended to distinguish between these
two scenarios. In particular, a distributed $\mcal{C}^*$ that emulates a
single party moving along a trajectory $\msf{st}$ is not considered to violate
the binding guarantee. Accordingly, rather than specifying a single
trajectory, we specify the set of spacetime points $\msf{S}^*$ occupied by
$\mcal{C}^*$ during the commit stage. Any subset
$\msf{st}\subseteq\msf{S}^*$ is then regarded as a valid trajectory that
$\mcal{C}^*$ may claim to have traversed during that stage.

As mentioned earlier, the recent work by Bartusek et
al.~\cite{bartusek2026track} introduced a notion called \emph{quantum
localization} to capture precisely the distinction between a single moving
party and a distributed party. However, \Cref{def:stc}, as well as the other
definitions in this paper, does not make this distinction. We leave exploring
the connection between our notion of spatiotemporal trajectory commitment and
quantum localization as an open question, as discussed previously in
\Cref{sec:extensions-and-future-directions}.

\end{remark}

Before presenting the extractable variant of spatiotemporal trajectory commitment, we first introduce the notion of a ``committed value.'' It will be used in \Cref{def:ext-stc} to avoid the possibility that the extracted value is not unique. 
\begin{definition}[Committed Value]\label{def:committed-value}
    For a spatiotemporal trajectory commitment scheme $\Pi = (\mcal{C},\mcal{R},\msf{T})$  (satisfying \Cref{def:stc}), consider the commit-stage execution $(\tau,  \rho_{\reg{R}}, b_{\mrm{com}}) \gets \langle \mcal{C}^*[\msf{S}^*], \mcal{R} \rangle (1^\secpar)$ with a (potentially malicious) committer  $\mcal{C}^*[\msf{S}^*]$, where we use the same notation as in the computational binding party of \Cref{def:stc}. We defined the \emph{committed value} of such an execution as:
    $$ 
        \msf{val}_{\Pi}(\tau, \rho_{\reg{R}}) \coloneqq 
        \begin{cases}
            \msf{st} & \text{if~$\exists$~a unique $\msf{st} \in \msf{T}$ s.t.~ $\exists~\msf{decom}$ s.t.~$\msf{Verify}(\msf{st}, \msf{decom}, \rho_{\reg{R}}, \tau) = 1$ } \\
            \bot & \text{otherwise}
        \end{cases}.
    $$
\end{definition}

\begin{definition}[Extractable Spatiotemporal Trajectory Commitment]\label{def:ext-stc} 
A spatiotemporal trajectory commitment scheme $\Pi = (\mcal{C},\mcal{R},\msf{T})$  (satisfying \Cref{def:stc}) is said to be \emph{extractable} if it additionally satisfies the following \textbf{Extractability} requirement:

For a (coalition of) quantum  $\mcal{C}^*$ interacting with the honest $\mcal{R}$, let $\msf{S}^* = \Set{(L_i, T_i)}_i$ denote the set of spatiotemporal points that (some part of) $\mcal{C}^*$ \emph{has occupied} (see the discussion in \Cref{rem:malicious-committer-trajectory})
 during the commit-stage interaction. We denote the commit-stage execution as
$$(\msf{ST}_{\mcal{C}^*}, \tau, \rho_{\reg{R}}, b_{\mrm{com}}) \gets \langle \mcal{C}^*[\msf{S}^*], \mcal{R} \rangle (1^\secpar),$$
where $\tau$ denotes the collection of the \emph{classical parts} across all the messages exchanged between $\mcal{C}^*[\msf{S}^*]$ and $\mcal{R}$ during the commit-stage, $\msf{ST}_{\mcal{C}^*}$ denotes  $\mcal{C}^*[\msf{S}^*]$'s commit-stage output, $\rho_{\reg{R}}$ denotes $\mcal{R}$'s commit-stage output, and $b_{\mrm{com}}$ is $\mcal{R}$'s bit indicating acceptance (i.e., $b_{\mrm{com}} = 1$) or rejection (i.e., $b_{\mrm{com}} = 0$) of the commit-stage interaction.

    Using this notation, we require that for any (coalition of) non-uniform QPT $\mcal{C}^*$ and her associated $\msf{S}^*$, there exists a two-stage extractor $\msf{Ext} = (\msf{Ext}_0, \msf{Ext}_1)$, a spatiotemporal verification protocol $\Pi' = (\mcal{P}', \mcal{V}', \msf{T}')$ (as per \Cref{def:STV}), and a polynomial $\poly(\cdot)$ such that the following holds:
    \begin{itemize}
        \item No parts of $\msf{Ext}=(\msf{Ext}_0, \msf{Ext}_1)$ can occupy any spatiotemporal point outside the set $\msf{S}^*$;
        \item $\msf{Ext}_0$ takes as input the security parameter $1^\secpar$ and outputs a pair $(\msf{st}^*, \rho)$;
        \item $\msf{Ext}_1$ has two modes of operation, indicated by a command string $\msf{cmd} \in \Set{\mathtt{sim}, \mathtt{prv}}$. When invoked on input $(\mathtt{sim}, 1^\secpar, \msf{st}^*, \rho)$, it outputs a tuple $(\tilde{\msf{ST}}_{\mcal{C}^*}, \tilde{b}_{\mrm{com}})$; when invoked on input $(\mathtt{prv}, 1^\secpar, \msf{st}^*, \rho)$, it could act as a prover strategy for the spatiotemporal verification protocol $\Pi'$.
    \end{itemize}
    These algorithms further satisfy the following properties:
    \begin{enumerate}
        \item \textbf{Simulation:} it holds that
         \begin{align*}
        & 
        \bigg\{
            \big(\Gamma_{\tilde{b}_{\mrm{com}}}(\msf{st}^*), \tilde{\msf{ST}}_{\mcal{C}^*}\big) 
            ~:~
            \begin{array}{l}
                (\msf{st}^*, \rho) \gets \msf{Ext}_0(1^\secpar) \\ 
                (\tilde{\msf{ST}}_{\mcal{C}^*}, \tilde{b}_{\mrm{com}}) \gets \msf{Ext}_1(\mathtt{sim}, 1^\secpar, \msf{st}^*, \rho)
            \end{array}
        \bigg\}_{\secpar \in\mathbb{N}} \\
        \cind
        ~~&
        \bigg\{
            (\msf{val}_{\Pi}(\tau, \rho_{\reg{R}}), \msf{ST}_{\mcal{C}^*})~:~(\msf{ST}_{\mcal{C}^*}, \tau, \rho_{\reg{R}}, b_{\mrm{com}}) \gets \langle \mcal{C}^*[\msf{S}^*], \mcal{R} \rangle (1^\secpar)
        \bigg\}_{\secpar \in\mathbb{N}},
        \end{align*}
        where $\Gamma_{\tilde{b}_{\mrm{com}}}(\msf{st}^*)$ is defined as follows: 
        $
        \Gamma_{\tilde{b}_{\mrm{com}}}(\msf{st}^*) \coloneqq 
        \begin{cases}
            \msf{st}^* & \text{$\tilde{b}_{\mrm{com}} = 1$} \\
            \bot & \text{$\tilde{b}_{\mrm{com}} = 0$}
        \end{cases}
        $.

        \item \textbf{Spatiotemporal Extraction:} If  
$$
    \delta(\secpar) 
    \coloneqq 
    \Pr[ b_{\mrm{com}} = 1 ~:~ (\msf{ST}_{\mcal{C}^*}, \tau, \rho_{\reg{R}}, b_{\mrm{com}}) \gets \langle \mcal{C}^*[\msf{S}^*], \mcal{R} \rangle (1^\secpar)]  \geq \kappa(\secpar)\footnote{All constructions in this work achieve security with respect to a negligible knowledge-error function \(\kappa(\secpar)\), which is left implicit throughout.},
$$
then 
$$  \Pr[
            (\msf{st}^* \in \msf{T}) ~\land~ 
              (b' =1)
        ~:~
        \begin{array}{l} 
            (\msf{st}^*, \rho)\gets\msf{Ext}_0(1^\secpar); \\
            b' \gets \langle \mathsf{Ext}_1(\mathtt{prv}, \rho), \mcal{V'} \rangle(1^\secpar, \msf{st}^*) 
        \end{array} 
        ] = \poly\big(\delta(\secpar)\big),$$
        where $b' \gets \langle \mathsf{Ext}_1(\mathtt{prv}, \rho), \mcal{V'} \rangle(1^\secpar, \msf{st}^*)$ denotes the execution of $\Pi'$ between $\mathsf{Ext}_1(\mathtt{prv}, 1^\secpar,  \msf{st}^*, \rho)$ acting as the prover and the honest verifier $\mcal{V'}(1^\secpar,  \msf{st}^*)$, who outputs $b'$ after the interaction.
    \end{enumerate}

\end{definition}

\begin{remark}[Comparison with Prior Work]
    Compared with the notion introduced in~\cite{girish2026private}, our notion provides two main extensions. First, it considers an entire spatiotemporal trajectory rather than a single spacetime point. This enables the verification of statements about movement over time---for example, that Alice remained within a specified region throughout a given time interval---rather than merely statements asserting her presence at a particular location and time. Second, our notion provides an extractability property by simulation-while-extracting. This property makes the notion more general and facilitates its application to secure multiparty computation.
    \end{remark}

\subsubsection{Spatial Dimension as a Barrier to Privacy}
\label{sec:spatial-dimension-barrier}

Our treatment of privacy must account for a physical
obstacle identified by Girish et al.~\cite[Section~5]{girish2026private}.
In a model allowing unrestricted directional communication,
a malicious verifier can selectively deliver protocol messages
along a chosen ray, or withhold them from selected directions,
and observe whether the interaction succeeds. Whenever
successful participation depends on receiving these messages,
the resulting behavior can reveal whether the prover lies
in the targeted region. The same obstacle applies to
our work: a malicious receiver may use
selective delivery to learn geometric information about
the committer's trajectory, even if the cryptographic
contents of the commitment are hiding. For example,
successful participation following a ray-restricted
challenge can reveal that the committer's position at
the relevant time lies on that ray. This obstacle is
inherent to the unrestricted directional-communication
model, rather than specific to our definition or
construction.

\para{Scope of this work.} 
Accordingly, throughout this paper we consider a
one-dimensional spatial interval, with the receiver's
two signal towers fixed at its endpoints and the
committer's admissible trajectory contained within
the interval. This geometry eliminates the directional-message
attack: from either tower, every admissible interior
position lies on the same inward-pointing ray.
Consequently, a malicious receiver cannot use the
choice of transmission direction to selectively target
some interior positions while excluding others.
The directional attack therefore reveals no additional
spatial information beyond the publicly specified
restriction that the committer lies within the interval.
Fixing the towers' locations does not require the
receiver to behave honestly; our security guarantees
allow malicious behavior within this physical model.
All our constructions, including trajectory commitments,
commit-and-prove protocols, and the spatiotemporal MPC
protocols developed in subsequent sections, are
formulated in this one-dimensional setting.

\para{Relation to higher-dimensional extensions.} It is worth noting that 
our  constructions and techniques are ``orthogonal'' to the
directional-delivery obstacle in a modular sense:
they rely on the interfaces and security guarantees
of the underlying spatiotemporal primitives, rather
than on a particular mechanism for preventing
directional leakage. Girish et al.~\cite[Section~5]{girish2026private}
suggest some potential approaches to overcoming
this obstacle, including restricting verifiers to
broadcast transmissions, enabling provers to distinguish
broadcast from directional signals, and using redundant
verifier networks under an honest-majority assumption.
If a higher-dimensional model or technique (e.g., \cite[Section~5]{girish2026private}) resolves
this obstacle and provides primitives satisfying the
same interfaces and security requirements used here,
our constructions and composition arguments carry
over to that setting.

\subsubsection{Construction}
\label{sec:the-protocol of spatiotemporal trajectory commitment}

In this section, we present a construction of an extractable spatiotemporal trajectory commitment scheme. Our construction extends the idea underlying \cite{girish2026private} from a single spacetime point to an entire spatiotemporal trajectory.

We use the following building blocks:
\begin{itemize}
    \item An extractable post-quantum commitment scheme $\msf{ExtCom} = \langle C,R\rangle$ satisfying \Cref{def:extcom}; let $r$ denote the number of rounds for this protocol.
   
    \item A ``nice'' position-verification protocol $\Pi_{\mathsf{PV}}$ satisfying \Cref{def:nice}. Note that this protocol has two rounds, For each spacetime point $\alpha$, let $W_{\alpha}$ denote the verification predicate associated with the corresponding position-verification protocol $\Pi_{\alpha}$.
    \item A post-quantum zero-knowledge proof system $\Pi_{\mathsf{ZK}}$ satisfying \Cref{def:zero-knowledge}.
\end{itemize}

We construct an extractable spatiotemporal trajectory commitment scheme $\Pi=(C,R,\msf{T})$ as follows. In the one-dimensional setting, it suffices to employ two coordinators positioned at distinct locations that delimit the relevant spatial interval.

\begin{ProtocolBox}[label={protocol:stc}]{Extractable Spatiotemporal Trajectory Commitment}

\para{Setup:}  The protocol uses the following setup:
\begin{itemize}
    \item 
    The protocol is parameterized by a security parameter \(1^\secpar\) and a nonempty set \(\msf{T}\) of admissible spatiotemporal trajectories. 
    \item 
A collection of designated $R = (R_1, R_2)$ is chosen as the coordinator. 

\item Let \(\Delta\) denote the maximum time required for a message to travel from any admissible space point\footnote{The set of admissible points consists of all space points appearing in any trajectory $\msf{st} \in \msf{T}$.} to any receiver \(R_i\). Let \(\Delta'\) denote the minimal time required for a message to travel from any admissible point to any receiver \(R_i\).
\item
 The protocol will execute within a time interval \([t_{\msf{init}},t_{\msf{final}}]\), which is chosen such that \(t_{\msf{init}}<t_{\msf{min}}(\msf{T})\leq t_{\msf{max}}(\msf{T})<t_{\msf{final}}\). More accurately, $t_{\msf{init}} \coloneqq t_{\msf{min}}(T) - (r+1)\Delta$ and $t_{\msf{final}} \coloneqq t_{\msf{max}}(T) + \Delta$.

\item
Let $\msf{L}$ denote the set of all location $L_i$'s in all trajectories in $\msf{T}$. Without loss of generality (e.g., by proper rounding to nearest space point), we assume that $|\msf{L}|$ is a polynomial of $\secpar$. That is, the supporting space of $\mcal{T}$ is meshed into polynomially many points. 
\item
The time interval \([t_{\msf{min}}(\msf{T}) - \Delta, t_{\msf{max}}(\msf{T})]\) is also meshed into a polynomial-size sequence of time steps $(t_1, t_2, \ldots, t_n)$ of equal length, where $t_1 = t_{\msf{min}}(\msf{T}) - \Delta$ and $t_n = t_{\msf{max}}(T)$, where $n$ is a polynomial of $\secpar$.
\end{itemize}

\para{Private Input:} The committer \(C\) follows a trajectory \(\msf{st}=\Set{(L_1,T_1),(L_2,T_2),\ldots,(L_m,T_m)} \in\msf{T}\), which it holds as its private input and to which it commits.

\para{Commit Phase.}
\begin{enumerate}

    \item \label[Step]{step:commit} At time \(t_{\msf{init}}\), \(C\) samples one-time-pad keys
    $\Set{\msf{sk}_{v, u}}_{v \in [n], u \in [|\msf{L}|]}$, where each $\msf{sk}_{v, u}\sample\{0,1\}^{\ell(\lambda)}$ is a uniformly random string of length $\ell(\secpar)$, and $\ell(\secpar)$ denote the length of the prover's message in the position-verification protocol \(\Pi_{\mathsf{PV}}\). $C$ then commits to $\msf{st}$ and $\Set{\msf{sk}_{v, u}}_{v \in [n], u \in [|\msf{L}|]}$ using $\msf{ExtCom}$, we denote the commitment transcript as \(\msf{com}\).

    \item For each time step $t_v$ in the sequence $(t_1, t_2, \ldots, t_n)$ and $t_v \le t_{\msf{max}}(\msf{T}) - \Delta'$, $R$ sends the first-round message (i.e., the challenge) of a fresh instance of $\Pi_{PV}$. We denote the instance of $\Pi_{PV}$ initiated at time $t_v$ as $\Pi^{(t_v)}_{PV}$. The challenge is broadcast to all locations in $\msf{L}$.

    \item $C$ then behaves as follows. For each time point $t_v$ and each space point $L_u\in\msf{L}$, let $y_{v,u}$ denote the response to the $\Pi^{(t_v)}_{PV}$ instance  \textbf{if $C$ were to be at location $L_u$}. $C$ sends the following $\Set{c_{v,u}}_{v \in [n], u \in [|\msf{L}|]}$ to $R$ with the correct timing:
    $$
        c_{v,u} \coloneqq
        \begin{cases}
            \msf{sk}_{v,u}\oplus y_{v,u} & \text{if there exists a $T$ such that $(L_u, T)\in \msf{st}$ } \\
            \msf{sk}_{v,u}\oplus \bot & \text{otherwise}
        \end{cases}, 
    $$
    where $\bot$ can be understood as a publicly designated dummy string (e.g., all zeros) representing an invalid or unsuccessful response. 

    By ``correct timing,'' we mean that the ``true'' response $\msf{sk}_{v,u}\oplus y_{v,u}$ is sent right at the time that the $\Pi^{(t_v)}_{PV}$ challenge arrives at $C$'s real location $L_u$, while the ``dummy'' response $\msf{sk}_{v,u}\oplus \bot$ is sent at a time that consists of the  $\Pi^{(t_v)}_{PV}$ timing constraints \textbf{if $C$ were to be at location $L_u$}. Note that our construction  allows this because $t_1$ is set to $t_{\msf{min}}(\msf{T}) - \Delta$, and thus $C$ can plan ahead of time to send the dummy responses at the right time.

        \item \label[Step]{step:zk}
        \textbf{(Consistency Proof.)} At the time $t_{\msf{final}}$, let $\msf{M}$ denote combined transcript exchanged between $C$ and $R$ during the time interval $[t_{\msf{init}}, t_{\msf{final}}]$.
    
        Starting from  time $t_{\msf{final}}$, $C$ and $R$  involve in the execution of the zero-knowledge argument where $C$ proves honest behavior as described above. More precisely, the common input is $x\coloneqq(\msf{M},\msf{com})$, and the witness is $w\coloneqq(\Set{\msf{sk}_{v,u}}, \msf{st}, \msf{decom})$; $C$ proves that $(x,w)$ satisfies the following conditions:
        \begin{itemize}
            \item $\msf{st} \in\msf{T}$;
            \item Using $\msf{decom}$ as the decommitment information, $\msf{com}$ is indeed a valid commitment to strings $\msf{st}$ and $\Set{\msf{sk}_{v,u}}_{v \in [n], u \in [|\msf{L}|]}$; 
            \item For each $c_{v, u}$ (where $u$ determines a unique $L_u \in \msf{L}$) in $\Set{c_{v,u}}_{v \in [n], u \in [|\msf{L}|]}$ 
            \begin{itemize}
            \item if there exists some (unique, by \Cref{def:spatiotemporal-trajectory}) $T$ such that $(L_u, T) \in \msf{st}$, then $c_{v, u}$ is equal to $\msf{sk}_{v,u} \xor y_{v,u}$ where $y_{v,u}$ is a correct response to the $\Pi^{(t_v)}_{PV}$ instance, consistent with  the constraint that $C$ is located at $L_u$ at time $T$. 
            \item if there does not exist any $T$ such that $(L_u, T) \in \msf{st}$, then $c_{v, u}$ is equal to $\msf{sk}_{v,u} \xor \bot$.
            \end{itemize}  
        \end{itemize}

        Importantly, this execution of ZK protocol needs to satisfies the timing constraints that $C$ pretends to be at the farthest location in $\msf{L}$. That is, for the $k$-th prover message of this ZK protocol, $C$ needs to time it such that it arrives at $R$ at time $t_{\msf{final}} + k\Delta$.

        \item $R$ outputs $b_{\mrm{com}}$ that is set to be the output of $R$ acting as the verifier of this ZK protocol.

\end{enumerate}

\para{Decommit Phase.}
\begin{enumerate}
    \item $C$ generates the decommitment information $\msf{decom}$ and sends $(\Set{\msf{sk}_{v,u}}_{v \in [n], u \in [|\msf{L}|]}, \msf{st}, \msf{decom})$ to $R$.
    \item $R$ accepts  if the decommitment information is valid; otherwise, it rejects.
\end{enumerate} 
\end{ProtocolBox}
\begin{lemma}[Security of the Spatiotemporal Trajectory Commitment Construction]\label{thm:security-stc}
    Assume that $\Pi_{\mathsf{PV}}$ is a nice position-verification protocol satisfying \Cref{def:nice} against unentangled (resp., entangled) adversaries. Assume further that $\Pi_{\mathsf{ZK}}$ is 
    a post-quantum zero-knowledge proof system satisfying \Cref{def:zero-knowledge}, and $\msf{ExtCom}$ is a  post-quantum extractable commitment scheme satisfying \Cref{def:extcom}. Then, \Cref{protocol:stc} is a spatiotemporal trajectory commitment scheme satisfying \Cref{def:ext-stc}  against unentangled (resp., entangled) adversaries.
    \end{lemma}

\begin{proof}

    \para{Completeness.} Consider an honest prover $P_{\msf{st}}$ following an admissible trajectory $\msf{st} \in \mathsf{T}$. At every point $\alpha \in \msf{st}$, $P_{\msf{st}}$ and the verifier execute the corresponding position-verification protocol $\Pi_\alpha$, with the prover's responses encrypted under the honestly generated secret key. By the completeness of $\Pi_\alpha$ and the correctness of the XOR operation, the encrypted transcript decrypts to responses satisfying the corresponding verification predicate $W_\alpha$, except with negligible probability. Moreover, the correctness of the commitment scheme implies that the honestly generated commitment can be successfully opened. Hence, the statement proved in zero knowledge is true, and the verifier accepts the proof by the completeness of the zero-knowledge proof system, except with negligible probability. It follows that both the commit and decommit phases are accepted. Taking a union bound over all $\alpha \in \msf{st}$, the overall failure probability remains negligible. 

    \para{Computational hiding.} Let $R^*$ be any malicious QPT receiver. We construct a QPT simulator $\Sim$ that emulates the joint view of $R^*$ throughout the commit phase without knowledge of the committed trajectory. The simulator replaces the commitment to $(\Set{\msf{sk}_{v, u}}_{v \in [n], u \in [|\msf{L}|]},\msf{st})$ with a commitment to an appropriately chosen dummy value, replaces all the encrypted position verification responses with encryptions of dummy messages of the same length, and simulates the zero-knowledge proof. By the computational hiding of the commitment scheme $\msf{ExtCom}$, the pseudorandom of the XOR operation, and the zero-knowledge property of the proof system, a standard hybrid argument shows that the output of $\Sim$ is computationally indistinguishable from the commit-stage output of $R^*$ in an honest execution. Thus, for every admissible trajectory $\Set{\msf{Sim}(1^\secpar)}_{\secpar \in \mathbb{N}} \cind \Set{\rho_{\reg{R}^*}~:~ \rho_{\reg{R}^*} \gets \langle {C}_{\msf{st}}, {R}^* \rangle (1^\secpar)}_{\secpar \in \mathbb{N}},$ and the resulting protocol~\Cref{protocol:stc} is computationally hiding.

    \para{Computational binding.} We prove the unentangled and entangled cases using the same reduction, always taking the adversarial attacking \Cref{protocol:stc} to be from the same  class against $\Pi_{\mathrm{PV}}$ satisfies \Cref{def:nice}.

    Let \(C^*\) be a coalition of malicious QPT committers, and suppose that it violates computational binding with probability \(p\) over \((\tau,\rho_{\reg{R}},b_{\mrm{com}})\gets\langle{C}^*[\msf{S}^*],{R}\rangle(1^\secpar)\). By the extractability of the underlying commitment scheme, a QPT extractor obtains the committed value \((\msf{st}^*,\Set{\msf{sk}_{v, u}}_{v \in [n], u \in [|\msf{L}|]})\), except with negligible probability. Moreover, the binding property of $\msf{ExtCom}$ ensures that \(C^*\) cannot successfully open the commitment to any trajectory \(\msf{st}\neq\msf{st}^*\), except with negligible probability. Hence, \Cref{cond:binding}  in \Cref{def:stc}  holds, and a violation of computational binding implies that \(b_{\mrm{com}}=1\) and \(\msf{st}^*\not\subseteq\msf{S}^*\) with probability at least \(p-\negl(\secpar)\).

    We construct a QPT adversary \(\mcal{A}\) against the soundness of the underlying spatiotemporal position verification protocol. The adversary \(\mcal{A}\) extracts \((\msf{st}^*, \Set{\msf{sk}_{v, u}}_{v \in [n], u \in [|\msf{L}|]})\) by the extractor of the underlying extractable post-quantum commitment scheme. Since \(\msf{st}^*\not\subseteq\msf{S}^*\), there exists a spacetime point \((L_u, T)\in\msf{st}^*\setminus\msf{S}^*\) can pass the verification. The adversary forwards the corresponding challenges generated by the honest verifiers for \(\Pi_{\msf{PV}}^{t_v}\) to \(C^*\), receives the corresponding ciphertext \(c_{u,v}\), decrypts it and forwards the recovered response to the verifiers while preserving the required timing.

    If \(b_{\mrm{com}}=1\), then, except with the negligible soundness error of the zero-knowledge proof system, the response is valid. Thus, with probability at least \(p-\negl(\secpar)\), \(\mcal{A}\) convinces the honest verifiers that it occupies the spatiotemporal point \((L_u,T)\), even though \((L_u,T)\notin\msf{S}^*\), meaning that \(C^*\) does not occupy this point. By the security of the underlying position-verification protocol, this probability must be negligible. It follows that \(p=\negl(\secpar)\), and hence the spatiotemporal trajectory commitment scheme is computationally binding.

    \para{Extractability.} Similar to the above, we prove the unentangled and entangled cases using the same reduction. For the extractability property, we prove both  \emph{simulation} and \emph{spatiotemporal extraction} properties.
   
    \subpara{Simulation.} Let $C^*[\msf{S}^*]$ be any malicious committer whose coalition occupies the set of spatiotemporal points $\msf{S}^*$ during the commit-stage execution. We construct a two-stage extractor $\msf{Ext}=(\msf{Ext}_0,\msf{Ext}_1)$ satisfying the simulation requirement. The first stage $\msf{Ext}_0(1^\lambda)$ invokes the extractor of the underlying extractable post-quantum commitment scheme against $C^*[\msf{S}^*]$. Let $(\Set{\msf{sk}_{v, u}}_{v \in [n], u \in [|\msf{L}|]},\msf{st}^*)$ denote the value extracted from the commitment. If the commitment has no unique valid opening of this form, $\msf{Ext}_0$ sets $\msf{st}^*:=\bot$. It then outputs $(\msf{st}^*,\rho)$, where the residual state $\rho$ is all information required to continue the interaction with $C^*$. This extraction is performed without placing any component of $\msf{Ext}_0$ at a spatiotemporal point outside $\msf{S}^*$.

        On input $(\mathtt{sim},1^\lambda,\msf{st}^*,\rho)$, the second stage $\msf{Ext}_1$ resumes the interaction with $C^*$ from the residual state $\rho$ and emulates the honest receiver $R$. In particular, $\msf{Ext}_1$ generates and transmits the challenges of the position-verification instances and receives the corresponding encrypted responses. Whenever $\msf{st}^*\neq\bot$, it decrypts these responses using the extracted key $(\Set{\msf{sk}_{v, u}}_{v \in [n], u \in [|\msf{L}|]},\msf{st}^*)$ and checks them against the position-verification predicates associated with the extracted trajectory $\msf{st}^*$. The extractor performs all remaining verification steps, including the zero-knowledge verification procedure, according to the honest-receiver algorithm, and sets $\widetilde{b}_{\mathrm{com}}:=0$ if any of these checks fails. It then outputs
            $\bigl(\widetilde{\msf{ST}}_{C^*},\widetilde{b}_{\mathrm{com}}\bigr)
            \gets
            \msf{Ext}_1(\mathtt{sim},1^\lambda,\msf{st}^*,\rho).$
        
        We now relate the effective extracted value $\Gamma_{\widetilde{b}_{\mathrm{com}}}(\msf{st}^*)$ to the committed value $\operatorname{val}_{\Pi}(\tau,\rho_R)$ from \Cref{def:committed-value}. If $\widetilde{b}_{\mathrm{com}}=1$, then all verification steps accept, and the soundness of the verification procedures, together with the binding property of the underlying commitment scheme, guarantees that $\msf{st}^*$ is the unique valid trajectory to which the commitment can be opened, except with negligible probability. Thus, $\operatorname{val}_{\Pi}(\tau,\rho_R)=\msf{st}^*=\Gamma_{\widetilde{b}_{\mathrm{com}}}(\msf{st}^*)$. If $\widetilde{b}_{\mathrm{com}}=0$, then either no unique valid value was extracted or one of the conditions required for a valid opening fails; hence, the real execution has no unique valid committed trajectory and $\operatorname{val}_{\Pi}(\tau,\rho_R)=\bot=\Gamma_{\widetilde{b}_{\mathrm{com}}}(\msf{st}^*)$. Consequently, except with negligible probability, $\Gamma_{\widetilde{b}_{\mathrm{com}}}(\msf{st}^*)
            \approx
            \operatorname{val}_{\Pi}(\tau,\rho_R).$

        By the simulation guarantee of the underlying extractable post-quantum commitment scheme, the joint distribution of the extracted value and the simulated residual state is computationally indistinguishable from that of the committed value and the corresponding residual state in a real commit-stage execution. Furthermore, $\msf{Ext}_1$ continues the interaction by applying the same efficient honest receiver algorithm to the simulated residual state. Hence, by closure of computational indistinguishability under efficient post-processing and the preceding correspondence between the effective extracted value and the committed value, we have the indistinguishability between the simulated output and the real output. Therefore, $\msf{Ext}$ satisfies the simulation requirement.

        \subpara{Spatiotemporal Extraction.} We follow the same reduction as in the simulation argument. The first stage $\msf{Ext}_0(1^\lambda)$ invokes the extractor of the underlying extractable post-quantum commitment scheme to obtain the committed value $(\Set{\msf{sk}_{v, u}}_{v \in [n], u \in [|\msf{L}|]},\msf{st}^*)$, together with the auxiliary state $\rho$ required to continue the execution. In the proving mode, $\msf{Ext}_1(\msf{prv},1^\lambda,\msf{st}^*,\rho)$ acts as the prover in the trajectory-verification protocol for $\msf{st}^*$. Whenever the verifier sends a challenge corresponding to a point $\alpha\in\msf{st}^*$, $\msf{Ext}_1$ forwards the challenge to $C^*$, receives the resulting encrypted response, decrypts it using $\Set{\msf{sk}_{v, u}}_{v \in [n], u \in [|\msf{L}|]}$, and forwards the recovered response to the external verifier. The (simulatable) extractability  of $\msf{ExtCom}$ guarantees that $(\Set{\msf{sk}_{v, u}}_{v \in [n], u \in [|\msf{L}|]},\msf{st}^*)$ is the uniquely committed value and that the post-extraction state of $C^*$ is preserved in a computationally-indistinguishable sense, while the correctness of the XOR operation guarantees that the response recovered by $\msf{Ext}_1$ is precisely the response generated by $C^*$. Moreover, by the soundness of the zero-knowledge proof system, if the proof at the end of the commit phase is accepted, then, except with negligible probability, the corresponding $\NP$ statement is true. In particular, all conditions encoded in the $\NP$ relation hold: $\msf{st}^*\in\mathsf{T}$, the commitment correctly opens to $(\Set{\msf{sk}_{v, u}}_{v \in [n], u \in [|\msf{L}|]},\msf{st}^*)$, and the decrypted responses is valid for all $\alpha\in\msf{st}^*$. Therefore, acceptance of the final proof certifies the validity of the intermediate protocol executions encoded in the relation. Consequently, if $C^*$ is accepted in the commit phase with probability $\delta(\lambda)\geq\kappa(\lambda)$, then $\msf{Ext}_1$ successfully convinces the external trajectory verifier, i.e., $b'=1$, and extracts a trajectory $\msf{st}^*\in\mathsf{T}$ with probability $\poly(\delta(\lambda))$, as required.

        This finishes the proof of \Cref{thm:security-stc}.

\end{proof}

Based on the adversarial model under consideration, we provide two instantiations.

\para{Instantiation against Unentangled Adversaries.}\label{sec:instantiation} The position-verification protocol proposed in~\cite{ITCS:LiuLiuQia22} satisfies \Cref{def:nice} and is secure against unentangled adversaries. Moreover, assuming QLWE, both the post-quantum extractable commitment and the zero-knowledge argument utilized in \Cref{protocol:stc} can be instantiated in the CRS model. Therefore, we have the following corollary.

\begin{corollary} In the CRS model, assuming QLWE, there exists an extractable spatiotemporal trajectory commitment scheme that is secure against unentangled QPT adversaries.
\end{corollary}

\para{Instantiation against Entangled Adversaries (as Defined in~\Cref{sec:model}).} The position-verification protocol presented in~\cite{C:Unruh14} achieves security against entangled adversaries (with polynomial queries) in the quantum random-oracle model (QROM), without any extra assumptions. This protocol satisfies \Cref{def:nice}. Then, using known QROM based post-quantum extractable commitment and zero-knowledge proof systems, we have the following \Cref{cor:entangled}.

However, we emphasize that this corollary does not
follow by simply replacing the commitment and ZK protocols
in \Cref{protocol:stc} with their QROM instantiations.
Such a replacement is insufficient because
\Cref{step:zk} of \Cref{protocol:stc} requires a ZK proof
of the validity of the extractable commitment execution
in \Cref{step:commit}. In the QROM instantiation, the
commitment algorithm accesses a random oracle, so the
statement to be proved involves oracle evaluations
that cannot be represented by an explicit circuit
for the oracle. This issue can be resolved by a straightforward
adaptation of the black-box commit-and-prove technique
of~\cite{C:CCLY22}. Specifically, one can instantiate
\Cref{step:commit,step:zk} jointly using an extractable
commit-and-prove protocol adapted from~\cite{C:CCLY22},
which uses the QROM-based extractable commitment
only as a black box. This yields the result claimed
in \Cref{cor:entangled}.

\begin{corollary}[of \Cref{thm:security-stc} and \cite{C:CCLY22}]\label{cor:entangled} 
    In the QROM, there exists an extractable spatiotemporal trajectory commitment scheme secure against entangled adversaries (with polynomial queries).
\end{corollary}

\section{UC-secure Spatiotemporal Multi-party Computation}
\label{sec:st-mpc}
\subsection{The SUC Framework}

We present the real-world execution, the ideal process, and the definition of \emph{Spatiotemporal UC} (SUC) realization for spacetime-aware multi-party functionalities. We then state the hybrid-model composition theorem for this notion. The exposition closely follows its ``ordinary'' UC analogue~\cite{FOCS:Canetti01,STOC:CLOS02,EC:Unruh10}.

\para{Protocol syntax.} Following the standard formulation of protocols as systems of quantum polynomial-time machines, we model a protocol as a collection of such machines, each specifying the program executed by a single party. These machines may exchange quantum information. Adversarial entities are likewise modeled as quantum polynomial-time (QPT) machines, and may receive arbitrary non-uniform auxiliary information.

\para{The basic framework.} The security of protocols that carry out a given task (or solve a given protocol problem) is defined in three steps. First, we formalize the execution of a protocol in the presence of an adversary and within a given computational environment; this is referred to as the \emph{real-world model}. Second, we formalize an ideal process for carrying out the task. In this ideal process, the parties do not communicate directly with one another. Instead, they interact with an \emph{ideal functionality}, which can be viewed as an incorruptible trusted party programmed to capture the desired behavior of the task. Finally, a protocol is said to SUC-realize an ideal functionality if its real-world execution \emph{emulates} the corresponding ideal process. We next describe the real-world execution model, the ideal process, and the notion of protocol emulation.

We focus on a computational model intended to capture realistic communication networks. Communication takes place over an asynchronous public network that does not guarantee message delivery. We assume authenticated communication; thus, the adversary cannot modify messages sent by honest parties or inject messages on their behalf. However, the adversary controls message delivery: it may deliver any previously sent message at most once or prevent its delivery altogether. Because the network is asynchronous, messages need not be delivered in the order in which they were sent. Parties may be corrupted adaptively during the execution, after which their behavior may be controlled arbitrarily by the adversary. Upon corrupting a party, the adversary obtains its entire internal state, including all retained information from prior stages of the execution. Finally, all entities involved in the execution are restricted to quantum polynomial-time (QPT), or \emph{feasible}, computation.

\para{Protocol execution in the real-world model.} We sketch the execution of a protocol $\pi$, represented by quantum polynomial-time machines $P_1,\ldots,P_n$, in the presence of an adversary $\Adv$ and an environment $\mathcal{Z}$, where all machines are parameterized by a security parameter $\secpar \in\mathbb{N}$. The execution proceeds as a sequence of activations, beginning with the activation of $\mathcal{Z}$. During each activation, the active machine may transfer designated classical or quantum registers to another machine. It then enters a waiting state and activates the recipient machine.

Whenever the adversary is activated, it may either deliver a message to a party or corrupt a party. Only messages previously sent by a party may be delivered, and each such message may be delivered at most once. Upon corrupting a party, the adversary obtains access to its entire internal state, including all retained information from prior stages of the execution, and controls all of its subsequent actions. The corrupted party may thereafter behave arbitrarily. In addition, the environment is notified whenever a party is corrupted. If, during its activation, the adversary delivers a message to an uncorrupted party, that party is activated after the adversary's activation ends. Otherwise, the environment is activated next.

Whenever a party is activated, either upon receiving an input from the environment or upon receiving a message delivered by the adversary, it executes its prescribed code and may send messages to other parties. The protocol execution terminates when the environment completes an activation without sending a message to any other entity. The output of the execution is defined as the output of the environment, which we assume to be a single bit.

In summary, the order of activations is as follows. The environment $\mcal{Z}$ is always activated first. During its activation, $\mcal{Z}$ may activate either the adversary $\Adv$ or a party $P_i$. If $P_i$ is activated directly by $\mcal{Z}$, it receives an input of the form $(x_i,\msf{st}_i)$, where $x_i$ is the classical input provided by $\mcal{Z}$ and $\msf{st}_i$ represents the trajectory of $P_i$; during the execution of the protocol, $P_i$ moves along this trajectory $\msf{st}_i$. If $\Adv$ is activated, it may either return control to $\mcal{Z}$ or activate a party $P_i$ by delivering a message to it. After the activation of $P_i$ is complete, control is always returned to $\mcal{Z}$. We emphasize that exactly one entity is active at any point in the execution. Furthermore, during each activation, $\mcal{Z}$ and $\Adv$ may each activate at most one other entity.

Let $\mathrm{REAL}_{\pi,\Adv,\mcal{Z}}(k,z,\bar{r})$ denote the output of the environment $\mcal{Z}$ in an execution with the adversary $\Adv$ and parties $P_1,\ldots,P_n$ running protocol $\pi$, where $k$ is the security parameter, $z$ is the input to $\mcal{Z}$, and $\bar{r}=(r_{\mcal{Z}},r_{\Adv},r_1,\ldots,r_n)$ specifies the randomness of all entities:\footnote{We remark that a quantum machine does not require a random tape, since it can generate randomness through quantum measurements as needed. Nevertheless, we retain random tapes for generality, allowing the model to accommodate settings in which some parties are classical machines.} $r_{\mcal{Z}}$ for the environment, $r_{\Adv}$ for the adversary, and $r_i$ for party $P_i$. Let $\mathrm{REAL}_{\pi,\Adv,\mcal{Z}}(k,z)$ denote the random variable induced by $\mathrm{REAL}_{\pi,\Adv,\mcal{Z}}(k,z,\bar{r})$ when $\bar{r}$ is chosen uniformly at random.

\para{The ideal process.} The security of a multiparty protocol is defined by comparing its execution in the real-world model with an ideal process for carrying out a single instance of the task under consideration. A central component of the ideal process is an ideal functionality that captures the desired behavior, or specification, of the task. The ideal functionality is modeled as a quantum polynomial-time (QPT) machine that interacts with the environment and the adversary according to the process described below. More specifically, the ideal process consists of an ideal functionality $\mcal{F}$, an ideal-process adversary (or simulator) $\mcal{S}$, an environment $\mcal{Z}$ with input $z$, and dummy parties $\widetilde{P}_1,\ldots,\widetilde{P}_n$.

In our setting, the input associated with a party $P_i$ is of the form $(x_i,\msf{st}_i)$, where $x_i$ denotes the logical input of $P_i$ and $\msf{st}_i \in \msf{T}$ specifies the trajectory along which $P_i$ evolves where $\msf{T}$ is the set of all admissible spacetime trajectories. Thus, upon receiving $(x_i,\msf{st}_i)$, the party processes the logical input $x_i$ while advancing along the trajectory $\msf{st}_i$.

As in a protocol execution in the real-world model, the environment is activated first. During each activation, it receives information from the dummy parties and the simulator and may send information to either a single dummy party or the simulator. In particular, the environment may provide a dummy party $\widetilde{P}_i$ with an input of the form $(x_i,\msf{st}_i)$, instructing it to process the logical input $x_i$ while advancing along the trajectory $\msf{st}_i$. Once the environment completes its activation, the entity to which it provided input is activated next.

The dummy parties are fixed, simple quantum polynomial-time machines. Whenever a dummy party is activated with an input, it forwards the corresponding information to the ideal functionality $\mcal{F}$. Upon completion of the dummy party's activation, the environment is activated. Dummy parties communicate only with $\mcal{F}$. In principle, messages exchanged between the dummy parties and $\mcal{F}$ are private and cannot be read by the simulator $\mcal{S}$. Each such message, however, consists of two components: a public header and private contents. The header is visible to $\mcal{S}$, whereas the contents are hidden from $\mcal{S}$. The specification of $\mcal{F}$ determines which information is included in the header and which is included in the contents. Certain information must necessarily appear in the public header. For example, the identity of a party to which $\mcal{F}$ sends an output must be public so that the output can be delivered. Beyond such necessary information, the specification of $\mcal{F}$ determines the precise division between the header and the contents.

For all functionalities considered in this work, headers follow a fixed format and specify the type of action, the session identifier, and the identities of the participating parties. For example, a commitment message may have the form $(\msf{commit},\mathit{sid},P_i,P_j,b)$, where $\msf{commit}$ specifies that a commitment is being made, $\mathit{sid}$ is the session identifier, $P_i$ is the committing party, $P_j$ is the receiving party, and $b$ is the committed value. In this example, the public header is $(\msf{commit},\mathit{sid},P_i,P_j)$, while the private contents consist solely of $b$.

Whenever the ideal functionality $\mcal{F}$ is activated, it processes the information it has received and may send messages to the dummy parties and the simulator. Once the activation of $\mcal{F}$ is complete, the environment $\mcal{Z}$ is activated next.

When the adversary $\mcal{S}$ is activated, it gets information from $\mcal{F}$. In contrast, $\mcal{S}$ cannot get the private contents of these messages (unless the recipient of the message of $\mcal{F}$ or a corrupted party). Likewise, $\mcal{S}$ can get the public headers of the message intended for $\mcal{F}$. Then, $\mcal{S}$ can execute one of the following actions. It may either send information to $\mcal{F}$'s, deliver a message between $\mcal{F}$ and $P_i$'s, or corrupt a party. Upon corrupting a party, both $\mcal{Z}$ and $\mcal{F}$ learn the identity of the corrupted party. In addition, the adversary learns all the information of the corrupted party. Finally, the adversary controls the party's action from the time that corruption takes place. 

If the adversary delivered a message to an uncorrupted dummy party $\tilde{P_i}$ or to the functionality $\mcal{F}$ in an activation, then this entity is activated next. Otherwise, the environment $\mcal{Z}$ is activated next.

As in the real-life model, the protocol execution ends when the environment do not send information to any entity. The output of the protocol execution is one-bit output of $\mcal{Z}$.

In summary, the order of activations in the ideal model is as follows. As in the real model, the environment $\mcal{Z}$ is always activated first, and then activates either the adversary $\mcal{A}$ or some dummy party $\tilde{P}_i$. If the adversary $\mcal{A}$ is activated, then it either activates a dummy party $\tilde{P}_i$ or the ideal functionality $\mcal{F}$ by delivering the entity a message, or it returns control to the environment. After the activation of a dummy party or the functionality, the environment is always activated next.

Let $\mathrm{IDEAL}_{\mcal{F},\mcal{S},\mcal{Z}}(k,z)$ denote the output of the environment $\mcal{Z}$ after interacting in the ideal process with adversary $\mcal{S}$ and ideal functionality $\mcal{F}$, on security parameter $k$, input $z$ and random input $\bar{r} = r_Z, r_S, r_F$, as described above, where $z$ and $r_Z$ are for $\mcal{Z}$, $r_S$ is for $\mcal{S}$, and $r_F$ is for $\mcal{F}$.  Let $\mathrm{IDEAL}_{\mcal{F},\mcal{S},\mcal{Z}}(k,z,\bar r)$ denote the random variable describing $\mathrm{IDEAL}_{\mcal{F},\mcal{S},\mcal{Z}}(k,z,\bar{r})$ when $\bar{r}$ is uniformly chosen.

\para{SUC security in the ideal process.} Having defined the real-life model and the ideal process augmented with spacetime trajectories, we now define what it means for a protocol $\pi$ to SUC-realize an ideal functionality $\mcal{F}$. Our definition combines computational indistinguishability between the real and ideal executions with a spatiotemporal verification requirement.

Let $\mcal{I}$ denote the set of identities of the corrupted parties. 
The ideal-process simulator consists of two components, $\mcal{S}=(\mcal{S}_0,\mcal{S}_1)$. During the execution, $\mcal{S}_0$ is first invoked and outputs strings $\Set{(x_i^*, \msf{st}_i^*)}_{i \in \mcal{I}}$ (they are supposed to be the extracted effective input for each corrupted party) and a residual state $\rho$ (meant to be used by $\mcal{S}_1$). We denote this procedure by 
$$\big(\Set{(x_i^*, \msf{st}_i^*)}_{i \in \mcal{I}}, \rho\big) \gets \mcal{S}_0(1^\secpar).$$

The component $\mcal{S}_1$ operates in two modes: $\mathtt{sim}$ mode and $\mathtt{prv}$ mode. The $\mathtt{sim}$ mode is activated immediately after $\mcal{S}_0$ produces its output. In this mode, $\mcal{S}_1(1^\secpar, \Set{(x_i^*, \msf{st}_i^*)}_{i \in \mcal{I}}, \rho)$ continues the ideal-world execution. It behaves analogously to an ideal-world simulator in the standard UC framework: it submits a value on behalf of each corrupted party, receives the corresponding output from $\mcal{F}$, and forwards that output to the environment. In addition, it must obey the following rule:
\begin{itemize}
    \item For each $i \in \mcal{I}$, it decides whether to submit $(x_i^*, \msf{st}_i^*)$ or $\bot$ on behalf of party $i$ through the prescribed ideal-process interface. We emphasize that $\mcal{S}_1$ cannot modify the extracted values $(x_i^*, \msf{st}_i^*)$; it may only decide whether to forward each pair unchanged to $\mcal{F}$ or to submit $\bot$ instead.
\end{itemize}
For every $i \in \mcal{I}$ and $\msf{st} \in \msf{T}$, we define
\begin{equation}\label{eq:define-delta-st}
    \delta(\secpar, i, \msf{st})
    \coloneqq
    \Pr\bigl[\text{$\mcal{S}_1$  in $\mathtt{sim}$ mode  sends  $(*,\msf{st})$ to $\mcal{F}$, in name of party $i$} \bigr],
\end{equation}
where $*$ is a wildcard and the probability is taken over all (classical and quantum) randomness in the entire ideal-world execution.

In the $\mathtt{prv}$ mode, $\mcal{S}_1$ is supposed to act as the prover in an associated spatiotemporal trajectory verification protocol $\Pi'$ for the extracted trajectory claims (see \Cref{def:uc-smpc}).

With this setup, we present the formal definition of SUC-realization in \Cref{def:uc-smpc}. In essence, SUC security requires that the real and ideal executions be computationally indistinguishable to every non-uniform QPT environment. In addition, it requires that non-negligible success in the simulation experiment imply non-negligible acceptance in the associated spatiotemporal trajectory verification.

\begin{definition}[SUC-realization]
    \label{def:uc-smpc}
    Let $\pi$ be an $n$-party protocol executed by $n$ parties with the common input $1^\secpar$. 
     Each party $P_i$ has an input $(x_i,\msf{st}_i)$, where $\msf{st}_i\in\msf{T}$ describes its spacetime trajectory. Let $\mcal{F}$ be an ideal functionality with the corresponding input and output interfaces. We say that $\pi$ \emph{SUC-realizes} $\mcal{F}$ if, for every non-uniform quantum polynomial-time real-world adversary $\mcal{A}$ and every non-uniform quantum polynomial-time environment $\mcal{Z}$, there exists a non-uniform quantum polynomial-time ideal-process simulator $\mcal{S}=(\mcal{S}_0,\mcal{S}_1)$, a spatiotemporal verification protocol $\Pi' = (\mcal{P}', \mcal{V}', \msf{T}')$ (as per \Cref{def:STV}), and a polynomial $\poly(\cdot)$ satisfying the aforementioned syntax such that the following conditions hold:
    \begin{itemize}
        \item  Let $\msf{S}^* = \Set{(L_i, T_i)}_i$ denote the set of spatiotemporal points that (some part of) $\mcal{A}$ \emph{has occupied} (see the discussion in \Cref{rem:malicious-committer-trajectory}) during the execution of the protocol. 
        \item No parts of $\mcal{S}=(\mcal{S}_0,\mcal{S}_1)$ can occupy any spatiotemporal point outside the set $\msf{S}^*$;
    \end{itemize}

    \begin{enumerate}
            \item \textbf{Indistinguishability.}  It holds that
            \begin{align*}
                & \left\{
                    \mathrm{REAL}_{\pi,\mcal{A},\mcal{Z}}(1^\lambda)
                \right\}_{\lambda\in\mathbb{N}} \\ 
                \cind ~ 
                &
                \left\{
                    \mathrm{IDEAL}_{\mcal{F},\mcal{S}_1(\mathtt{sim}, \Set{(x^*_i,\msf{st}^*_i)}_{i \in \mcal{I}}, \rho),\mcal{Z}}(1^\lambda)~:~ \big(\Set{(x^*_i,\msf{st}^*_i)}_{i \in \mcal{I}},\rho\big) \gets \mcal{S}_0(1^\lambda) 
                \right\}_{\lambda\in\mathbb{N}}.
            \end{align*}
  
        \item \textbf{Spatiotemporal verification.} \label[Condition]{cond:spacetime-verification} Let $\kappa(\secpar)$ be a knowledge error. For all $i \in \mcal{I}$ and all  $\msf{st}^*_i \in \msf{T}$, if
        $\delta(\secpar, i, \msf{st}^*_i) \ge \kappa(\secpar)$, then it holds that
        $$\Pr[b'=1~:b' \gets \langle \mcal{S}_1(\mathtt{prv}, \rho), \mcal{V}'\rangle(1^\lambda, \msf{st}^*_i)] = \poly\big(\delta(\secpar, i, \msf{st}^*_i)\big),$$  where $b' \gets \langle \mcal{S}_1(\mathtt{prv}, 1^\secpar,  \msf{st}^*_i, \rho), \mcal{V'} \rangle(1^\secpar, \msf{st}^*_i)$ denotes the execution of $\Pi'$ between $\mcal{S}_1(\mathtt{prv}, 1^\secpar)$ acting as the prover and the honest verifier $\mcal{V'}(1^\secpar,  \msf{st}^*_i)$, who outputs $b'$ after the interaction.  
            \end{enumerate}
    \end{definition}

\para{Composition.}
We state the spatiotemporal analogue of the standard UC composition theorem. This SUC composition theorem will later be used to reduce the security of \Cref{UC-ST-MPC} in the real model to its security in the corresponding hybrid model.

\begin{theorem}[Spatiotemporal Universal Composition Theorem ]\label{thm:suc-comp} 
Let $\Pi$ be multiple-party protocol in the $(\mcal{F}_1,\ldots,\mcal{F}_m)$-hybrid model, where $\mcal{F}_1,\ldots,\mcal{F}_m$ are ideal functionalities used as subroutines by $\Pi$. Suppose that $\Pi$ SUC-realizes a spacetime-aware functionality $\mcal{F}$ in the $(\mcal{F}_1,\ldots,\mcal{F}_m)$-hybrid model.

For each $\ell\in[m]$, let $\pi_\ell$ be a protocol that SUC-realizes $\mcal{F}_\ell$. Assume that the composed protocol is subroutine respecting and subroutine exposing. Let
\[
    \Pi[\pi_1,\ldots,\pi_m]
\]
denote the protocol obtained from $\Pi$ by replacing each ideal call to $\mcal{F_\ell}$ with an execution of $\pi_\ell$.

Then $\Pi[\pi_1,\ldots,\pi_m]$ SUC-realizes $\mcal{F}$ in the real model.
\end{theorem}

\begin{proof}[Proof (Sketch)] This proof is a rather straightforward adaptation of the proof techniques for the standard (quantum) UC composition theorem \cite{FOCS:Canetti01,EC:Unruh10}. We thus only provide a sketch of the argument below.

    For the indistinguishability condition, the SUC framework preserves the interface that allows the environment to interact with the adversary. Hence, the standard UC composition argument applies, and indistinguishability is preserved under subroutine replacement.

    For the spatiotemporal verification condition, replacing a cryptographic subroutine introduces neither a forbidden communication path nor additional pre-shared entanglement into the position-verification reduction. Therefore, the resulting adversarial strategy remains admissible in the underlying physical model, and the spatiotemporal-verification condition is preserved.

\end{proof}

\subsection{SUC Spatiotemporal Commit-and-Prove}
In this section, we first provide a formal definition of UC spatiotemporal  commit-and-prove and then present its construction.

\subsubsection{Definition}

We consider a functionality with the following properties:
\begin{itemize}
    \item It consists of two phases: a \textbf{Commit} phase and a \textbf{Prove} phase.
    \item During the \textbf{Commit} phase, the prover moves along a trajectory $\msf{st}$, which it holds as private input, and commits to $\msf{st}$.
    \item For a public predicate $f(\cdot)$, the \textbf{Prove} phase consists of a zero-knowledge argument for the statement $f(\msf{st})=1$.
\end{itemize}

\para{The ideal functionality.}
Let $\msf{T}$ be the set of admissible spacetime trajectories. The ideal functionality $\mathcal{F}^\msf{ST}_{\msf{CnP}}$ interacts with parties $P$ and $V$. $P$ who moves along a trajectory $\msf{st}$ commits to the spatiotemporal trajectory $\msf{st}$ with $\msf{st}\in \msf{T}$. $V$ zero-knowledge learns only whether the $P$ moves along the trajectory $\msf{st}$ and the corresponding committed values satisfy $f(\msf{st})=1$.

\begin{FigureBox}[label={F-ST-CnP}]{Ideal Functionality $\mathcal{F}^\msf{ST}_{\msf{CnP}}$}

The functionality maintains an initially empty table
$\mathcal{D}$.

\medskip
\noindent
\textbf{Commit phase.} Upon receiving $(\mathtt{commit},\msf{st},\msf{idx})$ from $P$, the functionality records $(\msf{st},\msf{idx})$ in $\mathcal{D}$ and sends $(\mathtt{commit\text{-}done},\msf{idx})$ to $V$.

\medskip
\noindent
\textbf{Prove phase.} Let $f$ be a predicate provided as common input. Upon receiving $(\mathtt{prove}, \msf{idx})$ from $P$, the functionality sends $(\mathtt{prove\text{-}ok},\msf{idx})$ to $V$ if there exists some $\msf{st}$ such that $(\msf{st},\msf{idx})\in\mathcal{D}$ and $f(\msf{st})=1$; otherwise, it sends $(\mathtt{prove\text{-}failed},\msf{idx})$ to $V$.
\end{FigureBox}

\begin{remark}[Commitments to multiple trajectories]\label{remark:st-cnp-multiple}
    The prover may commit to multiple trajectories by invoking the \textbf{Commit phase} repeatedly, using a distinct index for each trajectory. To jointly prove statements about any selected collection of committed trajectories, the prover could a single \textbf{Prove phase} to prove that it knows valid openings of all the selected commitments and that the corresponding trajectories jointly satisfy the prescribed predicates. 
    \end{remark}

\subsubsection{Construction}
We construct a protocol realizing the ideal functionality $\mathcal{F}^\msf{ST}_{\msf{CnP}}$ by modifying the construction in
\Cref{protocol:stc}. The construction uses the following building blocks:
\begin{itemize}
    \item A ``nice'' position-verification protocol $\Pi_{\mathsf{PV}}$ satisfying \Cref{def:nice}. For each spacetime point $\alpha$, let $W_{\alpha}$ denote the verification predicate.

    \item A UC-secure commitment scheme $\msf{UcCom} = \langle P,V\rangle$, where $P$ acts as the committer and $V$ as the receiver, satisfying \Cref{def:uc-commit}.

    \item A UC-secure zero-knowledge proof system $\Pi_{\msf{UC\text{-}ZK}}$ satisfying \Cref{def:uc-zk}.
\end{itemize}

\begin{ProtocolBox}[label={UC-ST-CnP}]{SUC commit-and-prove of spatiotemporal trajectory}
\para{Setup.} This phase is identical to the \textbf{Setup} phase of \Cref{protocol:stc}. The prover $P$ plays the role of the committer $C$, while the verifier $V$ plays the role of the receiver $R$.

\para{Private Input:} The prover $P$ follows a trajectory $\msf{st}=\bigl((L_1,T_1),(L_2,T_2),\ldots,(L_m,T_m)\bigr)\in\msf{T}$, which it holds as its private input.

\para{Commit Phase.} This phase proceeds identically to the \textbf{Commit Phase} of \Cref{protocol:stc}, except that $\msf{ExtCom}$ and $\Pi_{\msf{ZK}}$ are replaced by $\msf{UcCom}$ and $\Pi_{\msf{UC\text{-}ZK}}$, respectively. The prover $P$ plays the role of the committer $C$, while the verifier $V$ plays the role of the receiver $R$.

\para{Prove Phase.} The prover $P$ and verifier $V$ receive a polynomial-time computable predicate $f$. They then execute the $\Pi_{\msf{UC\text{-}ZK}}$ protocol in which $P$ proves that the committed trajectory $\msf{st}$ satisfies $f$; that is, $f(\msf{st})=1$.

\para{Output:} The verifier $V$ accepts if and only if neither phase aborts and the zero-knowledge proof is accepted. Otherwise, $V$ rejects.

\end{ProtocolBox}
\begin{theorem}[SUC Realization of \Cref{UC-ST-CnP}]
    \label{thm:st-cnp}
    Assume that $\Pi_{\mathsf{PV}}$ is a ``nice'' position-verification protocol, as specified in \Cref{def:nice}, that is secure against unentangled (respectively, entangled) adversaries. Further, the $\Pi_{\msf{UC\text{-}ZK}}$ satisfying \Cref{def:uc-zk}, and the $\msf{UcCom}$ satisfying \Cref{def:uc-commit}. Then the protocol in \Cref{UC-ST-CnP} SUC-realizes the ideal functionality \Cref{F-ST-CnP} according to \Cref{def:uc-smpc} against unentangled (respectively, entangled) QPT adversaries.
\end{theorem}
    
\begin{proof}
We consider the cases of a corrupted prover and a corrupted verifier separately.

\para{Corrupted prover.} The security of the this follows directly from the extractability argument in the proof of \Cref{thm:security-stc}. Combining this argument with the composition theorem in \Cref{thm:suc-comp} establishes the desired security claim.

\para{Corrupted verifier.} We establish the security of the commit phase and the prove phase separately. For each case, we construct an ideal-world simulator $\mathcal{S}=(\mathcal{S}_0,\mathcal{S}_1)$ such that, for every non-uniform QPT real-world adversary $\mathcal{A}$ and every non-uniform QPT environment $\mathcal{Z}$, the output of $\mathcal{Z}$ in the real execution is computationally indistinguishable from its output in the ideal execution with $\mathcal{S}$ and $\mathcal{F}_\msf{CnP}^\msf{ST}$. The simulator internally runs $\mathcal{A}$ and simulates the interfaces and messages of the honest parties. 
\begin{itemize}
    \item  \textbf{Commit Phase.} 
    \begin{enumerate}
        \item $\mathcal{S}_0$ sets the extracted trajectory to $\msf{st}^*=\bot$ and outputs the corresponding residual state $\rho$ needed to continue the execution.
        \item Since $\msf{st}^*=\bot$, $\mathcal{S}_1$ does not forward information to the $\mcal{F}$. It then internally runs $\mathcal{A}$ and emulates the honest prover $P$. Upon receiving $(\msf{commit\text{-}done},\msf{idx})$ from $\mathcal{F}_\msf{CnP}^\msf{ST}$, indicating that the honest prover has committed under index $\msf{idx}$, $\mathcal{S}_1$ samples an independent uniformly random key $\Set{\widetilde{\msf{sk}}_{v, u}}_{v \in [n], u \in [|\msf{L}|]}\sample\{0,1\}^{\ell(\lambda)}$ and selects an arbitrary fixed dummy trajectory $\widetilde{\msf{st}}\in\mathsf{T}$ of the appropriate format.
        \item $\mathcal{S}_1$ invokes the simulator of the underlying $\msf{UcCom}$ scheme to generate a simulated commitment $\widetilde{\msf{com}}$ corresponding to $(\Set{\widetilde{\msf{sk}}_{v, u}}_{v \in [n], u \in [|\msf{L}|]},\widetilde{\msf{st}})$ and sends $\widetilde{\msf{com}}$ to $\mathcal{A}$ according to the message schedule prescribed by the protocol.
        \item $\mathcal{S}_1$ pretends to be any possible location and returns all dummy ciphertexts $\widetilde{c}_{u,v}\gets\widetilde{\msf{sk}_{u,v}}\oplus\bot$, using an appropriate fixed-length encoding of $\bot$. These ciphertexts are delivered with the same lengths, number, ordering, and timing as the ciphertexts in an honest execution. And then give a zero-knowledge proof of doing a valid behavior.
        \item At the end of the commit phase, $\mathcal{S}_1$ outputs $\widetilde{\rho}=(\widetilde{\msf{M}},\widetilde{\msf{com}}, \widetilde{\pi})$.
    \end{enumerate}

    \subpara{Indistinguishability.} We show that no QPT environment $\mathcal{Z}$ can distinguish the real commit phase execution from the ideal execution generated by $\mathcal{S}_1$. In the real execution, each position verification response is masked by a fresh uniformly random one-time pad; hence, replacing the encrypted responses, one at a time, with encryptions of $\bot$ preserves their distribution. We then replace the honest commitment with a simulated commitment to an independently sampled dummy pair. By the property of the UC-secure commitment scheme, UC-secure zero-knowledge and \Cref{thm:suc-comp}, this replacement is computationally indistinguishable to $\mathcal{Z}$. The resulting hybrid is distributed exactly as the ideal execution generated by $\mathcal{S}_1$. Therefore, for every QPT environment, it can distinguish the real commit phase execution from the ideal execution with only negligible probability.

   \item \textbf{Prove Phase.} The simulator proceeds as follows.
   \begin{enumerate}
    
    \item
    $\mathcal{S}_1$ invokes the simulator of the UC-secure zero-knowledge proof system to generate a simulated accepting proof for the prove phase statement associated with the transcript produced during the simulated commit phase. In particular, $\mathcal{S}_1$ need not know the honest trajectory, secret key, or decommitment information. It interacts with $\mathcal{A}$ according to the prescribed message schedule and provides it with the resulting simulated proof.

    \item Finally, $\mathcal{S}_1$ records the resulting adversarial state and outputs the view of $\mathcal{A}$ together with the prove-phase outcome. By construction, this outcome agrees with that prescribed by $\mathcal{F}_\msf{CnP}^\msf{ST}$.
    \end{enumerate}

    \subpara{Indistinguishability.} Consider a hybrid in which the real zero-knowledge proof generated using the honest witness is replaced by a simulated proof. By the UC zero-knowledge property of the proof system and \Cref{thm:suc-comp}, this replacement is computationally indistinguishable to the corrupted verifier and any QPT environment. Combining this property with the indistinguishability of the simulated commitment and encrypted responses established for the commit phase, we conclude that the corrupted verifier's joint view across the commit and prove phases is computationally indistinguishable from its view in the real execution.

\end{itemize}
Also, note that in the corrupted-verifier case, $\mcal{S}_1$ do not need to submit any information to the functionality. No spatiotemporal verification argument is required.

This finishes the proof of \Cref{thm:st-cnp}.

\end{proof}

Following \Cref{sec:instantiation}, we obtain the following instantiations.

\para{Instantiation.} The position-verification protocol of \cite{ITCS:LiuLiuQia22} is secure against unentangled adversaries, whereas the protocol of \cite{C:Unruh14} is secure against entangled adversaries in the QROM; moreover, both satisfy \Cref{def:nice}. The required UC-secure commitment scheme and zero-knowledge proof system can also be instantiated either in the CRS model under the LWE assumption \cite{C:DFLSS09,EC:ABKK23} or in the QROM \cite{EC:ABKK23}. We therefore obtain the following corollary.

\begin{corollary}\label{corollary:cnp}
Depending on the adversarial model and setup assumptions, the following instantiations exist:
\begin{itemize}
    \item In the CRS model, assuming QLWE, there exists an SUC-secure commit-and-prove scheme against unentangled QPT adversaries.
    \item In QROM, there exists an SUC-secure commit-and-prove scheme against entangled QPT adversaries.
\end{itemize}
\end{corollary}

\begin{remark}[Local vs. Global Setup for UC Security]
  \label{rmk:local-global-setup}
For UC security with setups like CRS or ROM, we distinguish between local and global setup, following
the treatment of globally shared setup
in~ \cite{TCC:CDPW07,CCS:CanJaiSca14,EC:CDGLN18}.
In the local-setup model considered here, the setup is
a subroutine of the protocol instance under consideration;
in a global-setup model, it is also accessible to the
environment and other protocol instances. 
The limitation of local setup is that the resulting security
guarantee does not automatically extend to arbitrary external
protocols sharing the same CRS or random oracle. In particular,
a proof using independent local setup instances does not,
by itself, justify replacing them with a single globally
shared setup.

Nevertheless, this restriction does not reduce our goal
to stand-alone security: quantum UC security still supports
modular composition and concurrent execution, provided
the composition respects the modeled setup
boundaries~\cite{EC:Unruh10}.
Achieving this guarantee remains nontrivial, since the
simulator must maintain an indistinguishable ongoing
interaction with a quantum environment. In our setting,
it must additionally support certification of the extracted
spatiotemporal claims, as described above. Thus, local setup
provides a meaningful setting in which to establish
composable spatiotemporal security, while leaving
unrestricted setup sharing outside our present scope.
We leave extending our construction to an explicitly
specified global-setup model for future work.

\end{remark}

\subsection{Construction: Spatiotemporal MPC}
In this section, we present a construction for spatiotemporal multiparty computation (MPC). Each party $P_i$ holds a classical input $x_i$ and follows a spatiotemporal trajectory $\msf{st}_i$. The parties jointly evaluate the classical functionality
\[
    F((x_1,\msf{st}_1),\ldots,(x_n,\msf{st}_n)),
\]
while ensuring that the trajectory associated with each MPC message is physically verified at the time the message is transmitted.

The protocol consists of the following two phases:
\begin{itemize}
    \item \textbf{Commit phase.} Each party commits to its trajectory.
    
    \item \textbf{Compute phase.} The parties execute the MPC protocol in conjunction with the UC-secure commit-and-prove protocol for spatiotemporal trajectories. The MPC protocol evaluates the desired functionality, while the commit-and-prove protocol verifies the physical validity of the trajectories associated with the computation.
\end{itemize}

Because the MPC protocol may consist of multiple communication rounds, spatiotemporal trajectory verification is performed separately in each round. More precisely, whenever a party sends an MPC message, it must prove the physical validity of the corresponding portion of its trajectory at the time of transmission. The protocol must additionally guarantee that all such per-round trajectories are consistent with the trajectories to which the parties committed during the commit phase. To formalize this consistency requirement, we utilize a commitment-consistent MPC functionality.

\para{Commitment-consistent PQ-MPC ($\mathsf{CC\text{-}MPC}$).}
Let $C=(C_1,\ldots,C_n)$ be the public commitment vector, where $C_i=\mathsf{Com}(\msf{st}_i;r_i)$. The commitment-consistent MPC functionality $\mathcal{F}_ {\mathsf{CC\text{-}MPC}}$ receives from every party $P_i$ a private input $(x_i,\msf{st}_i,\mathsf{decom}_i)$. It checks that
\[
    \mathsf{Verify}(C_i,\msf{st}_i,\mathsf{decom}_i)=1
\]
for every $i\in[n]$. If any check fails, it outputs $\mathsf{abort}$. Otherwise, it outputs
\[
    F((x_1,\msf{st}_1),\ldots,(x_n,\msf{st}_n)).
\] 

Thus, the functionality guarantees that the trajectories used as inputs to the MPC computation are exactly those fixed by the public commitments.

\para{Round structure and Time Framing.} Recall that, in our treatment to  position/trajectory verification, we assume that local computations are performed instantaneously. This assumption is reasonable because some position/trajectory-verification protocols require only simple local computations, such as measuring qubits in the computational or Hadamard basis. Moreover, even when local computations are not instantaneous, one can typically adjust the verifiers' decision function to account for their duration, under a reasonable, commonly agreed-upon upper bound on the local computation time.

For MPC protocols, however, we choose to make the local computation time explicit for the following reasons: Depending on the target functionality, the local computations in an MPC protocol may be considerably more involved than those in a position/trajectory-verification protocol, and their duration may therefore be significant. Furthermore, modifying the decision function to account for local computation time may become increasingly complicated as the number of parties grows and the network topology becomes more complex. We therefore explicitly account for local computation time in our definition of spatiotemporal MPC.

To this end, we assume that each round $r$ of an MPC protocol is executed within a time interval $(t_{r-1},t_r]$ agreed upon by all parties at the beginning of the protocol. This time interval is further divided int two sub-intervals using a parameter $\Delta_{\msf{proof}}$, which denote the execution time of the prove stage of the trajectory commit-and-prove protocol. Looking forward in \Cref{UC-ST-MPC}:
\begin{itemize}
    \item

During the time interval $(t_{r-1},t_r - \Delta_{\msf{proof}})$, each party performs its local computation for round $r$. At time $t_r - \Delta_{\msf{proof}}$, all parties send their messages to the designated recipients or broadcast them over the public channel. In particular, even if a party completes its local round-$r$ computation before $t_r - \Delta_{\msf{proof}}$, it does not send its messages until time $t_r - \Delta_{\msf{proof}}$. This convention ensures that all round-$r$ messages are sent simultaneously, allowing their spatiotemporal verification to be performed at a common time point.

\item During the time interval $(t_r - \Delta_{\msf{proof}},t_r]$, each party performs the prove stage of the trajectory commit-and-prove protocol to verify the physical validity of its spatiotemporal trajectory during $(t_{r-1} - \Delta_{\msf{proof}},t_r - \Delta_{\msf{proof}}]$. 
\end{itemize}

We will utilize an $L$-round post-quantum MPC protocol $\msf{CC\text{-}MPC}$ that UC-realize the functionality  $\mathcal{F}_{\mathsf{CC\text{-}MPC}}$ defined above. The messages sent in round $r \in [L]$ may depend only on the parties' inputs, local randomness, and messages received in rounds $1,\ldots,r-1$. We associate the $L$ rounds with a strictly increasing sequence of times
\(
    t_0 < t_1 < \cdots < t_L,
\)
and define 
\begin{equation}\label{eq:round-interval}
I_{r, \msf{comp}} \coloneqq(t_{r-1},t_r - \Delta_{\msf{proof}}], ~~~ I_{r, \msf{proof}} \coloneqq(t_r - \Delta_{\msf{proof}},t_r]
\end{equation}  
as the time intervals corresponding to round $r$. 

We assume that the interval $[t_0,t_L]$ is a superset of the entire time span of all parties' spatiotemporal trajectories. This is necessary because, as explained in \Cref{sec:overview-cnp}, a spatiotemporal computation providing physical authentication guarantees cannot compute a future time point that lies beyond the duration of the protocol execution. This assumption can be achieved even if the original protocol $\msf{CC\text{-}MPC}$ can be executed over a shorter time interval: we can deliberately introduce a delay into each interval $I_r$, thereby extending $[t_0,t_L]$ to cover the entire time span of all parties' spatiotemporal trajectories. Specifically, the delay can be implemented by requiring each party to wait for a predetermined amount of time before sending its messages in each round.

For each party $P_i$ that sends a message in round $r$, let
\begin{equation}\label{eq:st-restriction}
    \msf{st}_{i,r, \msf{comp}}\coloneqq\msf{st}_i|_{I_{r, \msf{comp}}}, ~~~~ \msf{st}_{i,r, \msf{proof}}\coloneqq\msf{st}_i|_{I_{r, \msf{proof}}}
\end{equation}
denote the restriction of its trajectory to $I_{r, \msf{comp}}$ and $I_{r, \msf{proof}}$, respectively. The consistency predicate for round $r$ ensures that the trajectory segment used in the corresponding commit-and-prove instance is precisely the restriction of the trajectory committed by $P_i$ during the initial commit phase. Consequently, all per-round physical verifications are tied to a single globally committed trajectory.

\paragraph{Building blocks.}
The construction uses the following building blocks:
\begin{itemize}
    \item a post-quantum commitment scheme $\msf{Com}$ (e.g., Naor's commitment).
    \item a commit-and-prove of spatiotemporal trajectory protocol $\mathsf{TCnP}$ that SUC-realizes the functionality \Cref{F-ST-CnP}. Let $\Delta_{\msf{proof}}$ denote the maximum time required to complete the prove stage of $\mathsf{TCnP}$.
    \item a post-quantum MPC protocol $\mathsf{CC\text{-}MPC}$ in $L$ rounds that UC-realizes the classical functionality $\mathcal{F}_{\mathsf{CC\text{-}MPC}}$.
\end{itemize}

\begin{ProtocolBox}[label={UC-ST-MPC}]{SUC Spatiotemporal MPC}

\para{Setup.} The protocol is executed by $n$ parties and is parameterized by a security parameter $1^\lambda$ and a nonempty set $\msf{T}$ of admissible spatiotemporal trajectories. 
The execution of $\msf{CC\text{-}MPC}$ spans the entire time interval $[t_{\msf{min}}(\msf{T}),t_{\msf{max}}(\msf{T})]$. 

Each party $P_j$ is assisted by two coordinating servers that act as verifiers and perform the required spatiotemporal verification and computation. One of these servers is designated as the main server; it verifies the spatiotemporal proofs submitted by the other parties and exchanges $\msf{CC\text{-}MPC}$ messages with their main servers. Meanwhile, $P_j$ moves along a specific trajectory and acts as a prover, proving the validity of its spatiotemporal trajectory to the other parties without revealing the trajectory itself. Consequently, the exchange of $\msf{CC\text{-}MPC}$ messages reveals no spatiotemporal information about $P_j$.

\para{Private input:} The protocol runs by $n$ parties $\Set{P_j}_{j\in n}$. Each party $P_j$ holds an input $x_j\in\{0,1\}^*$ and moves along with a claimed spatiotemporal trajectory $\msf{st}_j \in \msf{T}$.

\para{Phase 1: Commitment.}
The parties interact so that each party $P_j$ commits to its  trajectory $\msf{st}_j$ using the post-quantum commitment scheme $\msf{Com}$. We denote the commitment generated by $P_j$ as $C_j = \mathsf{Com}\bigl(\msf{st}_j;r_j\bigr)$. Let $C=(C_1,\ldots,C_n)$ denote the public commitment vector. Each party $P_j$ stores its own trajectory $\msf{st}_j$ and the corresponding decommitment information $\mathsf{decom}_j$ for later use in the MPC execution.

\para{Phase 2: Concurrent MPC.}
The underlying protocol $\mathsf{CC\text{-}MPC}$ proceeds according to its prescribed rounds $r=1,\ldots,L$. At each round $r$, the protocol proceeds according to the following interval schedule:
\begin{itemize}

        \item  \textbf{(PQ-MPC execution for round $r$.)}  During the time interval $(t_{r-1},t_r - \Delta_{\msf{proof}})$: each party $P_j$'s main server computes $P_j$'s  round-$r$ message $m^{(j)}_r$ based on its input $x_j$, local randomness, and the messages received in all previous rounds. The messages are sent at time point $t_r - \Delta_{\msf{proof}}$ to the designated main servers.

    In parallel, during the round-$r$ computation time interval $I_{r, \msf{comp}}$ as defined in \Cref{eq:round-interval}, each party $P_j$ acts the committer to execute the commit stage of a $\mathsf{TCnP}$ instance, committing to the trajectory segment $\msf{st}_{j,r, \msf{comp}}$ defined in \Cref{eq:st-restriction}.

        \item \textbf{(TCnP prove phase for round $r$.)}  During the time interval $(t_r - \Delta_{\msf{proof}},t_r]$: each party $P_j$ acts as the prover to execute the prove stage of the $\mathsf{TCnP}$ instance for round $r$, proving that its trajectory segments $\msf{st}_{j,r-1, \msf{proof}}$ and  $\msf{st}_{j,r, \msf{comp}}$, committed using the commit phase of $\mathsf{TCnP}$ during $(t_{r-1} - \Delta_{\msf{proof}}, t_{r-1}]$ and $(t_{r-1},t_r - \Delta_{\msf{proof}}]$ respectively, are consistent with the trajectory $\msf{st}_j$ committed during the initial commit phase. The main server of each party verifies the proof and, if accepted, allows the corresponding round-$r$ message $m^{(j)}_r$ to be used in subsequent steps. If any proof fails, the protocol aborts.
        
        In parallel, during this round-$r$ (for $r=1,\ldots,L-1$) proof time interval $I_{r, \msf{proof}}$ as defined in \Cref{eq:round-interval}, each party $P_j$ acts the committer to execute the commit stage of a $\mathsf{TCnP}$ instance, committing to the trajectory segment $\msf{st}_{j,r, \msf{proof}}$ defined in \Cref{eq:st-restriction}. 
        
        We note that when $r=L$, we do not need to run the commit sage to commit to $\msf{st}_{j,L, \msf{proof}}$. This is because the last round message $m^{(j)}_L$ is already sent at time $t_L - \Delta_{\msf{proof}}$, representing the end of the real computation. No $st_j$ should contain any $(L, T)$ with $T>t_L - \Delta_{\msf{proof}}$.

\end{itemize}

\para{Output.}
If all required $\mathsf{TCnP}$ instances accept and $\mathsf{CC\text{-}MPC}$ terminates with local output $y_j$, party $P_j$ outputs $y_j$. Otherwise it outputs $\mathsf{abort}$.
\end{ProtocolBox}

\begin{theorem}[SUC-realization of~\Cref{UC-ST-MPC}]\label{thm:sucmpc}
    Let $\mathcal{F}_{\msf{MPC}}^{\msf{ST}}$ denote the spatiotemporal MPC functionality that computes $F\bigl((x_1,\msf{st}_1),\ldots,(x_n,\msf{st}_n)\bigr).$ Suppose that $\msf{TCnP}$ SUC-realizes the ideal functionality in \Cref{UC-ST-CnP} and that $\msf{CC\text{-}MPC}$ UC-realizes the functionality $\mathcal{F}_{\mathsf{CC\text{-}MPC}}$. Then the protocol in \Cref{UC-ST-MPC} SUC-realizes $\mathcal{F}_{\msf{MPC}}^{\msf{ST}}$.

\end{theorem}

\begin{proof}[Proof]
\subpara{Indistinguishability.} We prove indistinguishability by a sequence of hybrid games. The first hybrid $\mathcal{H}_0$ is the real execution of $\Cref{UC-ST-MPC}$, while the final hybrid $\mathcal{H}_3$ is the ideal execution with $\mathcal{F}_\msf{MPC}^\msf{ST}$. In every hybrid, the environment $\mathcal{Z}$ interacts with the same externally visible interfaces, and an abort in any invoked subprotocol causes the entire execution to abort. In each hybrid, we describe the differences from the previous hybrid.

\begin{itemize}
    \item \textbf{Hybrid $H_0$}. This is the real execution of Protocol~\ref{UC-ST-MPC}.
    \item \textbf{Hybrid $H_1$}. In this hybrid, we replace the execution of $\mathsf{CC\text{-}MPC}$ is replaced by an invocation of $\mathcal{F}_{\mathsf{CC\text{-}MPC}}$. Phase~1 and every $\mathsf{TCnP}$ instance remain real.
    \item \textbf{Hybrid $H_2$}. 
    In this hybrid, we additionally replace each round-$r$ instance of $\mathsf{TCnP}$ with an invocation of $\mathcal{F}^{\msf{ST}}_{\msf{CnP}}$. Upon receiving a prove request for predicate $g_r$ corresponding to the statement proven by the round-$r$ proof, the functionality accepts only if the committed trajectory segment is physically valid and satisfies $g_r$. In particular, any accepted segment must be consistent with the trajectory committed by the sender during the commit phase.
    
    \item \textbf{Hybrid $H_3$}. This is the ideal execution with $\mathcal{F}_\msf{MPC}^\msf{ST}$. 
\end{itemize}

We now prove indistinguishability between each pair of consecutive hybrids.

\begin{itemize}
    \item $\mathcal H_0 \cind \mathcal H_1$:
    The only difference between $\mathcal H_0$ and $\mathcal H_1$ is that the real execution of $\mathsf{CC\text{-}MPC}$ is replaced by an invocation of $\mathcal F_{\mathsf{CC\text{-}MPC}}$ together with its UC simulator. Since $\mathsf{CC\text{-}MPC}$ UC-realizes $\mathcal F_{\mathsf{CC\text{-}MPC}}$, and the surrounding protocol satisfies the required subroutine-respecting and subroutine-exposing conditions, \Cref{thm:suc-comp} implies that this replacement remains computationally indistinguishable in the complete spatiotemporal execution. Hence, $\mathcal H_0 \cind \mathcal H_1.$

    \item $\mathcal{H}_1 \cind \mathcal{H}_2$: The two hybrids differ only in that each real instance of $\mathsf{TCnP}$ in $\mathcal{H}_1$ is replaced in $\mathcal{H}_2$ by an invocation of $\mathcal{F}_{\mathsf{ST\text{-}CnP}}$ together with the corresponding SUC simulator. Consider a sequence of intermediate hybrids that replaces the $\mathsf{TCnP}$ instances one at a time. By the SUC security of $\mathsf{TCnP}$ and the universal composition guarantee of \Cref{thm:suc-comp}, every pair of adjacent hybrids is computationally indistinguishable. Since the protocol invokes only polynomially many $\mathsf{TCnP}$ instances, a standard hybrid argument yields $\mathcal{H}_1 \cind \mathcal{H}_2$.
    
    \item $\mathcal{H}_2\cind\mathcal{H}_3$: We construct an ideal-world simulator $\mathcal{S}=(\mathcal{S}_0,\mathcal{S}_1)$ such that the ideal execution with $\mathcal{F}_\msf{MPC}^\msf{ST}$ is computationally indistinguishable from $\mathcal{H}_2$. Let $\mathcal{A}$ be the adversary in $\mathcal{H}_2$. The simulator $\mathcal{S}$ runs an internal copy of $\mathcal{A}$, simulates the public interfaces of the $\mathcal{H}_2$ execution, and forwards all messages between $\mathcal{Z}$ and the internal copy of $\mathcal{A}$. It also records the public commitment vector $C=(C_1,\ldots,C_n)$ generated in Phase~1. At this stage, $\mathcal{S}$ neither learns nor uses the inputs of the honest parties; the commitment $C_i$ serves only to bind the trajectory used by $P_i$ in the MPC execution to the trajectory segments subsequently verified by the round-specific instances of $\mathcal{F}_\msf{CnP}^\msf{ST}$.

For every corrupted party $P_i\in\mcal{I}$, the simulator proceeds as follows:
\begin{enumerate}
    \item When the internal adversary supplies $(x_i^*,\mathsf{st}_i^*,\mathsf{decom}_i^*)$ to $\mathcal{F}_{\mathsf{CC\text{-}MPC}}$, the simulator emulates the honest parties and checks whether $\mathsf{Verify}(C_i,\mathsf{st}_i^*,\mathsf{decom}_i^*)=1.$
    
    If this check fails, $\mathcal{S}_1$ simulates the corresponding abort. Otherwise, it records $(x_i^*,\mathsf{st}_i^*)$ as the input of $P_i$.

    \item Whenever the internal adversary invokes $\mathcal{F}^\msf{ST}_{\msf{CnP}}$ in round $r$, $\mathcal{S}_0$ extracts and records the trajectory segment submitted by $\mcal{A}$. The component $\mathcal{S}_1$ then checks the outcome of the ideal functionality and records the extracted trajectory segment if it is accepted. By the definition of the predicate $g_r$, every accepted segment is the corresponding portion of a trajectory opening the same Phase~1 commitment $C_i$. Since $\mathsf{Com}$ is computationally binding, except with negligible probability, all accepted segments are consistent with the trajectory $\mathsf{st}_i^*$ supplied to $\mathcal{F}_{\mathsf{CC\text{-}MPC}}$ and with the unique trajectory committed in $C_i$.
\end{enumerate}

Once the required trajectory segments have been accepted and jointly cover the prescribed execution interval, $\mathcal{S}_1$ submits $(x_i^*,\mathsf{st}_i^*)$ to $\mathcal{F}_\msf{MPC}^\msf{ST}$ on behalf of $\mcal{A}$. If any required invocation of $\mathcal{F}_\msf{CnP}^\msf{ST}$ rejects, or if the submitted trajectory fails to open $C_i$, the simulator reproduces the corresponding abort. Otherwise, the corrupted-party inputs submitted to $\mathcal{F}_\msf{MPC}^\msf{ST}$ are exactly those accepted by $\mathcal{F}_{\mathsf{CC\text{-}MPC}}$ in $\mathcal{H}_2$, while the honest parties supply their inputs directly to the ideal functionality. Consequently, the output returned by $\mathcal{F}_\msf{MPC}^\msf{ST}$ is identical to the output returned by $\mathcal{F}_{\mathsf{CC\text{-}MPC}}$ in $\mathcal{H}_2$, and $\mathcal{S}_1$ forwards this output to the internal adversary.

It follows that the output and abort behavior in $\mathcal{H}_3$ agree with those in $\mathcal{H}_2$, except with negligible probability arising from a violation of the computational binding security of $\mathsf{Com}$. Therefore, $\mathcal{H}_2\cind\mathcal{H}_3.$

\end{itemize}

Combining the hybrid transitions, we obtain the indistinguishability of the real and ideal executions.

\subpara{Spacetime verification.} Building on the proof that $\mathcal{H}_2 \cind \mathcal{H}_3$, we establish the \Cref{cond:spacetime-verification} of \Cref{thm:sucmpc}. Consider any tuple $(\msf{st}_i^*,x_i^*,\rho)$ extracted by the simulator $\mcal{S}_0$. Whenever $\mcal{S}_1$ forwards $(\msf{st}_i^*,x_i^*)$ to the ideal functionality, $\mcal{F}_{\msf{CnP}}^{\msf{ST}}$ must have received the segment of $(\msf{st}_i^*)$ from its internal simulator. By the definition of $\mcal{F}_{\msf{CnP}}^{\msf{ST}}$, we may invoke its simulator $\mcal{S}_2^i$ on each segment of $\msf{st}_i^*$, thereby verifying every corresponding segment. Consequently, if $\mcal{S}_1$ forwards $(\msf{st}_i^*,x_i^*)$ to the ideal functionality with probability $\delta(\lambda)$, then the extracted trajectory $\msf{st}_i^*$ satisfies the required spacetime-verification condition with probability $\delta(\lambda)$. This completes the proof.

\end{proof}

\para{Instantiation.} By \Cref{corollary:cnp}, an SUC-secure commit-and-prove scheme can be instantiated against unentangled QPT adversaries in the CRS model under the LWE assumption and against entangled QPT adversaries in the QROM. Moreover, the required post-quantum commitment scheme and UC-secure zero-knowledge proof system can be instantiated in the respective models \cite{C:DFLSS09,EC:ABKK23}. We therefore obtain the following corollary (see also \Cref{rmk:local-global-setup}).

\begin{corollary}
Depending on the adversarial model and setup assumptions, the following instantiations exist:
\begin{itemize}
    \item In the CRS model, assuming QLWE, there exists an SUC-secure spatiotemporal MPC protocol against unentangled QPT adversaries.
    \item In QROM, there exists an SUC-secure spatiotemporal MPC protocol against entangled QPT adversaries.
\end{itemize}
\end{corollary}

\section{Spatiotemporal Multiparty Quantum Computation}
In this section, we extend our construction to quantum functionalities. The main difference from the classical setting is that each party's input is a quantum--classical pair $(\reg{x_i},\msf{st}_i)$, where $\reg{x_i}$ is a quantum register and $\msf{st}_i$ remains classical. Accordingly, the desired functionality $\mcal{Q}$ may perform quantum computation on the registers $\reg{x_1},\ldots,\reg{x_n}$ and classical computation on the trajectories $\msf{st}_1,\ldots,\msf{st}_n$. We therefore employ an MPQC protocol capable of evaluating such a hybrid quantum--classical functionality.

Because the trajectory components remain classical throughout the execution, their physical validity and commitment consistency can be enforced using the same approach as in \Cref{UC-ST-MPC}. Specifically, each party $P_i$ commits to its classical trajectory $\msf{st}_i$, and the MPQC protocol evaluates a modified functionality $\mcal{Q}'$. This functionality first verifies that each supplied trajectory, together with its corresponding decommitment, correctly opens the associated commitment. If all commitment checks succeed, it evaluates
\[
    \mcal{Q}
    \bigl(
        (\reg{x_1},\msf{st}_1),\ldots,
        (\reg{x_n},\msf{st}_n)
    \bigr);
\]
otherwise, it outputs $\bot$. Concurrently, the \textsf{TCnP} protocol verifies the physical validity of each classical trajectory and proves that the physically verified trajectory is consistent with the trajectory bound to the corresponding commitment. This is performed in exactly the same manner as in \Cref{UC-ST-MPC}.

Thus, the MPQC protocol performs the desired computation, whereas \textsf{TCnP} provides physical verification of the classical trajectory components. Since the trajectories remain classical, the commitment-consistency mechanism from the classical construction applies without modification.

\para{On the SUC framework for MPQC.} Technically speaking, the SUC framework we defined in \Cref{def:uc-smpc} is applicable to classical functionalities. However, the framework can be adapted to quantum functionalities with only minor modifications. In particular, the real-world execution, ideal process, and notion of SUC realization are defined as in \Cref{def:uc-smpc}, except that each party's input register $\reg{x_i}$ may now contain a quantum state, while its trajectory $\msf{st}_i$ remains classical. Accordingly, the ideal functionality may receive, process, and output quantum states. Since all other definitions and components remain unchanged from the classical framework, we do not restate them here and throughout use the notion of SUC realization given in \Cref{def:uc-smpc}.

\section{Acknowledgements and AI Disclosure}

\para{Acknowledgements.} Fuyuki Kitagawa and Xiao Liang would like to thank Tomoyuki Morimae, the organizer of \emph{the 2nd Kyoto Quantum Crypto Workshop} \cite{YITPWorkshopKyotoQuantumCrypto2}. The idea for this work was first conceived and discussed during their participation in the workshop.

\para{AI Disclosure.} Large Language Models  were used solely to polish the language and improve the clarity of certain parts of the manuscript. They were not used to generate paragraphs or to develop any of the ideas, techniques, or proofs presented in this work.

\clearpage

{
	\phantomsection
	\addcontentsline{toc}{section}{References}
}
\newcommand{\etalchar}[1]{$^{#1}$}

\end{document}